\documentclass[12pt, conference, final,
onecolumn, technote,a4paper]{IEEEtran}
\usepackage{latexsym}
\usepackage{todonotes}
\usepackage[numbers,square,comma,compress]{natbib}
\usepackage{mathrsfs}
\usepackage[normalem]{ulem}
\usepackage{url}
\usepackage{amsfonts,amssymb,amsmath,amsthm,bm}
\usepackage{xspace}
\usepackage{xcolor}
\usepackage{hyperref}
\usepackage{cleveref}
\newtheorem{theorem}{Theorem}
\newtheorem{assumption}{Assumption}

\newtheorem{definition}[theorem]{Definition}
\newtheorem{remark}[theorem]{Remark}
\newtheorem{corollary}{Corollary}
\newtheorem{lemma}[theorem]{Lemma}

\newtheorem{construction}{Construction}
\newtheorem{claim}[theorem]{Claim}
\newcommand{\cproblem}[1]{\ensuremath{\mathsf{#1}}\xspace}
\newcommand{\LWE}{\cproblem{LWE}}

\newcommand{\SIS}{\cproblem{SIS}}

\newcommand{\rand}{\xleftarrow{\$}}

\newcommand{\ABE}{\cproblem{ABE}}
\newcommand{\FE}{\cproblem{FE}}
\newcommand{\IPFE}{\cproblem{IPFE}}
\newcommand{\maabe}{\cproblem{MA}\text{-}\cproblem{ABE}}
\newcommand{\maipfe}{\cproblem{MA}\text{-}\cproblem{ABIPFE}}
\newcommand{\EVIPFE}{\cproblem{evIPFE}}
\newcommand{\flipLWE}{\cproblem{FlipLWE}}
\newcommand{\INDIPFE}{\cproblem{IND}\text{-}\cproblem{evIPFE}}
\newcommand{\NIPFE}{\cproblem{NIPFE}}

\newcommand{\DNF}{\cproblem{DNF}}
\newcommand{\ABIPFE}{\cproblem{AB}\text{-}\cproblem{IPFE}}
\newcommand{\LSSS}{\cproblem{LSSS}}
\newcommand{\MAABIPFE}{\cproblem{MA}\text{-}\cproblem{ABIPFE}}
\newcommand{\manipfe}{\cproblem{MA}\text{-}\cproblem{ABNIPFE}}
\newcommand{\mannipfe}{\mathsf{MA}\text{-}\mathsf{AB(N)IPFE}}
\newcommand{\maev}{\cproblem{MA}\text{-}\cproblem{ABevIPFE}}
\newcommand{\IND}{\cproblem{IND}}

\newcommand{\globalset}{\cproblem{GlobalSetup}}
\newcommand{\authset}{\cproblem{AuthSetup}}
\newcommand{\gp}{\cproblem{gp}}
\newcommand{\pk}{\cproblem{pk}}
\newcommand{\msk}{\cproblem{msk}}
\newcommand{\mpk}{\cproblem{mpk}}
\newcommand{\keygen}{\cproblem{Keygen}}
\newcommand{\sk}{\cproblem{sk}}
\newcommand{\gid}{\cproblem{gid}}
\newcommand{\aid}{\cproblem{aid}}
\newcommand{\enc}{\cproblem{Enc}}
\newcommand{\ct}{\cproblem{ct}}
\newcommand{\dec}{\cproblem{Dec}}
\renewcommand{\H}{\cproblem{H}}

\renewcommand{\P}{\mathcal{P}}
\newcommand{\AU}{\mathcal{AU}}
\newcommand{\GID}{\mathcal{GID}}
\newcommand{\A}{\mathcal{A}}
\newcommand{\pr}{\mathrm{Pr}}
\newcommand{\SD}{\mathrm{SD}}
\newcommand{\trapgen}{\cproblem{TrapGen}}
\newcommand{\samplepre}{\cproblem{SamplePre}}
\newcommand{\td}{\cproblem{td}}
\newcommand{\samp}{\cproblem{Samp}}
\newcommand{\aux}{\cproblem{aux}}
\newcommand{\pub}{\cproblem{pub}}
\newcommand{\pri}{\cproblem{pri}}
\newcommand{\Adv}{\cproblem{Adv}}
\newcommand{\pre}{\cproblem{Pre}}

\newcommand{\post}{\cproblem{Post}}

\newcommand{\dbdh}{\cproblem{DBDH}}

\newcommand{\superpoly}{\mathrm{superpoly}}
\newcommand{\partset}{\cproblem{par}}

\newcommand{\grow}{\lceil \log q\rceil}
\newcommand{\chal}{\cproblem{chal}}

\newcommand{\midd}{\cproblem{mid}}

\newcommand{\iitem}[1]{\bigskip\noindent\textbf{#1.}\ }
\newcommand{\dotitem}[1]{\noindent\textbf{#1:}}

\newcommand{\evlwe}{\cproblem{evLWE}}

\DeclareMathOperator{\poly}{poly}
\DeclareMathOperator{\negl}{negl}

\newcommand{\matA}{\ensuremath{\mathbf{A}}}
\newcommand{\matB}{\ensuremath{\mathbf{B}}}
\newcommand{\matC}{\ensuremath{\mathbf{C}}}
\newcommand{\matD}{\ensuremath{\mathbf{D}}}

\newcommand{\matG}{\ensuremath{\mathbf{G}}}

\newcommand{\matI}{\ensuremath{\mathbf{I}}}

\newcommand{\matK}{\ensuremath{\mathbf{K}}}

\newcommand{\matP}{\ensuremath{\mathbf{P}}}
\newcommand{\matQ}{\ensuremath{\mathbf{Q}}}

\newcommand{\matU}{\ensuremath{\mathbf{U}}}
\newcommand{\matV}{\ensuremath{\mathbf{V}}}

\newcommand{\matX}{\ensuremath{\mathbf{X}}}

\newcommand{\vecc}{\ensuremath{\mathbf{c}}}
\newcommand{\vecd}{\ensuremath{\mathbf{d}}}
\newcommand{\vece}{\ensuremath{\mathbf{e}}}

\newcommand{\vecg}{\ensuremath{\mathbf{g}}}

\newcommand{\veck}{\ensuremath{\mathbf{k}}}

\newcommand{\vecr}{\ensuremath{\mathbf{r}}}
\newcommand{\vecs}{\ensuremath{\mathbf{s}}}
\newcommand{\vect}{\ensuremath{\mathbf{t}}}
\newcommand{\vecu}{\ensuremath{\mathbf{u}}}
\newcommand{\vecv}{\ensuremath{\mathbf{v}}}

\newcommand{\vecx}{\ensuremath{\mathbf{x}}}
\newcommand{\vecy}{\ensuremath{\mathbf{y}}}
\newcommand{\vecz}{\ensuremath{\mathbf{z}}}
\newcommand{\veczero}{\ensuremath{\mathbf{0}}}

\newcommand{\N}{\ensuremath{\mathbb{N}}}

\newcommand{\R}{\ensuremath{\mathbb{R}}}
\newcommand{\T}{\cproblem{T}}
\newcommand{\Z}{\ensuremath{\mathbb{Z}}}

\newcommand{\vecdelta}{\ensuremath{\bm{\delta}}}

\begin{document}

\title{Decentralized Multi-Authority Attribute-Based Inner-Product Functional Encryption: Noisy and Evasive Constructions from Lattices$^\dag$}
\author{
  Jiaqi Liu,
  Yan Wang, 
  Fang-Wei Fu\thanks{Jiaqi Liu, Yan Wang and Fang-Wei Fu are with Chern Institute of Mathematics and LPMC, Nankai University, Tianjin 300071, P. R. China, Emails: ljqi@mail.nankai.edu.cn, yan.wang@mail.nankai.edu.cn, fwfu@nankai.edu.cn.}
  \thanks{$^\dag$This research is supported by the National Key Research and Development Program of China (Grant No. 2022YFA1005000), the National Natural Science Foundation of China (Grant Nos. 12141108 and 62371259), the Fundamental Research Funds for the Central Universities of China (Nankai University), and the Nankai Zhide Foundation.}
\thanks{manuscript submitted \today}}

\maketitle

\begin{abstract}
  We initiate the study of multi-authority attribute-based functional encryption for \emph{noisy inner-product functionality}, and propose two new primitives: (1) multi-authority attribute-based (noisy) inner-product functional encryption ($\mannipfe$), and (2) multi-authority attribute-based evasive inner-product functional encryption ($\maev$). The $\mannipfe$ primitive generalizes the existing multi-authority attribute-based inner-product functional encryption schemes by Agrawal et al.~\cite{AGT21}, by enabling \emph{approximate} inner-product computation under decentralized attribute-based control. This newly proposed notion combines the approximate function evaluation of noisy inner-product functional encryption (\IPFE) with the decentralized key-distribution structure of multi-authority attribute-based encryption. To better capture noisy functionalities within a flexible security framework, we formulate the $\maev$ primitive under a generic-model view, inspired by the evasive $\IPFE$ framework by Hsieh et al.~\cite{HLL24}. It shifts the focus from pairwise ciphertext indistinguishability to a more relaxed pseudorandomness-based game.

  To support the above notions, we introduce two variants of lattice-based computational assumptions: \begin{itemize}
      \item The \emph{evasive $\IPFE$ assumption} ($\EVIPFE$): it generalizes the assumption introduced in~\cite{HLL24} to the multi-authority setting and admits a reduction from the evasive $\LWE$ assumption proposed by Waters et al.~\cite{WWW22};

      \item The \emph{indistinguishability-based evasive $\IPFE$ assumption} ($\INDIPFE$): it is an indistinguishability-based variant of the evasive $\IPFE$ assumption designed to capture the stronger security guarantees required by our $\mannipfe$ scheme.
  \end{itemize}

  We present concrete lattice-based constructions for both primitives supporting subset policies, building upon the framework of \cite{WWW22}. Our schemes are proven to be statically secure in the random oracle model under the standard $\LWE$ assumption and the newly introduced assumptions. Additionally, we demonstrate that our $\mannipfe$ scheme can be transformed, via standard modulus switching, into a \emph{noiseless} $\maipfe$ scheme that supports exact inner-product functionality consistent with the $\maipfe$ syntax in \cite{AGT21,DP23}. This yields the first lattice-based construction of such a primitive. All our schemes support arbitrary polynomial-size attribute policies and are secure in the random oracle model under lattice assumptions with a sub-exponential modulus-to-noise ratio, making them practical candidates for noise-tolerant, fine-grained access control in multi-authority settings.
\end{abstract}

\newpage
\tableofcontents

\newpage
\section{Introduction}\label{intro}
\noindent \emph{Functional Encryption} (\FE), introduced by Boneh, Sahai, and Waters \cite{BSW11}, is a versatile cryptographic paradigm that extends the capabilities of traditional public-key encryption schemes. It facilitates fine-grained access control over encrypted data, enabling authorized users to compute specific functions of the plaintext without recovering the entire plaintext. More formally, in an \FE scheme, each secret key $\sk_f$ is associated with a particular function $f$. Given an encryption $\enc(\mpk, \vecx)$ of a message $\vecx$ encrypted under the master public key $\mpk$, decrypting with $\sk_f$ reveals only the value $f(\vecx)$ without leaking any additional information about the message $\vecx$. The inherent property of selective disclosure in \FE makes it particularly valuable in applications involving sensitive or confidential data.
For example, in healthcare settings, \FE allows researchers to aggregate statistics from encrypted patient records while maintaining the confidentiality of individual data entries.

The standard security framework for \FE is \emph{indistinguishability-based security} (\IND). In this model, an adversary attempts to distinguish between the encryptions of two selected messages $\vecx_0$ and $\vecx_1$, while being allowed to query secret keys $\sk_f$ for functions $f$ such that $f(\vecx_0)=f(\vecx_1)$. This restriction ensures that the adversary gains no additional information beyond the outputs of permitted functions. The security guarantee requires that the adversary remain unable to distinguish between the ciphertexts, even when given access to multiple such keys.

\iitem{Inner-Product Functional Encryption} A notable subclass of \FE schemes designed for computing linear functions is known as \emph{Inner-Product Functional Encryption} (\IPFE), which has become an active area of research over the past decade. The study of \IPFE began with the work of Abdalla et al.~\cite{ABDP15} and has since been extensively studied in a series of works~\cite{ABDP16,ALS16,DDM16,AGRW17,BBL17,CSGPP18,ACFGU18,LT19,ABKW19,Tom19,ABM20}. In an \IPFE scheme, a ciphertext $\ct_{\vecu}$ encrypts an $\ell$-dimensional vector $\vecu\in \mathcal{R}^{\ell}$ over some ring $\mathcal{R}$, while a secret key $\sk_{\vecv}$ for a vector $\vecv\in \mathcal{R}^{\ell}$ enables the computation of the inner-product function $f_{\vecv}(\cdot)=\langle\cdot,\vecv\rangle$. Decrypting $\ct_{\vecu}$ using $\sk_{\vecv}$ yields the value $\langle\vecu,\vecv\rangle$, without leaking any other information about the message $\vecu$. \IPFE provides a powerful cryptographic tool for selective computation over encrypted data, ensuring that sensitive information remains protected. 
Notable applications include secure data analysis, privacy-preserving machine learning, and privacy-enhanced database queries. In addition, \IPFE serves as a foundational building block for more advanced cryptographic primitives, such as \FE for quadratic functions~\cite{JLS19,Gay20} and \emph{attribute-based encryption} (\ABE)~\cite{WFL19,HLL24}. These extensions further enhance the versatility and applicability of functional encryption.

Building on \IPFE, Agrawal~\cite{Agr19,AP20} introduced \emph{Noisy Inner-Product Functional Encryption} (\NIPFE), an extension that incorporates noise into the computation process. Unlike the standard \IPFE scheme in \cite{ABDP15}, which supports the \emph{exact} evaluation of inner products, \NIPFE enables the computation of \emph{approximate} inner products with an additive noise term of the form $\langle\vecu,\vecv\rangle + e$, where $e$ is a small-norm noise term. This relaxation allows decryption to return an approximate result, which improves flexibility in noisy settings. The security guarantee of \NIPFE ensures that for any two plaintext vectors $\vecu$ and $\vecu'$, and key vectors $\vecv_1, \dots, \vecv_m$, if $\langle \vecu,\vecv_j\rangle\approx\langle \vecu',\vecv_j\rangle$ holds for each $j$, the ciphertexts of $\vecu$ and $\vecu'$ remain computationally indistinguishable, even when the adversary is given the secret keys corresponding to $\vecv_j$. The security property is achieved by leveraging the noise term $e$, which effectively ``smudges'' small differences in the inner-product computations, thereby preventing the adversary from extracting meaningful distinguishing information. \NIPFE provides a security guarantee in applications in which a certain level of noise tolerance is acceptable, balancing between utility and privacy.

\iitem {Attribute-Based Inner-Product Functional Encryption}
\emph{Attribute-Based Inner-Product Functional Encryption} ($\ABIPFE$), introduced by Abdalla et al.~\cite{ACGU20}, is an advanced cryptographic primitive that combines the access control capabilities of \ABE with the inner-product computation functionality of \IPFE. This hybrid notion enables fine-grained access control over encrypted data while supporting privacy-preserving computations on ciphertexts. In a \emph{ciphertext-policy} $\ABIPFE$ scheme, each ciphertext is associated with a set of attributes, and each secret key corresponds to an access policy defined over the attribute universe, following a structure similar to that of a ciphertext-policy \ABE. In the dual setting, namely the \emph{key-policy} $\ABIPFE$ scheme, ciphertexts are associated with access policies, and secret keys are associated with attribute sets. In either case, decryption is allowed only when the attributes satisfy the access policy. Unlike traditional \IPFE, $\ABIPFE$ allows inner-product computations on encrypted data only when decryption is authorized, making it particularly useful for applications such as privacy-preserving machine learning, secure data analytics, and encrypted search. By integrating the strengths of \ABE and \IPFE, $\ABIPFE$ offers a powerful tool for secure data sharing and processing in multi-user environments, with practical relevance to applications such as cloud computing, the Internet of Things (IoT), and healthcare systems.

\iitem{Multi-Authority Attribute-Based Inner-Product Functional Encryption}
Most prior work on $\ABE$ or $\ABIPFE$ focuses on the single-authority setting, where a central trusted authority is the only entity responsible for validating user attributes and issuing the corresponding secret keys. However, in many practical scenarios, it is more natural for multiple separate authorities to independently manage and authorize different subsets of attributes for users. \emph{Multi-Authority} \ABE ($\maabe$)~\cite{LW11,RW15,DKW21a,DKW21b,WWW22} addresses this limitation by decentralizing the authority structure in an \ABE system, allowing multiple independent authorities to operate and manage attribute-based access control. This decentralization improves the scalability and flexibility of the system. Motivated by the need for secure and privacy-preserving computation in decentralized environments, recent efforts have explored the integration of $\maabe$ and \IPFE into a unified framework known as \emph{Multi-Authority Attribute-Based Inner-Product Functional Encryption} ($\MAABIPFE$). Agrawal et al.~\cite{AGT21} first highlighted the potential of combining multi-authority systems with functional encryption by extending the attribute-based component of $\ABIPFE$ to a multi-authority setting. $\MAABIPFE$ is a significant cryptographic framework that unifies the decentralized access control of $\maabe$ and the computation functionality of \IPFE. 

Informally, in an $\MAABIPFE$ scheme, each authority generates its own master secret key and is responsible for issuing secret keys associated with attributes it governs. Let $\ct_{f}(\vecu)$ be a ciphertext encrypting a plaintext vector $\vecu\in\mathcal{R}^{\ell}$ under an access policy $P$. A user holding secret keys corresponding to a vector $\vecv\in\mathcal{R}^{\ell}$ and a set of attributes $A$ can recover the inner product $\langle \vecu, \vecv \rangle$ if and only if $A$ satisfies the policy $P$ (i.e., $A$ is \emph{authorized}); otherwise, the ciphertext reveals no information about $\vecu$.

\subsection{Related Works} 
The $\MAABIPFE$ scheme proposed in \cite{AGT21} lifts the $\ABIPFE$ construction in \cite{ACGU20} to the multi-authority setting. It is known as the first nontrivial multi-authority \FE scheme beyond $\maabe$. It supports access policies that can be realized by \emph{Linear Secret Sharing Schemes} (\LSSS) and is constructed using pairing-based techniques. Its security relies on variants of the subgroup decision assumptions over composite-order bilinear groups introduced in \cite{BSW13}. 

Subsequently, Datta and Pal~\cite{DP23} proposed two $\maipfe$ schemes that also support \LSSS-based access structures. These schemes are built in the more efficient prime-order bilinear group setting and rely on the well-studied \emph{Decisional Bilinear Diffie-Hellman Assumption} (\dbdh) and its variants. Their constructions are proven secure in the Random Oracle Model. By moving to the prime-order group setting, their constructions achieve notable improvements in efficiency over the scheme of \cite{AGT21}. The construction supports computations over vectors of a priori unbounded length and an unbounded number of authorities. Furthermore, the schemes in \cite{DP23} overcome the ``one-use'' restriction, allowing attributes to appear arbitrarily many times within access policies, thereby providing greater flexibility and versatility in policy design.

\iitem{Evasive \LWE Assumption}
\emph{Evasive Learning with Errors} (Evasive \LWE)~\cite{Wee22} is a non-standard variant of the \LWE assumption. Roughly, the assumption states that for any efficient sampling algorithm $\samp$ that outputs a matrix $\matQ\in \Z_q^{n\times t}$,
\begin{align}\label{precon}
&\text{if}\quad (\matA,\matQ,\vecs^\top\matA+\vece_1^\top,\vecs^\top\matQ+\vece_2^\top,\aux)\approx (\matA,\matQ,\$_1,\$_2,\aux),\\
&\text{then}\quad 
(\matA,\matQ,\vecs^\top\matA+\vece_1^\top,\matA^{-1}(\matQ),\aux)\approx (\matA,\matQ,\$_1,\matA^{-1}(\matQ),\aux),\nonumber
\end{align}
for uniformly random matrix $\matA\rand\Z_q^{n\times m}$, uniformly random \LWE secret $\vecs\rand\Z_q^n$, Gaussian noise vectors $\vece_1\in \Z_q^m,\vece_2\in\Z_q^t$, and uniformly random vectors $\$_1\in\Z_q^m,\$_2\in\Z_q^t$. Here, $\matA^{-1}(\matQ)$ denotes a low-norm (typically Gaussian) preimage of $\matQ$ with respect to $\matA$, i.e., a matrix with low-norm entries such that $\matA\cdot (\matA^{-1}(\matQ))=\matQ$. In \cite{Wee22}, the assumption is formulated for \emph{public-coin} sampling algorithms $\samp$, meaning that the auxiliary string $\aux$ contains all the tossed coins (randomness) used by $\samp$. 

Intuitively, the evasive \LWE assumption essentially asserts that the only meaningful way to exploit the preimage $\matA^{-1}(\matQ)$ is to multiply it by the $\LWE$ sample $\vecs^\top\matA+\vece_1^\top$,  obtaining $$(\vecs^\top\matA+\vece_1^\top)\cdot(\matA^{-1}\matQ)\approx \vecs^\top\matQ,$$ and then attempting to distinguish this value from uniform samples. However, the precondition \eqref{precon} in the assumption guarantees that this advantage remains negligible, thereby preventing zeroizing attacks such as those discussed in \cite{CHLR15,CVW18,HJL21,JLLS23}.

In follow-up work, \cite{WWW22} proposed a variant of public-coin evasive \LWE assumption involving multiple matrix pairs $(\matA_i,\matQ_i)_i$, along with their respective \LWE samples $(\vecs_i^\top\matA_i+\vece_{1,i}^\top,\vecs_i^\top\matQ_i+\vece_{2,i}^\top)_i$. Roughly, the evasive \LWE assumption in \cite{WWW22} states that if \begin{align*}
&\text{if}\quad (\{\matA_i,\vecs_i^\top\matA_i+\vece_{1,i}^\top\}_i,\{\vecs_i^\top\matQ_i+\vece_{2,i}^\top\}_{i},\aux)   \approx  (\{\matA_i,\$_{1,i}\}_i,\{\$_{2,i}\}_{i},\aux),\\
&\text{then}\quad (\{\matA_i,\vecs_i^\top\matA_i+\vece_{1,i}^\top\}_i,\{\matA_i^{-1}(\matQ_i)\}_{i},\aux)   \approx  (\{\matA_i,\$_{1,i}\}_i,\{\matA_i^{-1}(\matQ_i)\}_{i},\aux),
\end{align*}
where $\vecs^\top=[\vecs_1^\top\mid\cdots\mid\vecs_\ell^\top]\rand\Z_q^{n\ell}$ is an $\LWE$ secret.

\iitem{Evasive \IPFE} \emph{Evasive Inner-Product Functional Encryption} (Evasive $\IPFE$), introduced by Yao-Ching Hsieh, Huijia Lin, and Ji Luo in \cite{HLL24}, extends the standard \IPFE framework~\cite{ABDP15}. Similar to the \NIPFE scheme~\cite{Agr19}, evasive $\IPFE$ encrypts an $\ell$-dimensional vector $\vecu\in \mathcal{R}^\ell$ over some ring $\mathcal{R}$ into a ciphertext $\ct_\vecu$, and allows the generation of secret keys $\{\sk_{\vecv_j}\}_j$ for vectors $\{
\vecv_j\}_j$. Decryption using $\sk_{\vecv_j}$ yields noisy inner products of the form  $\langle\vecu, \vecv_j\rangle + e_j$, where each $e_j$ is a small-norm noise term, while revealing no additional information about $\vecu$. The key distinction lies in the underlying security model, which reflects a \emph{generic-model view} analogous to that of the evasive \LWE assumption. Specifically,
\begin{align*}
    &\text{ if}\quad \{\vecv_j,\langle\vecu,\vecv_j\rangle+e_j\}_j\approx \{\vecv_j,\$_j\}_j,\\
    &\text{ then} \quad (\ct_{\vecu},\{\vecv_j,\sk_{\vecv_j}\}_j)\approx (\ct_{\$},\{\vecv_j,\sk_{\vecv_j}\}_j), 
\end{align*} where each $e_j$ is small-norm noise term, and each $\$_j$ and $\$$ is a uniformly random elements over the appropriate range. In other words, if the noisy inner products with $\vecu$ are pseudorandom, then the ciphertext encrypting $\vecu$ is also pseudorandom, even given the key-vector pairs $\{\vecv_j,\sk_{\vecv_j}\}_j$.

\subsection{Our Results}
In this work, we make the following conceptual and technical contributions:
We begin by proposing two new notions for attribute-based functional encryption in the multi-authority setting (Section~\ref{sec:5}).
\begin{itemize}
    \item \textbf{Multi-authority attribute-based (noisy) inner-product functional encryption ($\mannipfe$)}: We propose the notion of $\manipfe$, a natural generalization of the $\maipfe$ scheme~\cite{AGT21,DP23} by allowing the \emph{approximate} computation of inner products rather than \emph{exact} ones when decrypting using authorized secret keys. The $\mannipfe$ primitive integrates two core features: (i) the approximate inner-product functionality of $\NIPFE$, and (ii) the decentralized access control mechanism of $\maabe$, where secret keys are issued by multiple authorities based on their attribute sets. This combination enables fine-grained, decentralized access control over approximate inner-product computations in a multi-authority setting.

    We then formalize the security of the $\mannipfe$ scheme in the static setting, following the framework of \cite{RW15,WWW22}. In this model (also referred to as the \emph{selective} model)
 the adversary must commit to its challenge plaintexts, secret-key queries, and authority corruption choices immediately after the global parameters are initialized. Unlike the standard $\maabe$ setting, where secret-key queries must be unauthorized to decrypt the challenge ciphertext, our static security model permits the adversary to request even \emph{authorized} secret keys. However, such secret keys yield sufficiently \emph{close} inner products when evaluated over the two challenge plaintexts. This relaxed notion is tailored to the noisy functional nature of our scheme. It ensures that no efficient adversary can distinguish between the encryptions of two plaintexts unless it can produce keys that yield meaningfully different decryption results.

 \item \textbf{Multi-authority attribute-based evasive inner-product functional encryption ($\maev$)}: We further introduce the notion of $\maev$, which serves as a relaxed variant of $\mannipfe$ in terms of its underlying security definition. While both primitives enable approximate inner-product computation in a multi-authority setting, the $\maev$ scheme adopts a ``generic-model view'' inspired by recent works on evasive $\IPFE$ in \cite{HLL24}. Instead of requiring the adversary to distinguish between two challenge ciphertexts, the security model is formulated via a pseudorandomness-based game: the goal of the adversary is to distinguish the ciphertext of a structured plaintext generated by a public-coin sampler from that of a uniformly random plaintext. 

The main difference between the adversarial capabilities in the $\mannipfe$ and $\maev$ games lies in the secret-key queries. For unauthorized queries, both models handle them identically. However, for authorized queries, the adversary in $\maev$ is more restricted. Specifically, it does not have full control over both the attribute set and the key vector $\vecv$ that together define the secret-key queries. Instead, the adversary specifies only the attribute set, while the key vector $\vecv$ is sampled by a public-coin sampler $\samp_{\vecv}$. This restriction mirrors the structure of evasive $\IPFE$ games. In addition, the challenge plaintext $\vecu$ is sampled by another public sampling algorithm $\samp_{\vecu}$ with private randomness (i.e., the adversary does not have access to $\vecu$). The security definition guarantees that if the sampler pair $(\samp_{\vecv},\samp_{\vecu})$ produces noisy pseudorandom inner products, i.e., $$(\vecr_{\pub},\langle\vecu,\vecv\rangle+e)\approx(\vecr_{\pub},\$),$$ then no efficient adversary can distinguish the ciphertext of $\vecu$ and that of a uniformly random plaintext with non-negligible advantage. Here $\vecr_{\pub}$ denotes the randomness used by $\samp_{\vecv}$, $e$ is a noise term, and $\$$ is a uniformly random element.
\end{itemize}

To instantiate the above notions, we provide three concrete lattice-based constructions, supporting approximate inner-product computation under subset policies  with decentralized access control, offering different trade-offs between noise management and security guarantees.

\iitem{Construction of $\maev$} We construct an $\maev$ scheme for subset policies, following the framework of \cite{WWW22}. Our construction is proven to be statically secure in the random oracle model under the standard \LWE assumption and a new evasive $\IPFE$ assumption ($\EVIPFE$) introduced in this work, which generalizes the evasive $\IPFE$ assumption proposed in~\cite{HLL24}. The idea of evasive $\IPFE$ assumption is conceptually motivated by the evasive $\LWE$ assumption and admits a reduction from the evasive $\LWE$ assumption from~\cite{WWW22} (cf. Appendix~\ref{sec:app}). Full details of the construction, security proof, and parameter selection are provided in Section~\ref{sec:6}.

\begin{theorem}[Informal]
    Suppose that the $\LWE$ assumption and the $\evlwe$ assumption (cf. Section~\ref{Evasiveipfe}) hold with sub-exponential modulus-to-noise ratio, then there exists a statically secure $\maev$ scheme for subset policies of arbitrary polynomial size in the random oracle model. 
\end{theorem}

\iitem{Construction of $\manipfe$} We then construct an $\manipfe$ scheme for subset policies. This construction follows the same syntax as the $\maev$ scheme but achieves a stronger notion of security under different lattice-based assumptions. Specifically, it is proven to be statically secure in the random oracle model under the $\LWE$ assumption and the new $\INDIPFE$ assumption (cf. Section~\ref{Evasiveipfe}) introduced in this work. The $\INDIPFE$ assumption is an indistinguishability-based variant of the $\EVIPFE$ discussed earlier. Full details of the construction, security analysis, and parameter selection are provided in Section~\ref{sec:7}.

\begin{theorem}[Informal]
    Suppose that the $\LWE$ assumption and the $\INDIPFE$ assumption (cf. Section~\ref{Evasiveipfe}) hold with sub-exponential modulus-to-noise ratio, then there exists a statically secure $\manipfe$ scheme for subset policies of arbitrary polynomial size in the random oracle model. 
\end{theorem}

\iitem{Construction of Noiseless $\maipfe$} Finally, we modify our $\manipfe$ scheme into a noiseless $\maipfe$ scheme, which enables \emph{exact} computation of the inner products upon successful decryption.
This construction aligns with the definition of $\maipfe$ from \cite{AGT21,DP23}. The construction proposed in this work is the first such scheme based on lattice-related assumptions. Our modification relies on a standard modulus switching technique that removes the noise introduced during decryption. The resulting scheme retains a similar overall structure but achieves exact correctness. The scheme is proven to be statically secure in the random oracle model, under the $\LWE$ assumption and the $\INDIPFE$ assumption. Full details of the construction, security analysis, and parameter selection are provided in Section~\ref{sec:8}.

\begin{theorem}[Informal]
   Suppose that the $\LWE$ assumption and the $\INDIPFE$ assumption (cf. Section~\ref{Evasiveipfe}) hold with sub-exponential modulus-to-noise ratio, then there exists a statically secure \emph{noiseless} $\maipfe$ scheme for subset policies of arbitrary polynomial size in the random oracle model. 
\end{theorem}

\subsection{Paper Organization}
The paper is organized as follows. Section~\ref{technique} provides a technical overview of our main constructions and ideas. Section~\ref{sec:2} introduces necessary notations, basic concepts of lattices, the \LWE assumption, and its variants. Section~\ref{Evasiveipfe}, presents two variants of the Evasive \IPFE assumptions, building on the framework of~\cite{HLL24}. In Section~\ref{sec:5}, we formalize monotone access structures and give the definitions of multi-authority attribute-based evasive inner-product functional encryption ($\maev$) and multi-authority attribute-based (noisy) inner-product functional encryption ($\mannipfe$). Section~\ref{sec:6} presents our construction of an $\maev$ scheme. Section~\ref{sec:7} introduces a construction of an $\manipfe$ scheme. Section~\ref{sec:8} describes a construction of a noiseless $\MAABIPFE$ scheme using the modulus-switching technique. Sections~\ref{sec:6} through~\ref{sec:8} include detailed analyses of correctness, security, and parameter selection.

\section{Technical Overviews}\label{technique}
In this section, we demonstrate a high-level overview of our construction of multi-authority attribute-based (noisy/evasive) inner-product functional encryption with subset policies in the random oracle model. Our construction is inspired by the frameworks of \cite{WWW22} and \cite{HLL24}.

\iitem{Preliminaries}
We introduce some notations used throughout this section. To simplify expressions involving noise, we adopt the wavy underline to indicate that a term is perturbed by some noise term. For example, we use the notation $\uwave{\vecs^\top\matA}$ to denote the term $\vecs^\top\matA+\vece$ for an unspecified noise vector $\vece$. For a discrete set $S$, we denote $x\rand S$ as $\vecx$ is uniformly sampled from the set $S$. The term $\matA^{-1}(\vecy)$ means sampling a short preimage $\vecx$ such that $\matA\vecx=\vecy$. For a positive integer $n\in \N$, denote $[n]$ as the set $\{k\in\N:1\leq k\leq n\}$. Let $\$$ denote a uniformly random vector of appropriate dimension. Let $\matG_n=\matI_n\otimes\vecg^{\top}\in \Z_q^{n\times n\lceil\log q\rceil}$ denote the gadget matrix, where $\vecg^{\top}=(1,2,\ldots,2^{\lceil\log q\rceil-1})$. The inverse function $\matG_n^{-1}$ maps the vectors in $\Z_q^n$ to their binary expansions in $\{0,1\}^{n\lceil\log q\rceil}$. We will omit the subscript when the dimension $n$ is clear. We use $\overset{c}{\approx}$ as the abbreviation for computationally indistinguishable.

\subsection{Our Schemes}\label{scheme}
 Let $[L]$ denote the index set of all authorities. In the subset policy setting, each ciphertext is associated with a subset $X\subseteq [L]$, while each secret key is associated with a subset $Y\subseteq [L]$ and a key vector $\vecv\in\Z_q^n$. Decryption of a ciphertext encrypting $\vecu\in\Z_q^n$ associated with $X$ using the secret key corresponding to $Y$ and $\vecv\in \Z_q^n$ yields an approximate value of the Euclidean inner product $\langle \vecu,\vecv\rangle=\vecu^\top\vecv$ whenever $X\subseteq Y$.

We next provide an informal description of our construction. The scheme consists of the following components:
\begin{itemize}
    \item The master public key for the authority indexed by $i$ is given by $(\matA_i,\matB_i,\matP_i)\rand\Z_q^{n\times m}\times \Z_q^{n\times m'}\times \Z_q^{n\times m}$.
    
    \item The master secret key for the authority indexed by $i$ is the trapdoor $\td_{\matA_i}$ for the matrix $\matA_i$. This trapdoor allows efficient sampling of short preimages for $\matA$.
    
    \item Given a subset $Y\subseteq [L]$, a key vector $\vecv\in\Z_q^n$, and a global identifier $\gid$, each authority $i \in Y$ uses its trapdoor to generate the secret key component corresponding to $i$. The corresponding secret key is generated as $$\sk\leftarrow\{\matA_i^{-1}(\matP_i\matG^{-1}(\vecv)+\matB_i\H(\gid,\vecv))\}_{i\in Y}.$$ Here $\H$ is a hash function modeled as a random oracle that outputs vectors in $\Z_q^{m'}$ with small norm.
    
    \item To encrypt a message $\vecu\in\Z_q^n$ under an attribute set $X\subseteq [L]$, the ciphertext is generated as $$\ct\leftarrow\left(\{\uwave{\vecs_i^\top\matA_i}\}_{i\in X},\uwave{\sum_{i\in X}\vecs_i^\top\matB_i},\uwave{\sum_{i\in X}\vecs_i^\top\matP_i}+\vecu^\top\matG\right),$$ where $\vecs_i\rand\Z_q^n$ for each $i\in X$.

\item The decryption operation using $\sk$ is performed as follows. It first computes $\vecr\leftarrow \H(\gid,\vecv)$, and then computes $$-\sum_{i\in X}\uwave{\vecs_i^\top\matA_i}\cdot\matA_i^{-1}(\matP_i\matG^{-1}(\vecv)+\matB_i\vecr)+\uwave{\left(\sum_{i\in X}\vecs_i^\top\matB_i\right)}\cdot\vecr+\left(\uwave{\sum_{i\in X}\vecs_i^\top\matP_i}+\vecu^\top\matG\right)\cdot\matG^{-1}(\vecv)\approx\vecu^\top\vecv$$ whenever $Y\supseteq X$.
\end{itemize}

\iitem{The randomization of $\gid$ and $\vecv$}
To prevent collusion among users in the system who might attempt to combine their individual secret keys to decrypt ciphertexts, which contradicts the intended security requirements of our setting, we adopt the \emph{global identifier model} (\gid), which is a standard technique widely used in multi-authority cryptographic schemes. 

In the \gid model, each user is assigned a unique and efficiently verifiable global identifier during setup, which remains fixed once global parameters are established. Upon verifying a user's global identifier, the authority validates his access rights and issues secret keys only for attributes authorized for that user. When instantiated in the random oracle model, the hash value of the global identifier $\H(\gid)$ provides consistent randomness across secret keys issued by different authorities, thereby preventing unauthorized key combination. In our $\MAABIPFE$ setting, we further extend this model by additionally incorporating the key
vector $\vecv$ as part of the input to the hash function. This ensures that secret keys are randomized with respect to both the user's identity and their functional target, further mitigating potential collusion attacks. We elaborate on this approach and propose our constructions in Sections \ref{sec:6}$\sim$\ref{sec:8}.

\subsection{Static Security for Our Scheme}
In the following, we provide a high-level overview demonstrating that the proposed scheme is statically secure under certain lattice-based assumptions. The static security model of $\mannipfe$ follows a structure similar to that of $\maabe$. Specifically, the security game involves a set of corrupt authorities $\mathcal{C}\subseteq [L]$. The adversary is allowed to arbitrarily assign the public keys and master secret keys for these corrupt authorities. Let $X\subseteq [L]$ denote the subset of authorities associated with the challenge ciphertexts. The adversary's goal is to distinguish between the encryptions of two selected vectors $\vecu_0$ and $\vecu_1$ under the authority set $X$.

Within this setting, the adversary is also allowed to submit secret-key queries of the form $\{(\gid,Y,\vecv)\}$, where $Y$ is a subset of non-corrupt (honest) authorities, $\gid$ is a global user identifier, and $\vecv$ is a key vector. Each secret-key query must satisfy at least one of the two following conditions:
\begin{enumerate}
    \item [(a)] $(Y\cup\mathcal{C})\cap X\subsetneqq X$: This condition reflects the multi-authority aspect of the setting. The security requirement states that even if the adversary has access to secret keys with authorities that do not satisfy the policy (unauthorized keys), the ciphertexts remain indistinguishable, aligning with the static security notion in the $\maabe$ scheme.
    \item [(b)] $(\vecu_0-\vecu_1)^\top\vecv\approx 0$: This condition enforces the inner-product functional encryption scheme requirement. It guarantees that secret keys satisfying the policy reveal only an approximate value of the inner-product $\vecu_i^\top\vecv$ ($i=0,1$), and no additional information beyond that.
\end{enumerate}

\iitem{Weakened Static Security: $\maev$\ Scheme}
Inspired by the evasive $\IPFE$ framework introduced in \cite{HLL24}, we begin our analysis with a weaker notion of static security, under which we define the $\maev$ scheme. While sharing the same syntax as the $\mannipfe$, this primitive imposes weaker adversarial
control. In particular, the queries mentioned above are no longer entirely determined by the adversary. More precisely, the secret-key queries satisfying condition (b) are jointly generated by both the adversary and the challenger, while the challenge plaintext $\vecu_0$ is generated by a public-coin sampler and $\vecu_1$ is uniformly sampled. This setting reflects a ``generic-model view'' of the scheme, analogous to evasive $\LWE$~\cite{Wee22} and evasive $\IPFE$~\cite{HLL24}. The challenge plaintext $\vecu_0$ and the secret key queries $\{\sk_{\vecv_j}\}_j$ (associated with key vector $\vecv_j$) satisfying condition (b) are generated such that $$\{\vecv_j,\{\uwave{\vecu_0^\top\vecv_j}\}_j\}\overset{c}{\approx}\{\vecv_j,\{\$_j\}_j\}.$$
The security requirement in this model is that the adversary cannot distinguish between the ciphertext of $\vecu_0$ and that of uniformly generated $\vecu_1$, even when it has access to the corresponding secret keys of its queries.

\iitem{Our Assumption: Generalized Evasive $\IPFE$}
We prove the static security of our construction as an $\maev$ scheme under the standard $\LWE$ assumption and the $\EVIPFE$ assumption, which we formally introduce in Section~\ref{Evasiveipfe}. This assumption can be viewed as an extension of the evasive $\IPFE$ assumption proposed by~\cite{HLL24}. We begin by informally recalling the assumption from~\cite{HLL24}. Let $\vecv_1,\ldots,\vecv_k$ be vectors generated by a public-coin sampler with public randomness $\vecr_{\pub}$, and let $\vecu$ be generated by a private-coin sampler. Let $\matA,\matP$ be uniformly random matrices, and define $$\matQ:=\matP[\matG^{-1}(\vecv_1)\mid\cdots\mid\matG^{-1}(\vecv_k)].$$ The assumption states that if the following \emph{precondition} holds: \begin{align*}
    (\vecr_{\pub},\matA,\uwave{\vecs^\top\matA},\matP,\uwave{\vecs^\top\matP}+\vecu^\top\matG,\uwave{\vecs^\top\matQ})\overset{c}{\approx}(\vecr_{\pub},\matA,\$_1,\matP,\$_2',\$_3).
\end{align*}
then the following \emph{postcondition} also holds:
\begin{align*}
(\vecr_{\pub},\matA,\uwave{\vecs^\top\matA},\matP,\uwave{\vecs^\top\matP}+\vecu^\top\matG,\matK)\overset{c}{\approx}(\vecr_{\pub},\matA,\$_1,\matP,\$_2',\matK),
\end{align*}
where $\matK\leftarrow\matA^{-1}(\matQ)$ and $\vecr_{\pub}$ consists of all public randomness used by the samplers. The intuition behind this assumption is that the only way to use the matrix $\matK$ is via computing \begin{align*}
\uwave{\vecs^\top\matA}\matK\approx\uwave{\vecs^\top\matQ},
\end{align*} which is pseudorandom by the  precondition, and hence rules out potential zeroizing attacks.

We now generalize this assumption in two stages:
\begin{enumerate}
    \item \textbf{Introducing an auxiliary matrix $\matB$.} Let $\vecr_{\pub},\matA,\matP,\vecu,\vecv_1,\ldots,\vecv_k$ be the same as in the original assumption. Let $\matB$ be an additional uniformly random matrix, and let $\vecr_1,\cdots,\vecr_{k}$ be generated by a public-coin algorithm. Let $\matQ$ be modified as $$\matQ=[\matB\mid\matP]\left[\begin{array}{c|c|c}
            \vecr_{1} & \cdots &\vecr_{k}  \\
            \matG^{-1}(\vecv_{1}) & \cdots &\matG^{-1}(\vecv_{k})
        \end{array}\right].$$
    The corresponding modified \emph{precondition} becomes \begin{align*}
(\vecr_{\pub},\matA,\uwave{\vecs^\top\matA},\matB,\uwave{\vecs^\top\matB},\matP,\uwave{\vecs^\top\matP}+\vecu^\top\matG,\uwave{\vecs^\top\matQ})\overset{c}{\approx}(\vecr_{\pub},\matA,\$_1,\matB,\$_2,\matP,\$_2’,\$_3).
\end{align*}
 and the \emph{postcondition} is given as
\begin{align*}
(\vecr_{\pub},\matA,\uwave{\vecs^\top\matA},\matB,\uwave{\vecs^\top\matB},\matP,\uwave{\vecs^\top\matP}+\vecu^\top\matG,\matK)\overset{c}{\approx}(\vecr_{\pub},\matA,\$_1,\matB,\$_2,\matP,\$_2',\matK),
\end{align*}where $\matK$ is also sampled as $\matK\leftarrow \matA^{-1}(\matQ)$. The intuition behind this modified assumption remains similar to the original one.

\item \textbf{Generalized to the multi-instance setting.} In analogy to \cite{WWW22}, which modified the evasive \LWE assumption of \cite{Wee22}, we present a tailored variant of the evasive $\IPFE$ assumption~\cite{HLL24} for our multi-authority setting. We can also extend the variant from Step 1) into our final $\EVIPFE$ assumption. Specifically, for each index $i\in [\ell]$, let $\matA_i,\matB_i,\matP_i$ be uniformly and independently random matrices. Let $\{\vecv_{i,j}\}_{i\in[\ell],j\in[k_i]},\{\vecr_{i,j}\}_{i\in[\ell],j\in[k_i]}$ be vectors generated by a public-coin sampler with public randomness $\vecr_{\pub}$, and let $\vecu$ be generated by a private-coin sampler. Define the matrices $\matQ_i$ for each $i\in [\ell]$ as: $$\matQ_i\leftarrow [\matB_i\mid\matP_i]\left[\begin{array}{c|c|c}
            \vecr_{i,1} & \cdots &\vecr_{i,k_i}  \\
            \matG^{-1}(\vecv_{i,1}) & \cdots &\matG^{-1}(\vecv_{i,k_i})
        \end{array}\right].$$ The $\EVIPFE$ assumption states that if the following \emph{precondition} holds:
        \begin{align}\label{tecevpre}
        \begin{split}
             &(\vecr_{\pub},\{\matA_i,\uwave{\vecs_i^\top\matA_i}\},\{\matB_i\},\uwave{\sum\vecs_i^\top\matB_i},\{\matP_i\},\uwave{\sum\vecs_i^\top\matP_i}+\vecu^\top\matG,\{\uwave{\vecs_i^\top\matQ_i}\})\\
            \overset{c}{\approx}\ &(\vecr_{\pub},\{\matA_i,\$_{1,i}\},\hspace{0.8em} \{\matB_i\},\hspace{1.8em}\$_{2},
             \hspace{1.1em}\{\matP_i\}, \hspace{3em}\$_2', \hspace{3em}\{\$_{3,i}\}),
        \end{split}
        \end{align}  then the \emph{postcondition}
        \begin{align}\label{tecevpost}
        \begin{split}
            &(\vecr_{\pub},\{\matA_i,\uwave{\vecs_i^\top\matA_i}\},\{\matB_i\},\uwave{\sum\vecs_i^\top\matB_i},\{\matP_i\},\uwave{\sum\vecs_i^\top\matP_i}+\vecu^\top\matG,\{\matK_i\})\\
            \overset{c}{\approx}\ &(\vecr_{\pub},\{\matA_i,\$_{1,i}\},\hspace{0.8em} \{\matB_i\},\hspace{1.8em}\$_{2},
             \hspace{1.1em}\{\matP_i\}, \hspace{3em}\$_2', \hspace{3em}\{\matK_i\})
        \end{split}
        \end{align} also holds, where $\matK_i\leftarrow\matA_i^{-1}(\matQ_i)$. Informally, the $\EVIPFE$ assumption is no stronger than the evasive $\LWE$ assumption of~\cite{WWW22}, as it admits a reduction from the latter. Details of the reduction are provided in Appendix~\ref{sec:app}.
\end{enumerate}

Based on this $\EVIPFE$ assumption, we can establish the (weakened) static security as an $\maev$ scheme. The matrices $\matA_i,\matB_i,\matP_i$ in the assumption correspond to the public keys for the authorities. Each column of $\matQ_i$ is of the form $\matP_{i}\matG^{-1}(\vecv_{i,j})+\matB_{i}\vecr_{i,j}$. Here,  $\vecv_{i,j}$ corresponds to the key vector in a secret-key query in the game, and $\vecr_{i,j}$ is interpreted as the hash output of some tuple $(\gid,\vecv)$ in the scheme. By carefully designing the matrix, the columns of $\matK_i$ in the postcondition instance serve as the secret-key responses to the secret-key queries. Noting that each secret-key query satisfies either condition (a) or condition (b), we can show the constructed matrices satisfy the precondition \eqref{tecevpre}. By the $\EVIPFE$ assumption, we obtain that postcondition \eqref{tecevpost} holds, implying that the resulting ciphertext is pseudorandom. A complete formal proof is provided in Section~\ref{sec:6}.

\iitem{Desired Static Security: $\manipfe$ Scheme}
In the static security model, the adversary is required to commit to its full set of secret-key queries and challenge plaintexts in advance. A statically secure $\manipfe$ scheme requires that no efficient adversary can distinguish the ciphertexts encrypting the challenge plaintexts with non-negligible advantage, even in the presence of corrupted authorities and access to the secret keys corresponding to its queries.

\iitem{Our Assumption: Indistinguishability-Based Evasive $\IPFE$}
We base the security of our construction on the \emph{Indistinguishability-based Evasive} $\IPFE$ ($\INDIPFE$) assumption, which is a variant of the $\EVIPFE$ assumption discussed earlier.

The $\INDIPFE$ assumption follows a structure similar to the  $\EVIPFE$ assumption introduced in this work. Specifically, for each index $i\in [\ell]$, let $\matA_i,\matB_i,\matP_i$ be independent and uniformly random matrices. Let $\{\vecv_{i,j}\}_{i\in[\ell],j\in[k_i]},\{\vecr_{i,j}\}_{i\in[\ell],j\in[k_i]}$ and $\vecu_0,\vecu_1$ be vectors generated by a public-coin sampler with public randomness $\vecr_{\pub}$. Define the matrices $\matQ_i$ for each $i\in [\ell]$ as: $$\matQ_i\leftarrow [\matB_i\mid\matP_i]\left[\begin{array}{c|c|c}
            \vecr_{i,1} & \cdots &\vecr_{i,k_i}  \\
            \matG^{-1}(\vecv_{i,1}) & \cdots &\matG^{-1}(\vecv_{i,k_i})
        \end{array}\right].$$ The $\INDIPFE$ assumption states that if the following \emph{precondition} holds:
        \begin{align}\label{tecindpre}
        \begin{split}
            &(\vecr_{\pub},\{\matA_i,\uwave{\vecs_i^\top\matA_i}\},\{\matB_i\},\uwave{\sum\vecs_i^\top\matB_i},\{\matP_i\},\uwave{\sum\vecs_i^\top\matP_i}+\vecu_0^\top\matG,\{\uwave{\vecs_i^\top\matQ_i}\})\\
            \overset{c}{\approx}\ &(\vecr_{\pub},\{\matA_i,\uwave{\vecs_i^\top\matA_i}\},\{\matB_i\},\uwave{\sum\vecs_i^\top\matB_i},\{\matP_i\},\uwave{\sum\vecs_i^\top\matP_i}+\vecu_1^\top\matG,\{\uwave{\vecs_i^\top\matQ_i}\}),
            \end{split}
        \end{align}  then the following \emph{postcondition}
        \begin{align}\label{tecindpost}
        \begin{split}
            &(\vecr_{\pub},\{\matA_i,\uwave{\vecs_i^\top\matA_i}\},\{\matB_i\},\uwave{\sum\vecs_i^\top\matB_i},\{\matP_i\},\uwave{\sum\vecs_i^\top\matP_i}+\vecu_0^\top\matG,\{\matK_i\})\\
            \overset{c}{\approx}\ &(\vecr_{\pub},\{\matA_i,\uwave{\vecs_i^\top\matA_i}\},\{\matB_i\},\uwave{\sum\vecs_i^\top\matB_i},\{\matP_i\},\uwave{\sum\vecs_i^\top\matP_i}+\vecu_1^\top\matG,\{\matK_i\}),
             \end{split}
        \end{align} where $\matK_i\leftarrow\matA_i^{-1}(\matQ_i)$.

However, we need to make some additional assumptions about the sampler for vectors $\vecr_{i,j}$ and $\vecv_{i,j}$. In this assumption, each pair $(\vecr_{i,j},\vecv_{i,j})$ is required to be nonzero to avoid the obvious attack against the assumption. For the justification of this setting, we refer to Section~\ref{INDIPFE} for details.

Following a strategy similar to the $\maev$ case, we can prove that our construction satisfies the desired $\manipfe$ static security notion. In this setting, all secret-key queries are fully determined by the adversary and the challenger plaintexts $\vecu_0,\vecu_1$ are chosen by the adversary. Using the same method as for the $\maev$ scheme, we construct each matrix $\matQ_i$ based on the submitted secret-key queries. By carefully designing $\matQ_i$, the resulting columns of each $\matK_i$ serve as valid responses to the corresponding secret-key queries. The restriction on secret-key queries to satisfy either condition (a) or (b) ensures that the parameters $\matA_i,\matB_i,\matP_i,\matQ_i$ meet the precondition \eqref{tecindpre} of the assumption. Then by the postcondition \eqref{tecindpost} of the $\INDIPFE$ assumption, we obtain that the ciphertexts of $\vecu_0$ and $\vecu_1$ are indistinguishable, thereby establishing the desired security.

\iitem{Removing the Noise: Noiseless $\maipfe$}
Finally, we construct a noiseless $\maipfe$ scheme enabling computation of \emph{exact} inner products rather than approximate values, by slightly modifying the scheme proposed in Section~\ref{scheme}. Specifically, we eliminate the additive error term in the approximate inner products during decryption. Unlike the scheme in Section~\ref{scheme} that operates over plaintext vectors in $\Z_q^n$, this modified construction operates over input vectors in the smaller domain $\Z_p^n$ with $p\ll q$.

The global setup, authority setup, and key generation algorithms remain unchanged. To support computation over $\Z_p^n$ within a $\Z_q$-based lattice setting, we adopt the \emph{modulus-switching function} $\lceil\cdot\rfloor_{p\rightarrow q}$ that maps elements from $\Z_p$ to $\Z_q$: For a vector $\vecu\in\Z_p^n$, we define the encoding function: $$\lceil\vecu\rfloor_{p\rightarrow q}:\vecu\mapsto\left\lceil\frac{q}{p}\cdot \vecu\right\rfloor\in \Z_q^n.$$ Similarly, we can define the decoding function $\lceil\cdot\rfloor_{q\rightarrow p}$ as the inverse mapping of $\lceil\cdot\rfloor_{p\rightarrow q}$. With $\lceil\cdot\rfloor_{p\rightarrow q}$, we can embed $\vecu\in\Z_p^n$ into $\lceil\vecu\rfloor_{p\rightarrow q}\in\Z_q^n$. The encryption algorithm proceeds as 
$$\ct\leftarrow\left(\{\uwave{\vecs_i^\top\matA_i}\}_{i\in X},\uwave{\sum_{i\in X}\vecs_i^\top\matB_i},\uwave{\sum_{i\in X}\vecs_i^\top\matP_i}+\lceil\vecu^\top\rfloor_{p\rightarrow q}\matG\right),$$ where $\vecs_i\rand\Z_q^n$ for each $i\in X$. On the decryption side, we can use the same decryption process over $\Z_q$ as in Section~\ref{scheme}, followed by a final application of $\lceil\cdot\rfloor_{q\rightarrow p}$ to recover the inner-product result over $\Z_p$.

This modification to the encryption and decryption procedures removes the approximation gap in the decryption procedure while still enabling fine-grained multi-authority access control. The resulting scheme, formally given in Construction \ref{con2}, is particularly well-suited for applications requiring deterministic, noise-free computation of inner products, and its security is established under the standard \LWE assumption and the $\INDIPFE$ assumption tailored for noiseless settings. A detailed description of this noiseless scheme and its security analysis can be found in Section~\ref{sec:8}.

\section{Preliminaries}
\label{sec:2}
\subsection{Notations}
Let $\lambda\in\N$ denote the security parameter used throughout this paper. For a positive integer $n\in \N$, denote $[n]$ as the set $\{k\in\N:1\leq k\leq n\}$. For a positive integer $q\in\N$, let $\Z_q$ denote the ring of integers modulo $q$. Throughout this paper, vectors are assumed to be column vectors by default. We use bold lower-case letters (e.g., $\vecu,\vecv$) for vectors and bold upper-case letters (e.g., $\matA,\matB$) for matrices. Let $\vecv[i]$ denote the $i$-th entry of the vector $\vecv$, and let $\matU[i,j]$ denote the $(i,j)$-entry of the matrix $\matU$. 
 We write $\veczero_n$ for the all-zero vector of length $n$, and $\veczero_{n\times m}$ for the all-zero matrix of dimension $n\times m$. The \emph{infinity norm} of a vector $\vecv$ and the corresponding operator norm of a matrix $\matU$ are defined as: $$\|\vecv\|=\max_i|\vecv[i]|,\ \|\matU\|=\max_{i,j}|\matU[i,j]|.$$ For two matrices
$\matA,\matB$ of dimensions $n_1\times m_1$ and $n_2\times m_2$, respectively, their \emph{Kronecker product} is an $n_1n_2\times m_1m_2$ matrix given by $$\matA\otimes\matB=\left[\begin{array}{ccc}
     \matA[1,1]\matB&\cdots&\matA[1,m_1]\matB  \\
     \vdots&\ddots&\vdots \\
     \matA[n_1,1]\matB&\cdots&\matA[n_1,m_1]\matB
\end{array}\right].$$ 

A function $f(\lambda)$ is called \emph{negligible} if $f(\lambda) = O(\lambda^{-c})$ for every constant $c > 0$. We denote a negligible function of $\lambda$ by $\negl(\lambda)$. A function $f(\lambda)$ is called \emph{polynomial} if $f(\lambda) = O(\lambda^c)$ for some constant $c > 0$. We denote a polynomial function of $\lambda$ by $\poly(\lambda)$. A function $f(\lambda)$ is called \emph{super-polynomial}, if $f(\lambda)=\omega(\lambda^{c})$ for every constant $c>0$. We denote a super-polynomial function of $\lambda$ by $\superpoly(\lambda)$. We say an event occurs with \emph{overwhelming probability} if its probability is $1-\negl(\lambda)$. An algorithm is called \emph{efficient} if it runs in probabilistic polynomial time in its input length, typically parameterized by the security parameter $\lambda$ in this work.

We denote the sampling of an element $x$ from a distribution $D$ by $x\leftarrow D$, and the uniform sampling of an element $x$ from a set $S$ by $x\xleftarrow{\$}S$. For a distribution $D$ and a positive integer $m$, we write $\vecx \leftarrow D^m$ to denote an $m$-dimensional vector whose entries are independently sampled from $D$. More generally, for integers $m,n \in \N$, we write $\matX \leftarrow D^{m \times n}$ to denote an $m \times n$ matrix with independently sampled entries from $D$. For two distributions $D_1$ and $D_2$, We use the notation $D_1 \equiv D_2$ to denote that the two distributions are identical.

We also define the \emph{modulus-switching function}. For an element $x\in\Z_q$ and some positive integer $p$, we define a function mapping from $\Z_q$ to $\Z_p$ as $$\lceil \cdot\rfloor_{q\rightarrow p}:x\mapsto \left\lceil \frac{p}{q}\cdot x\right\rfloor\in\Z_p.$$ This map is applied coefficient-wise when extended to vectors or matrices. Such functions are commonly used in cryptographic schemes for modulus compression and message encoding.

\iitem{Indistinguishability}
 For two distributions $D_1$ and $D_2$ over a discrete domain $\Omega$, the \emph{statistical distance} between them is defined as $$\SD(D_1,D_2)=(1/2)\cdot \sum_{\omega\in\Omega}|D_1(\omega)-D_2(\omega)|.$$ Let $D_1 = \{D_{1,\lambda}\}_{\lambda \in \N}$ and $D_2 = \{D_{2,\lambda}\}_{\lambda \in \N}$ be two ensembles of distributions parameterized by the security parameter $\lambda$. We say $D_1$ and $D_2$ are \emph{statistically indistinguishable} if there exists a negligible function $\negl(\cdot)$ such that for all $\lambda\in \N$, $$\SD(D_{1,\lambda},D_{2,\lambda})\leq\negl(\lambda).$$ We say $D_1$ and $D_2$ are \emph{computationally indistinguishable} if for all $\lambda\in \N$, for all \emph{efficient} algorithms $\A$, there exists a negligible function $\negl(\cdot)$ such that $$\Pr[\A(1^\lambda,x_{\lambda})=1:x_{\lambda}\leftarrow D_{1,\lambda}]-\Pr[\A(1^\lambda,y_{\lambda})=1:y_{\lambda}\leftarrow D_{2,\lambda}]\leq \negl(\lambda).$$ We use the notations $\overset{s}{\approx}$ and $\overset{c}{\approx}$ to denote statistical and computational indistinguishability, respectively.

\subsection{Lattice Preliminaries}
\noindent We recall some basic concepts related to lattices. 

\iitem{Discrete Gaussians} Let $D_{\Z,\sigma}$ represent the centered \emph{discrete Gaussian distribution} over $\Z$ with standard deviation $\sigma\in\R^{+}$ (e.g.,  \cite{banaszczyk1993new}). For a matrix $\matA\in\Z_q^{n\times m}$ and a vector $\vecv\in\Z_q^n$, let $\matA_\sigma^{-1}(\vecv)$ denote the distribution of a random variable $\vecu\leftarrow D_{\Z,\sigma}^m$ conditioned on $\matA\vecu=\vecv\mod q$. 
When $\matA_{\sigma}^{-1}$ is applied to a matrix, it is understood as being applied independently to each column of the matrix.

The following lemma (e.g.,  \cite{micciancio2007worst}) shows that for a given Gaussian width parameter $\sigma=\sigma(\lambda)$, the probability of a vector drawn from the discrete Gaussian distribution having norm greater than $\sqrt{\lambda}\sigma$ is negligible.

\begin{lemma}\label{truncated}
    Let $\lambda$ be a security parameter and $\sigma=\sigma(\lambda)$ be a Gaussian width parameter. Then for all polynomials $n=n(\lambda)$, there exists a negligible function $\negl(\lambda)$ such that for all $\lambda\in\N$,
$$\Pr\left[\|\vecv\|>\sqrt{\lambda}\sigma:\vecv\leftarrow D_{\Z,\sigma}^n\right]=\negl(\lambda). $$
\end{lemma}

\iitem{Smudging Lemma} In the following, we present the standard smudging lemma, which formalizes the intuition that sufficiently large standard deviation can ``smudge out'' small perturbations, making the resulting distributions  statistically indistinguishable.

\begin{lemma}[\cite{BDE18}]\label{smudge}
    Let $\lambda$ be a security parameter, and let $e\in\Z$ satisfy $|e|\leq B$. Suppose that $\sigma\geq B\cdot\lambda^{\omega(1)}$. Then, the following two distributions are statistically indistinguishable:
    $$ \{z: z\leftarrow D_{\Z,\sigma}\}\, \text{ and }\,
    \{z+e: z\leftarrow D_{\Z,\sigma}\}.  $$
\end{lemma}

\iitem{Lattice Trapdoors} Lattices with trapdoors are structured lattices that are computationally indistinguishable from randomly chosen lattices (without auxiliary information). However, they possess specific ``trapdoors''---short bases or auxiliary structures---that enable efficient solutions to otherwise hard lattice problems such as preimage sampling for the \emph{Short Integer Solution} (\SIS) function. In this work, we adopt the trapdoor frameworks outlined in~\cite{GPV08,MP12}, following the formalization in \cite{BTVW17}.

\begin{lemma}[\cite{GPV08,MP12}]\label{trapdoor}
    Let $n,m,q$ be lattice parameters. Then there exist two efficient algorithms (\trapgen, \samplepre) with the following syntax:
    \begin{itemize}
        \item $\trapgen(1^n,1^m,q)\rightarrow (\matA,\td_\matA)$: On input the lattice dimension $n$,  number of samples $m$, modulus $q$, this randomized algorithm outputs a matrix $\matA\rand\Z_q^{n\times m}$ together with a trapdoor $\td_\matA$.
        \item $\samplepre(\matA,\td_\matA,\vecy,\sigma)\rightarrow \vecx$: On input ($\matA,\td_\matA$) from \trapgen, a target vector $\vecy\in\Z_q^n$ and a Gaussian width parameter $\sigma$, this randomized algorithm outputs a vector $\vecx\in\Z^m$.
    \end{itemize}
    
    Moreover, there exists a polynomial $m_0=m_0(n,q)=O(\sqrt{n\log q})$ such that for all $m\geq m_0$, the above algorithms satisfy the following properties:
    \begin{itemize}
        \item \textbf{\emph{Trapdoor distribution:}} The matrix $\matA$ output by $\trapgen(1^n,1^m,q)$ is statistically close to a uniformly random matrix. Specifically, if $(\matA,\td_\matA)\leftarrow \trapgen(1^n,1^m,q)$ and $\matA'\rand \Z_q^{n\times m}$, then the statistical distance between the distributions of $\matA$ and $\matA'$ is at most $2^{-n}$.
        
        \item \textbf{\emph{Preimage sampling:}} Suppose  $\tau=O(\sqrt{n\log q})$. Then for all $\sigma\geq \tau\cdot \omega(\sqrt{\log n})$ and all target vectors $\vecy\in\Z_q^n$, the statistical distance between the following two distributions is at most $2^{-n}$:
        $$\{\vecx \leftarrow \samplepre(\matA,\td_\matA,\vecy,\sigma) \}\, \text{ and }\, 
       \{\vecx \leftarrow \matA_\sigma^{-1}(\vecy)\}.   $$
    \end{itemize}
\end{lemma}

\iitem{The Gadget Matrix}
Here we recall the definition of the \emph{gadget matrix} introduced in \cite{MP12}. For positive integers $n,q\in\N$, let $\vecg^{\top}=(1,2,\ldots,2^{\lceil\log q\rceil-1})$ denote the \emph{gadget vector}. Define the \emph{gadget matrix} $\matG_n=\matI_n\otimes \vecg^{\top}\in\Z_q^{n\times m}$, where $m=n\lceil \log q\rceil$. The inverse function $\matG_n^{-1}$ maps vectors in $\Z_q^n$ to their binary expansions in $\{0,1\}^m$. More precisely, for a vector $\vecx\in\Z_q^n$, the output $\vecy=\matG_n^{-1}(\vecx)$ is given by $$\vecy=(y_{1,0},\ldots,y_{1,\grow-1},\ldots,y_{n,0},\ldots,y_{n,\grow-1}),$$ where the sequence $(y_{i,0},\ldots,y_{i,\grow})$ corresponds to the $\grow$ bits of the binary representation of $\vecx[i]$. The function $\matG_n^{-1}(\cdot)$ extends naturally to matrices by applying it column-wise. It is straightforward to verify that $\matG_n\cdot \matG_n^{-1}(\matA)=\matA \mod q$ for any matrix $\matA\in\Z_q^{n\times t}$. For simplicity, we will omit the subscript when the parameter $n$ is clear from context.

\iitem{Preimage Sampling} The following lemma describes a useful property of discrete Gaussian distributions. Specifically, it basically states that, given a uniformly random matrix $\matA$ with sufficiently many columns, two methods of generating an input-output pair $(\vecx,\vecy=\matA\vecx)$ produce statistically indistinguishable distributions.

\begin{lemma}[Preimage Sampling~\cite{GPV08}]\label{preimage}
   Let $n,m,q\in \N$, where $q>2$ is a prime. Then there exists polynomials $m_0(n,q)=O(\sqrt{n\log q})$ and $\chi_0(n,q)=\sqrt{n\log q}\cdot\omega(\sqrt{\log n})$ such that for all $m\geq m_0(n,q)$ and $\chi\geq \chi_0(n,q)$, the following two distributions are statistically indistinguishable:
    $$
\{(\matA,\vecx,\matA\vecx):\matA\xleftarrow{\$}\Z_q^{n\times m},\vecx\leftarrow D_{\Z,\chi}^m\}\,\text{ and } \,\{(\matA,\vecx,\vecy):\matA\xleftarrow{\$}\Z_q^{n\times m},\vecy\xleftarrow{\$} \Z_q^n,\vecx\leftarrow\matA_{\chi}^{-1}(\vecy)\}. $$
\end{lemma}

\subsection{\LWE Assumption and Its Variants}
\iitem{Learning with Errors Assumption} We review the Learning With Errors (\LWE) assumption and its variants used throughout this work.

\begin{assumption}[Learning with Errors (\LWE)~\cite{regev2009lattices}]
    Let $\lambda$ be a security parameter and let $n=n(\lambda),m=m(\lambda),q=q(\lambda),\sigma=\sigma(\lambda)$. The decisional learning with errors assumption, denoted $\LWE_{n,m,q,\sigma}$, states that for $\matA\xleftarrow{\$}\Z_q^{n\times m},\vecs\xleftarrow{\$}\Z_q^n,\vece\leftarrow D_{\Z,\sigma}^m$, and $\vecdelta\xleftarrow{\$}\Z_q^m$, the following two distributions are computationally indistinguishable: $$(\matA,\vecs^{\top}\matA+\vece^{\top})\,\text{ and }\,(\matA,\vecdelta).$$
\end{assumption}

\begin{assumption}[Flipped \LWE ~\cite{BLMR13}]
    Let $\lambda$ be a security parameter and let $k=k(\lambda),m=m(\lambda),q=q(\lambda),\chi=\chi(\lambda),\sigma=\sigma(\lambda)$. The flipped \LWE assumption, denoted $\flipLWE_{k,m,q,\chi,\sigma}$,  states that for $\matA\xleftarrow{\$} D_{\Z,\chi}^{k\times m},\vecs\xleftarrow{\$}\Z_q^k,\vece\leftarrow D_{\Z,\sigma}^m$, and $\vecdelta\xleftarrow{\$}\Z_q^m$, the following two distributions are computationally indistinguishable: $$(\matA,\vecs^{\top}\matA+\vece^{\top})\,\text{ and }\,(\matA,\vecdelta^\top).$$
\end{assumption}

The following theorem establishes the hardness of the flipped $\LWE$ assumption based on the standard $\LWE$ assumption.

\begin{theorem}[\cite{BLMR13}]\label{flipped}
    Let $q$ be an integer, and suppose $k\geq 6n\log q$ and $\sigma=\Omega(\sqrt{n\log q})$. Then, if the standard $\LWE_{n,m,q,\chi}$ assumption holds, we have that the assumption $\flipLWE_{k,m,q,\sigma,\chi}$ holds.
\end{theorem}
 
 \section{Evasive $\IPFE$ Assumptions}\label{Evasiveipfe}
\subsection{The Evasive $\IPFE$ Assumption}
 We first present a variant of evasive $\IPFE$ assumption introduced by \cite{HLL24}, with slight modifications tailored to our setting.
 
\begin{assumption}[Evasive \IPFE (\EVIPFE) Assumption, adapted, \cite{WWW22,HLL24}]\label{evIPFE}
    Let $\lambda\in\N$ be a security parameter, and let $n,q,m=n\lceil\log q\rceil,\chi,\chi'$ be lattice parameters specified by $\lambda$. Denote $\gp=(1^\lambda,q,1^n,1^m,1^{m'},1^\chi,1^{\chi'})$. Let $\mathcal{S}=(\mathcal{S}_{\vecv},\mathcal{S}_{\vecu})$ be a pair of sampling algorithms with the following syntax: 
    \begin{itemize}
        \item $\mathcal{S}_{\vecv}(\gp;\vecr_{\pub})$: Given the global parameter $\gp$ and randomness $\vecr_{\pub}\in\{0,1\}^*$, the algorithm outputs parameters $$1^{\ell}, \{1^{k_i}\}_{i\in[\ell]},$$ as well as  $k_1+k_2+\cdots+k_{\ell}$ vector tuples $$(\vecr_{1,1},\vecv_{1,1}),(\vecr_{1,2},\vecv_{1,2}),\ldots,(\vecr_{1,k_1},\vecv_{1,k_1});\ldots;(\vecr_{\ell,1},\vecv_{\ell,1}),\ldots,(\vecr_{\ell,k_\ell},\vecv_{\ell,k_\ell}),$$ where $\vecr_{i,j}\in\Z_q^{m'}, \vecv_{i,j}\in\Z_q^n$ for all $i\in [\ell],j\in [k_i]$.
        
        \item $\vecu\leftarrow \mathcal{S}_{\vecu}(\gp,\vecr_{\pub};\vecr_{\pri})$: Given the global parameter $\gp$, public randomness $\vecr_{\pub}$ used in $\mathcal{S}_{\vecv}$, and private randomness $\vecr_{\pri}\in\{0,1\}^*$. It outputs a vector $\vecu\in\Z_q^n$.
    \end{itemize}
 For two adversaries $\mathcal{A}_0$ and $\mathcal{A}_1$, we define their advantage functions as follows:
    \begin{align*}
        \Adv_{\mathcal{S},\mathcal{A}_0}^{\pre}(\lambda)&:=\left|\pr\left[\mathcal{A}_0(1^\lambda,\vecr_{\pub},\{\matA_i,\vecz_{1,i}^\top\}_{i\in[\ell]},[\matB\mid\matP],\vecz_2^\top,\{\vecz_{3,i}^\top\}_{i\in [\ell]})=1\right]\right.\\
        &-\left.\pr\left[\mathcal{A}_0(1^\lambda,\vecr_{\pub},\{\matA_i,\vecdelta_{1,i}^\top\}_{i\in[\ell]},[\matB\mid\matP],\vecdelta_2^\top,\{\vecdelta_{3,i}^\top\}_{i\in 
        [\ell]})=1\right]\right|;\\
        \Adv_{\mathcal{S},\mathcal{A}_1}^{\post}(\lambda)&:=\left|\pr\left[\mathcal{A}_1(1^\lambda,\vecr_{\pub},\{\matA_i,\vecz_{1,i}^\top\}_{i\in[\ell]},[\matB\mid\matP],\vecz_2^\top,\{\matK_i\}_{i\in [\ell]})=1\right]\right.\\
        &-\left.\pr\left[\mathcal{A}_1(1^\lambda,\vecr_{\pub},\{\matA_i,\vecdelta_{1,i}^\top\}_{i\in[\ell]},[\matB\mid\matP],\vecdelta_2^\top,\{\matK_i\}_{i\in [\ell]})=1\right]\right|;
    \end{align*}
    where the parameters are sampled as follows:
    \begin{itemize}
        \item $\left(\begin{array}{l}
          1^{\ell}, \{1^{k_i}\}_{i\in[\ell]};\\
        (\vecr_{1,1},\vecv_{1,1}),\ldots,(\vecr_{1,k_1},\vecv_{1,k_1});\\
            \cdots\\
            (\vecr_{\ell,1},\vecv_{\ell,1}),\ldots,(\vecr_{\ell,k_\ell},\vecv_{\ell,k_\ell})      \end{array}\right)\leftarrow\mathcal{S}_{\vecv}(\gp;\vecr_{\pub}), 
       \quad \vecu\leftarrow\mathcal{S}_{\vecu}(\gp,\vecr_{\pub};\vecr_{\pri}),$
       \item $\matB_1,\ldots,\matB_{\ell}\rand\Z_q^{n\times m'},\matP_1,\ldots,\matP_{\ell}\rand\Z_q^{n\times m},
    \matB^\top\leftarrow[\matB_1^\top\mid\cdots\mid\matB_{\ell}^\top],\matP^\top\leftarrow[\matP_1^\top\mid\cdots\mid\matP_{\ell}^\top],$
    \item $\matQ_i\leftarrow [\matB_i\mid\matP_i]\left[\begin{array}{c|c|c}
            \vecr_{i,1} & \cdots &\vecr_{i,k_i}  \\
            \matG^{-1}(\vecv_{i,1}) & \cdots &\matG^{-1}(\vecv_{i,k_i})
        \end{array}\right]\in \Z_q^{n\times k_i},$
        \item $(\matA_1,\td_1),\ldots,(\matA_{\ell},\td_\ell)\leftarrow\trapgen(1^n,1^m,q),$
        \item $\vecs_1,\ldots,\vecs_{\ell}\rand\Z_q^n, \vecs^\top\leftarrow[\vecs_1^\top\mid\ldots\mid\vecs_{\ell}^\top]\in\Z_q^{n\ell},$
        \item $\vece_{1,i}\leftarrow D_{\Z,\chi}^m,\vece_{3,i}\leftarrow D_{\Z,\chi}^{k_i} \text{ for each } i\in [\ell], \vece_2\leftarrow D_{\Z,\chi}^{m'+m},$
        \item $\vecdelta_{1,i}\rand \Z_q^m,\vecdelta_{3,i}\rand \Z_q^{k_i} \text{ for each } i\in [\ell], \vecdelta_2\rand \Z_q^{m'+m},$
        \item $\vecz_{1,i}^\top\leftarrow \vecs_i^\top\matA_i+\vece_{1,i}^\top,\vecz_{3,i}^\top\leftarrow\vecs_i^\top\matQ_i+\vece_{3,i}^\top\text{ for each } i\in [\ell], \quad\vecz_{2}^\top\leftarrow\vecs^\top[\matB\mid\matP]+\vece_2^\top+[\veczero\mid\vecu^\top\matG]$,
        \item $\matK_i\leftarrow\samplepre(\matA_{i},\td_{i},\matQ_i,\chi') \text{ for each } i\in [\ell].$
    \end{itemize}
    
     We say that the $\EVIPFE_{n,m,m',q,\chi,\chi'}$ assumption holds, if for all efficient samplers $\mathcal{S}$, the following implication holds: If there exists an efficient adversary $\A_1$ with a non-negligible advantage function $\Adv_{\mathcal{S},\mathcal{A}_1}^\post(\lambda)$, then there exists another efficient adversary $\A_0$ with a non-negligible advantage function $\Adv_{\mathcal{S},\mathcal{A}_0}^\pre(\lambda)$. 
\end{assumption}

\begin{remark}[Relation to the Evasive \IPFE Assumption in~\cite{HLL24}]
    Assumption \ref{evIPFE} defined above originates from the evasive $\IPFE$ assumption introduced by \cite{HLL24}. Following the approach of~\cite{WWW22}, which defines a variant of public-coin evasive \LWE based on the formulation in~\cite{Wee22} by introducing multiple (i.e., $\ell \geq 1$) independently and uniformly sampled matrices $\matA_i \in \Z_q^{n \times m}$, we apply an analogous modification to the original evasive $\IPFE$ assumption proposed in~\cite{HLL24}.
    
    In the following, we focus on the connection between Assumption \ref{evIPFE} and the evasive \IPFE assumption from \cite{HLL24}. Informally, the latter basically states if \begin{align}\label{evprecon}
    (1^\lambda,\vecr_{\pub},\matA,\vecz_1^\top,\matP,\vecz_2^\top+\vecu^\top\matG,\vecz_{3}^\top)\overset{c}{\approx} (1^\lambda,\vecr_{\pub},\matA,\vecdelta_1^\top,\matP,\vecdelta_2^\top,\vecdelta_{3}^\top), 
    \end{align} where $\vecu,\vecv_1,\ldots,\vecv_{k}\in \Z_q^{d}$ are generated by some sampler $\mathcal{S}'=(\mathcal{S}'_{\vecv},\mathcal{S}'_{\vecu})$, and the other parameters are sampled as follows: \begin{align*}
        &\matA\rand\Z_q^{n\times m},\matP\rand\Z_q^{n\times D},\matQ\leftarrow\matP[\matG^{-1}(\vecv_1)\mid\cdots\mid\matG^{-1}(\vecv_{k})],\\
        &\vecs\rand\Z_q^n,\vece_1\leftarrow D_{\Z,\chi}^m,\vece_2\leftarrow D_{\Z,\chi}^D,\vece_3\leftarrow D_{\Z,\chi}^k,\\
        &\vecz_{1}^\top\leftarrow\vecs^\top\matA+\vece_1^\top,\vecz_2^\top\leftarrow\vecs^\top\matP+\vece_2^\top,\vecz_3^\top\leftarrow \vecs^\top\matQ+\vece_3^\top,\\
        &\vecdelta_1\rand\Z_q^m,\vecdelta_2\rand\Z_q^D,\vecdelta_3\rand\Z_q^{k},
    \end{align*}
    then the following postcondition holds $$(1^\lambda,\vecr_{\pub},\matA,\vecz_1^\top,\matP,\vecz_2^\top+\vecu^\top\matG,\matK)\overset{c}{\approx} (1^\lambda,\vecr_{\pub},\matA,\vecdelta_1^\top,\matP,\vecdelta_2^\top,\matK),$$ where $\matK\leftarrow \matA_{\chi'}^{-1}(\matQ)$. 

    We now compare Assumption \ref{evIPFE} with the version in \cite{HLL24}. On the one hand, when $\ell=1$, Assumption \ref{evIPFE} generalizes original setting in \cite{HLL24} by incorporating an auxiliary matrix $\matB\in\Z_q^{n\times m'}$ and additional vectors $\vecr_{1,j}$. Specifically, if we set $m'=0, d=n$ and $D=m$, then Assumption \ref{evIPFE} collapses to the formulation used in \cite{HLL24}.
    
    On the other hand, when $\ell=1$, Assumption \ref{evIPFE} can also be roughly viewed as a special case of \cite{HLL24}, instantiated with a restricted sampler $\mathcal{S}'$. Concretely, suppose $D=m'+m$ and we write the matrix $\matP$ as $[\matB\mid\matP']$ with $\matB\in\Z_q^{n\times m'},\matP'\in \Z_q^{n\times m}$. If the sampler $\mathcal{S}'_{\vecu}$ ensures that $\vecu$ is always of the form $[\veczero_{m'}^\top\mid \vecu']$ for some $\vecu'\in \Z_q^m$, then $$\vecz_2^\top+\vecu^\top\matG=\vecs^\top[\matB\mid\matP]+\vece_2^\top+[\veczero_{m'}^\top\mid\vecu'^\top\matG],$$ which aligns with the corresponding component in Assumption \ref{evIPFE}. Finally, we explain the role of the $\vecr_{1,j}$ components (note $\ell=1$) in Assumption \ref{evIPFE}. 
    Roughly, we consider the term  
    \begin{align*}
        \vecz_3^\top&=\vecs^\top\matP[\matG^{-1}(\vecv_1)\mid\cdots\mid\matG^{-1}(\vecv_{k})]+\vece_3^\top\\&=(\vecs^\top\matP+\vecu^\top\matG)[\matG^{-1}(\vecv_1)\mid\cdots\mid\matG^{-1}(\vecv_{k})]-\vecu^\top[\vecv_1\mid\cdots\mid\vecv_{k}]+\vece_3^\top\\
        &\overset{c}{\approx}\vecdelta_2^\top[\matG^{-1}(\vecv_1)\mid\cdots\mid\matG^{-1}(\vecv_{k})]-\vecu^\top[\vecv_1\mid\cdots\mid\vecv_{k}]+\vece_3^\top.
    \end{align*}
    Since $\vecv_1,\vecv_2,\ldots,\vecv_k$ are public, and $\vecz_3^\top$ includes the term $\vecu^\top[\vecv_1 \mid\cdots\mid \vecv_k]$ masked with the noise $\vece_3$, its distribution—conditioned on all other public values—is effectively determined by the inner products $\vecu^\top \vecv_j$ and the added noise. In particular, if $\vecu$ is constrained to have zero entries in its first $m'$ coordinates, then the first $m'$ entries of $\vecv_j$ (corresponding to the $\vecr_{1,j}$ components in Assumption \ref{evIPFE}) have no influence on the distribution of $\vecz_3$. This explains why such components can be included in the assumption without affecting its security implications.
\end{remark}

\begin{remark}[Public-coin and private-coin sampler]
As justified by \cite{HLL24}, it is necessary for the sampling procedure of $\vecu$ to be \emph{private-coin}, in order to prevent trivial attacks where the adversary directly computes the inner-products $\vecu^\top\vecv_j$. In our setting, all other components—including the global parameters and the tuples $(\vecr_{i,j},\vecv_{i,j})$—are sampled using a \emph{public-coin} process, thereby avoiding the obfuscation-based counterexamples (cf.~\cite{BUW24}). Moreover, the public-coin nature of $(\vecr_{i,j},\vecv_{i,j})$ ensures that these vectors can be reused across multiple experiments.
\end{remark}

\begin{remark}[Reduction from Evasive $\LWE$ in~\cite{WWW22}]
    In Appendix~\ref{sec:app}, we provide a reduction from the evasive $\LWE$ assumption in~\cite{WWW22} to our $\EVIPFE$ assumption. Consequently, the security of our $\maev$ construction (cf. Section~\ref{sec:6}) can be based on the standard $\LWE$ assumption and the evasive $\LWE$ assumption from~\cite{WWW22}.
\end{remark}

\subsection{The Indistinguishability-Based Evasive $\IPFE$ Assumption}
In this section, we present an indistinguishability-based variant of the $\EVIPFE$ assumption introduced in the previous subsection. This version refines the original assumption by replacing distributional indistinguishability with a two-challenge indistinguishability game, where an adversary is challenged to distinguish between two inner-product encodings corresponding to $\vecu_0$ and $\vecu_1$. 

\begin{assumption}[Indistinguishability-Based Evasive $\IPFE$ ($\INDIPFE$) Assumption]\label{INDIPFE}
    Let $\lambda\in\N$ be a security parameter, and let $q,n,m=n\lceil \log q\rceil,m',\chi,\chi'$ be lattice parameters specified by $\lambda$. Denote $\gp=(1^\lambda,q,1^n,1^m,1^{m'},1^\chi,1^{\chi'})$. Let $\mathcal{S}$ be an algorithm defined as follows: 
    \begin{itemize}
        \item $\mathcal{S}(\gp;\vecr_{\pub})$: Given the security parameter $\lambda$ and public randomness $\vecr_{\pub}\in\{0,1\}^*$, the algorithm outputs parameters $$1^{\ell}, \{1^{k_i}\}_{i\in[\ell]},$$ followed by $k_1+\cdots+k_{\ell}$ vector tuples $$(\vecr_{1,1},\vecv_{1,1}),(\vecr_{1,2},\vecv_{1,2}),\ldots,(\vecr_{1,k_1},\vecv_{1,k_1});\ldots;(\vecr_{\ell,1},\vecv_{\ell,1}),\ldots,(\vecr_{\ell,k_\ell},\vecv_{\ell,k_\ell}),$$ where $\vecr_{i,j}\in\Z_q^{m'},\vecv_{i,j}\in\Z_q^n$ and $(\vecr_{i,j},\vecv_{i,j})\neq (\veczero^{m'},\veczero^n)$ for all $i\in [\ell],j\in [k_i]$. The output also includes two additional vectors $\vecu_0,\vecu_1\in\Z_q^n$.
    \end{itemize}
 For two adversaries $\mathcal{A}_0$ and $\mathcal{A}_1$, we define their respective advantage functions as follows:
    \begin{align*}
        \Adv_{\mathcal{A}_0}^{\pre}(\lambda)&:=\left|\pr\left[\mathcal{A}_0(1^\lambda,\vecr_{\pub},\{\matA_i,\vecz_{1,i}^\top\}_{i\in[\ell]},[\matB\mid\matP],\vecz_2^\top+[\veczero\mid\vecu_0^\top\matG],\vecz_{3}^\top)=1\right]\right.\\
        &-\left.\pr\left[\mathcal{A}_0(1^\lambda,\vecr_{\pub},\{\matA_i,\vecz_{1,i}^\top\}_{i\in[\ell]},[\matB\mid\matP],\vecz_2^\top+[\veczero\mid\vecu_1^\top\matG],\vecz_{3}^\top)=1\right]\right|;\\
        \Adv_{\mathcal{A}_1}^{\post}(\lambda)&:=\left|\pr\left[\mathcal{A}_1(1^\lambda,\vecr_{\pub},\{\matA_i,\vecz_{1,i}^\top\}_{i\in[\ell]},[\matB\mid\matP],\vecz_{2}^\top+[\veczero\mid\vecu_0^\top\matG],\{\matK_i\}_{i\in [\ell]})=1\right]\right.\\
        &-\left.\pr\left[\mathcal{A}_1(1^\lambda,\vecr_{\pub},\{\matA_i,\vecz_{1,i}^\top\}_{i\in[\ell]},[\matB\mid\matP],\vecz_{2}^\top+[\veczero\mid\vecu_1^\top\matG],\{\matK_i\}_{i\in [\ell]})=1\right]\right|;
    \end{align*}
    where the parameters are sampled as follows:
    \begin{itemize}
    \item $\left(\begin{array}{l}
         1^{\ell}, \{1^{k_i}\}_{i\in[\ell]};\\
    (\vecr_{1,1},\vecv_{1,1}),\ldots,(\vecr_{1,k_1},\vecv_{1,k_1});\\
            \cdots\\(\vecr_{\ell,1},\vecv_{\ell,1}),\ldots,(\vecr_{\ell,k_\ell},\vecv_{\ell,k_\ell});\\ \vecu_0,\vecu_1\end{array}\right)\leftarrow\mathcal{S}(\gp;\vecr_{\pub}),$
       \item $\matB_1,\ldots,\matB_{\ell}\rand\Z_q^{n\times m'},\matP_1,\ldots,\matP_{\ell}\rand\Z_q^{n\times m},
    \matB^\top\leftarrow[\matB_1^\top\mid\cdots\mid\matB_{\ell}^\top],\matP^\top\leftarrow[\matP_1^\top\mid\cdots\mid\matP_{\ell}^\top],$
    \item $\matQ_i\leftarrow [\matB_i\mid\matP_i]\left[\begin{array}{c|c|c}
            \vecr_{i,1} & \cdots &\vecr_{i,k_i}  \\
            \matG^{-1}(\vecv_{i,1}) & \cdots &\matG^{-1}(\vecv_{i,k_i})
        \end{array}\right]\in \Z_q^{n\times k_i},$
        \item $(\matA_1,\td_1),\ldots,(\matA_{\ell},\td_\ell)\leftarrow\trapgen(1^n,1^m,q),$
        \item $\vecs_1,\ldots,\vecs_{\ell}\rand\Z_q^n, \vecs^\top\leftarrow[\vecs_1^\top\mid\ldots\mid\vecs_{\ell}^\top]\in\Z_q^{n\ell},$
        \item $\vece_{1,i}\leftarrow D_{\Z,\chi}^m,\vece_{3,i}\leftarrow D_{\Z,\chi}^{k_i} \text{ for each } i\in [\ell], \vece_2\leftarrow D_{\Z,\chi}^{m'+m},$
        \item $\vecz_{1,i}^\top\leftarrow \vecs_i^\top\matA_i+\vece_{1,i}^\top,\vecz_{3,i}^\top\leftarrow\vecs_i^\top\matQ_i+\vece_{3,i}^\top\text{ for each } i\in [\ell], \quad\vecz_{2}^\top\leftarrow\vecs^\top[\matB\mid\matP]+\vece_2^\top$,
        \item $\matK_i\leftarrow\samplepre(\matA_{i},\td_{i},\matQ_i,\chi') \text{ for each } i\in [\ell].$
    \end{itemize}
     We say that the $\INDIPFE_{n,m,m',q,\chi,\chi'}$ assumption holds if, for every efficient sampler $\mathcal{S}$, the following implication holds: If there exists an efficient adversary $\A_1$ with non-negligible advantage function $\Adv_{\mathcal{S},\mathcal{A}_1}^\post(\lambda)$, then there exists another efficient adversary $\A_0$ with non-negligible advantage function $\Adv_{\mathcal{S},\mathcal{A}_0}^\pre(\lambda)$. 
\end{assumption}

\begin{remark}[Public-coin Sampler]
In contrast to
Assumption \ref{evIPFE}, Assumption \ref{INDIPFE} adopts a  public-coin sampler. This design is justified by the fact that, in the indistinguishability-based setting, there is no need to keep the challenge vectors $\vecu_0,\vecu_1$ private, since computing $\vecu_0^\top\vecv_j$ and $\vecu_1^\top\vecv_j$ does not help the adversary distinguish between the two challenge instances. Consequently, the two separate algorithms $\mathcal{S}_{\vecu}$ and $\mathcal{S}_{\vecv}$ in Assumption \ref{evIPFE} can be unified into a single public-coin algorithm in this context. The adoption of the public-coin setting also helps to exclude obfuscation-based counterexamples, as discussed in \cite{BUW24}. 
\end{remark}

\begin{remark}
   Assumption \ref{INDIPFE} can be viewed as an indistinguishability-based variant of Assumption \ref{evIPFE}. However, it is important to emphasize that the evasive \IPFE assumption introduced in \cite{HLL24} cannot be directly reformulated as an indistinguishability-style assumption. Specifically, even if the precondition
   \begin{align}\label{precondition2}
   \begin{split}
   &(1^\lambda,\vecr_{\pub},\matA,\vecz_1^\top,\matP,\vecz_2^\top+\vecu_0^\top\matG,\vecz_3^\top)\overset{c}{\approx}(1^\lambda,\vecr_{\pub},\matA,\vecz_1^\top,\matP,\vecz_2^\top+\vecu_1^\top\matG,\vecz_3^\top),     
   \end{split} 
    \end{align} holds, where $\vecu_0,\vecu_1\in \Z_q^d,\matV=[\vecv_1\mid\cdots\mid\vecv_{k}]\in \Z_q^{d\times k}$ are sampled by some algorithm $\mathcal{S}'=(\mathcal{S}'_{\vecv},\mathcal{S}'_{\vecu})$, and the remaining parameters are sampled as follows: \begin{align*}
        &\matA\rand\Z_q^{n\times m},\matP\rand\Z_q^{n\times D},\matQ=\matP\matG^{-1}(\matV)=\matP[\matG^{-1}(\vecv_1)\mid\cdots\mid\matG^{-1}(\vecv_{k})],\\
        &\vecs\rand\Z_q^n,\vece_1\leftarrow D_{\Z,\chi}^m,\vece_2\leftarrow D_{\Z,\chi}^D,\vece_3\leftarrow D_{\Z,\chi}^k,\\
        &\vecz_{1}^\top\leftarrow\vecs^\top\matA+\vece_1^\top,\vecz_2^\top\leftarrow\vecs^\top\matP+\vece_2^\top,\vecz_3^\top\leftarrow \vecs^\top\matQ+\vece_3^\top,
        \end{align*}
    we cannot in general conclude the corresponding postcondition as in the pseudorandomness-version assumption, namely, \begin{align}\label{postcondition2}
    \begin{split}
(1^\lambda,\vecr_{\pub},\matA,\vecz_1^\top,\matP,\vecz_2^\top+\vecu_0^\top\matG,\matK)
    \overset{c}{\approx}(1^\lambda,\vecr_{\pub},\matA,\vecz_1^\top,\matP,\vecz_2^\top+\vecu_1^\top\matG,\matK),
    \end{split} 
    \end{align}
    where $\matK\leftarrow\matA_{\chi'}^{-1}(\matQ)$. To illustrate this, we consider a sampler $\mathcal{S}'$ that always outputs the all-zero matrix $\matV = \veczero_{d\times k}$. Informally, we then observe the following chain of distributions: \begin{align*}
    \begin{split}      &(1^\lambda,\vecr_{\pub},\matA,\vecz_1^\top,\matP,\vecz_2^\top+\vecu_0^\top\matG,\vecz_3^\top)\\
    \equiv\ &  (1^\lambda,\vecr_{\pub},\matA,\vecz_1^\top,\matP,\vecz_2^\top+\vecu_0^\top\matG,(\vecz_2^\top+\vecu_0^\top\matG)\matG^{-1}(\matV)+\vece_3^\top-\vece_2^\top\matG^{-1}(\matV)-\vecu_0^\top\matV)\\
    \overset{c}{\approx}\ & (1^\lambda,\vecr_{\pub},\matA,\vecdelta_1^\top,\matP,\vecdelta_2^\top,\vecdelta_2^\top\matG^{-1}(\matV)+\vece_3^\top-\vece_2^\top\matG^{-1}(\matV)-\vecu_0^\top\matV)\\
    \approx\ & (1^\lambda,\vecr_{\pub},\matA,\vecdelta_1^\top,\matP,\vecdelta_2^\top,\vecdelta_2^\top\matG^{-1}(\matV)+\vece_3^\top-\vece_2^\top\matG^{-1}(\matV)-\vecu_1^\top\matV)\\
    \overset{c}{\approx}\ &(1^\lambda,\vecr_{\pub},\matA,\vecz_1^\top,\matP,\vecz_2^\top+\vecu_1^\top\matG,(\vecz_2^\top+\vecu_1^\top\matG)\matG^{-1}(\matV)+\vece_3^\top-\vece_2^\top\matG^{-1}(\matV)-\vecu_1^\top\matV)\\
    \equiv\ &(1^\lambda,\vecr_{\pub},\matA,\vecz_1^\top,\matP,\vecz_2^\top+\vecu_1^\top\matG,\vecz_3^\top),
 \end{split}
    \end{align*} which implies the precondition \eqref{precondition2}. However, in the postcondition, the presence of the low-norm matrix $\matK$ satisfying $\matA\matK=\matP\matG^{-1}(\matV)$ yields a short basis of the $q$-ary lattice $\{\vecx\in \Z_q^m: \matA\vecx=\veczero\bmod{q}\}$ with high probability---a trapdoor for the matrix $\matA$. This trapdoor enables the adversary to distinguish between $\vecz_2^\top+\vecu_0^\top\matG$ and $\vecz_2^\top+\vecu_1^\top\matG$, in the postcondition, thereby breaking indistinguishability.

    In contrast, Assumption \ref{INDIPFE} avoids this issue by defining the matrix $\matQ_i$ as $$\matQ_i\leftarrow [\matB_i\mid\matP_i]\left[\begin{array}{c|c|c}
            \vecr_{i,1} & \cdots &\vecr_{i,k_i}  \\
            \matG^{-1}(\vecv_{i,1}) & \cdots &\matG^{-1}(\vecv_{i,k_i})
        \end{array}\right]\in \Z_q^{n\times k_i},$$
    Each column vector $\left[\begin{smallmatrix}
        \vecr_{i,j}\\
        \matG^{-1}(\vecv_{i,j})
    \end{smallmatrix}\right]$ is selected to be a nonzero vector and independent of the uniformly random matrix $[\matB_i\mid\matP_i]$. As a result, the marginal distribution of each column of $\matQ_i$ is uniform. Thus, the zeroizing attack, for example, by computing $$(\vecs_i^\top\matA_i+\vece_{1,i}^\top)\matK_i\approx\vecs_i^\top\matQ_i$$ does not yield a valid equation over the integers. This design aligns with the idea of \cite{Wee22} when introducing the evasive \LWE assumption.  
\end{remark}

\section{Multi-Authority Attribute-Based Inner Product Functional Encryption}
\label{sec:5}

In this section, we present and formalize the core definitions underlying our work. We begin by introducing the notion of multi-authority attribute-based (noisy) inner-product functional encryption ($\mannipfe$), which generalizes previous \NIPFE~\cite{Agr19,AP20} frameworks to support noise-tolerant decryption and multi-authority access control. We then introduce a relaxed ``evasive'' variant, called multi-authority attribute-based evasive inner-product functional encryption ($\maev$), inspired by the evasive \IPFE model of~\cite{HLL24}, which captures a weaker security notion more suitable for certain lattice-based instantiations. Both schemes are formalized in the static security model, with a focus on subset policies for simplicity and clarity.

\subsection[$\mannipfe$]{Multi-authority Attribute-based (Noisy) Inner-Product Functional Encryption ($\mannipfe$)}
We first introduce the concept of \emph{monotone access structures}~\cite{Bei96}, which play a fundamental role in access control mechanisms and secret-sharing schemes. The monotonicity property captures the idea that adding participants to an authorized subset preserves its authorization status. In this context, subset policies can be regarded as a specific instance of monotone access structures.

\begin{definition}[Access Structure,~\cite{Bei96}]
Let $S$ be a set and $2^S$ be the power set of $S$, i.e., the collection of all subsets of $S$. An \emph{access structure} on $S$ is a set $\mathbb{A}\subseteq 2^S\setminus\{\varnothing\}$, consisting of some non-empty subsets of $S$. A subset $A\in 2^{S}$ is called \emph{authorized} if $A\in\mathbb{A}$, and \emph{unauthorized} otherwise. 

An access structure is called \emph{monotone} if it satisfies the following condition: for all subsets $B,C\in 2^{S}$, if $B\in\mathbb{A}$ and $B\subseteq C$, then $C\in\mathbb{A}$. In other words, adding more elements to an authorized subset does not invalidate its authorization.
\end{definition}

In the following, we begin by presenting the definition of multi-authority attribute-based (noisy) inner-product functional encryption ($\mannipfe$). For simplicity, we adopt the \emph{small universe} setting, as described in \cite{DP23}, where each authority controls a single attribute, following the convention of \cite{RW15,DKW21a}. In this setting, an authority and its corresponding attribute can be considered equivalent to some extent. This definition naturally extends to a more general setting in which each authority can manage an arbitrary polynomial number of attributes (with respect to the security parameter), as outlined in \cite{RW15}. Our definition extends the standard $\maipfe$ model~(e.g., \cite{DP23}) by allowing the decryption to yield a noisy approximation of the inner product. This reflects more realistic applications where exact reconstruction of the inner product is not required and allows for a broader range of lattice-based constructions. Formally, we define the following.

\begin{definition}[$\mannipfe$,~\cite{AGT21,DP23}]\label{def1}
    Let $\lambda$ be the security parameter, and $n=n(\lambda),q=q(\lambda)$ be lattice parameters.  Let $\AU$ denote the universe of authority identifiers, and $\GID$ be the universe of global identifiers for users. Let $B_0=B_0(\lambda),B_1=B_1(\lambda)$ be bounding values. A $(B_0,B_1)$-\emph{multi-authority attribute-based (noisy) inner-product functional encryption scheme} for a class of policies $\P$ (described by a monotone access structure on a subset of $\AU$) over the vector space $\Z_q^n$ is defined as a tuple of efficient algorithms $\Pi_{\mannipfe}=(\globalset, \authset,\keygen,\enc,\dec)$. These algorithms proceed as follows: 
\begin{itemize}
    \item $\globalset(1^\lambda)\rightarrow \gp$: The global setup algorithm takes as input the security parameter $\lambda$, and outputs the global parameters $\gp$. We assume that $\gp$ specifies the description of $\AU$, $\GID$, $n$, and $q$.
    \item $\authset(\gp,\aid)\rightarrow (\pk_{\aid},\msk_{\aid})$: The authority (local) setup algorithm takes as input the global parameters $\gp$ and an authority identifier $\aid\in\AU$. It outputs a public key $\pk_{\aid}$ and a master secret key $\msk_{\aid}$.
    \item $\keygen(\gp,\pk_{\aid},\msk_{\aid},\gid,\vecv)\rightarrow \sk_{\aid,\gid,\vecv}$: The key generation algorithm takes as input the global parameters $\gp$, the public key $\pk_{\aid}$, the authority's master secret key $\msk_{\aid}$, a user identifier $\gid\in \GID$, and a key vector $\vecv\in\Z_q^n$. It outputs a secret key $\sk_{\aid,\gid,\vecv}$ associated with attribute $\aid$, user identifier $\gid$, and key vector $\vecv$.
    \item $\enc(\gp,\mathbb{A},\{\pk_{\aid}\}_{\aid\in A},\vecu)\rightarrow\ct$: The encryption algorithm takes as input the global parameters $\gp$, an access structure $\mathbb{A}\in \P$ on a set of authorities $A\subseteq\AU$, the set of public keys $\{\pk_{\aid}\}$ associated with authorities $\aid\in A$, and a plaintext vector $\vecu\in \Z_q^n$. It outputs a ciphertext $\ct$.
    \item $\dec(\gp,\{\sk_{\aid,\gid,\vecv}\}_{\aid\in A},\gid,\vecv,\ct)\rightarrow \Gamma$: The decryption algorithm takes as input the global parameters $\gp$, a collection of secret keys $\{\sk_{\aid,\gid,\vecv}\}$ associated with authorities $\aid\in A$, a user identifier $\gid\in \GID$, and a key vector $\vecv\in \Z_q^n$, and a ciphertext $\ct$. It outputs a value $\Gamma\in\Z_q$, corresponding to the (approximate) inner product of $\vecu$ and $\vecv$ or a symbol $\perp$ to indicate the failure of decryption. 
\end{itemize}
\end{definition}

\iitem{Approximate Correctness} A  $(B_0(\lambda),B_1(\lambda))$-$\manipfe$ scheme is said to be \emph{correct} if, for every security parameter $\lambda\in\N$, every plaintext vector $\vecu\in\Z_q^n$, every key vector $\vecv\in\Z_q^n$, every global identifier $\gid\in\GID$, every set of authorities $A\subseteq \AU$, every access structure $\mathbb{A}\in\mathcal{P}$ defined over $A$, and every subset of authorized authorities $
A'\in \mathbb{A}$, it holds that
$$\pr\left[\Gamma=\vecu^\top\vecv+e_0\left|
\begin{array}{c}\gp\leftarrow\globalset(1^{\lambda});\\
\forall \aid\in A,(\pk_{\aid},\msk_{\aid})\leftarrow\authset(\gp,\aid);\\
\forall\aid\in A', \sk_{\aid,\gid,\vecv}\leftarrow\keygen(\gp,\pk_{\aid},\msk_{\aid},\gid,\vecv);\\
\ct\leftarrow\enc(\gp,\mathbb{A},\{\pk_{\aid}\}_{\aid\in A},\vecu);\\
\Gamma\leftarrow\dec(\gp,\{\sk_{\aid,\gid,\vecv}\}_{\aid\in A'},\gid,\vecv,\ct);\\
e_0\in [-B_0(\lambda),B_0(\lambda)]
\end{array}\right.\right]\geq 1-\negl(\lambda).$$
Intuitively, we adopt an approximate correctness guarantee: if the decryption is authorized, the decryption result $\Gamma$ should recover the inner product $\vecu^\top \vecv$ up to a small additive noise term $e_0$, bounded in absolute value by $B_0(\lambda)$.

\iitem{Static Security} In this paper, we extend the static security model of \cite{AGT21,DP23} by incorporating a \emph{noisy} setting. This model is adapted from the security model for $\maabe$ in \cite{RW15} and for noisy $\IPFE$ in \cite{AP20}. This model is formalized as a game between a challenger and an adversary. The term \emph{static} refers to the restriction that the adversary must specify all of its queries at the beginning of the game.

For a security parameter $\lambda\in\N$, we define the static security game between a challenger and an adversary $\mathcal{A}$ for an $\mannipfe$ scheme as follows:

\dotitem{Global Setup} The challenger runs $\gp\leftarrow\globalset(1^\lambda)$ to generate the global parameters and sends $\gp$ to the adversary $\A$.

\dotitem{Adversary's Queries} The adversary specifies the following queries:
\begin{itemize}
    \item A set of corrupt authorities $\mathcal{C}\subseteq \AU$, along with their respective public keys $\pk_{\aid}$ for each corrupt authority $\aid\in \mathcal{C}$. 
    \item A set of non-corrupt authorities $\mathcal{N}\subseteq \AU$, where $\mathcal{N}\cap\mathcal{C}=\varnothing$.
    
    \item A set of secret-key queries $\mathcal{Q}=\{(\gid,A,\vecv)\}$  where each query specifies a global identifier $\gid\in\GID$, a subset of non-corrupt authorities $A\subseteq \mathcal{N}$, and a key vector $\vecv\in\Z_q^n$. We assume without loss of generality that each pair 
     $(\gid,\vecv)$ appears at most once in the query set $\mathcal{Q}$, even if associated with different authority subsets $A$. This avoids the adversary combining keys from multiple key responses to enable decryption that would otherwise be unauthorized.
    
    \item A pair of plaintext vectors $\vecu_0,\vecu_1\in\Z_q^n$, a set of authorities $A^*\subseteq \mathcal{C}\cup\mathcal{N}$, and an access structure $\mathbb{A}\in\mathcal{P}$ defined over $A^*$.
\end{itemize}

\dotitem{Challenger's Replies} The challenger first flips a fair coin $\beta\rand\{0,1\}$ and generates key pairs $(\pk_{\aid},\msk_{\aid})\leftarrow \authset(\gp,\aid)$ for each $\aid\in \mathcal{N}$. It then responds to adversary $\mathcal{A}$ as follows:
\begin{itemize}
    \item The public keys $\pk_{\aid}$ for each non-corrupt authority $\aid\in\mathcal{N}$.
    \item The secret keys $\sk_{\aid,\gid,\vecv}\leftarrow\keygen(\gp,\pk_{\aid},\msk_{\aid},\gid,\vecv)$ for each secret-key query $(\gid,A,\vecv)$ and each  $\aid\in A$.
    \item The challenge ciphertext $\ct_{\beta}\leftarrow\enc(\gp,\mathbb{A},\{\pk_{\aid}\}_{\aid\in A^*},\vecu_\beta)$.
\end{itemize}

\dotitem{Guess Phase} The adversary outputs a guess $\beta'\in\{0,1\}$ for the value of $\beta$.

The \emph{advantage} of the adversary $\mathcal{A}$ in this game is defined by $$\Adv_{\mathcal{A}}(\lambda):=|\pr[\beta=\beta']-1/2|.$$

Unlike the static security model of $\maabe$ (e.g., \cite{DKW21a,WWW22}), where the adversary is restricted to making only \emph{unauthorized} secret-key queries, the static security model for $\manipfe$ schemes---as considered in \cite{AGT21,DP23}---allows the adversary to ask secret-key queries that may potentially decrypt the challenge ciphertext and derive an approximate inner product. To account for this broader adversarial capability, we impose the following limitations on the adversary in the game of a $(B_0(\lambda),B_1(\lambda))$-$\manipfe$ scheme:

\begin{definition}[Admissible Adversary]\label{admiss}
   An adversary $\mathcal{A}$ is said to be admissible for the $(B_0,B_1)$-$\manipfe$ security game defined above if, the set $A^*\cap\mathcal{C}$ is not authorized under the access structure $\mathbb{A}$, i.e., $A^*\cap\mathcal{C}\notin \mathbb{A}$, and for every secret-key query $(\gid,A,\vecv)\in\mathcal{Q}$, at least one of the following conditions holds, ensuring either unauthorized access or functional indistinguishability:
    \begin{itemize}
        \item \emph{\textbf{Unauthorized access}}: $(A\cup\mathcal{C})\cap A^*\notin \mathbb{A}$.
        \item \emph{\textbf{Functional indistinguishability}}: The vector $\vecv$ satisfies $\|(\vecu_0-\vecu_1)^\top\vecv\|\leq B_1$. (This ensures that the adversary cannot gain useful distinguishing information from the inner product of the challenge plaintexts.)
    \end{itemize}
\end{definition}

\begin{definition}[Static Security for an $\manipfe$ scheme]
An $(B_0,B_1)$-$\manipfe$ scheme is said to be \emph{statically secure} if, for every efficient admissible adversary $\mathcal{A}$, there exists a negligible function $\negl(\cdot)$ such that, for all $\lambda\in\N$, the advantage of $\mathcal{A}$ in the above security game is at most $\negl(\lambda)$.
\end{definition}

\begin{remark}[Noiseless $\maipfe$]
    For a correct and statically secure $(B_0,B_1)$-$\manipfe$, it is natural to require that $B_0\geq B_1$ holds. Otherwise, an admissible adversary could distinguish between the two challenge ciphertexts by querying an authorized secret key associated with a key vector $\vecv$, for which $(\vecu_0-\vecu_1)^\top\vecv=B_1>B_0$. In this scenario, the difference between the two inner products would not be sufficiently obscured by the lower noise level guaranteed by the decryption process, thereby compromising security.
    
    In particular, when $B_0=B_1=0$, the $(B_0,B_1)$-$\manipfe$ degenerates into a \emph{noiseless} $\maipfe$ scheme, or simply an $\MAABIPFE$ scheme for brevity. Specifically, in this case, the decryption process using authorized secret keys is expected to return the \emph{exact} inner product (except with negligible probability). Meanwhile, in the security game, an \emph{admissible} adversary is restricted to submitting secret-key queries that are either unauthorized or guaranteed to yield the \emph{same} inner product when decrypting both challenge ciphertexts. In this case, this definition coincides with the standard definition of $\maipfe$ as formulated in~\cite{AGT21,DP23}. 
\end{remark}

\begin{remark}[Static Security in the Random Oracle Model]
   In this work, we analyze the static security of our $\mannipfe$ construction in the \emph{random oracle model}~\cite{BR93}, following the approach of \cite{RW15,DKW21a,DKW21b,WWW22,DP23}. In this setting, the hash function $\H$ (modeled as a random oracle programmed by the challenger) is specified as part of the global parameters, and all parties in the scheme have access to this function. 
   
   During the static security game in the random oracle model, the adversary is allowed to submit the random oracle queries as part of its initial queries. It is further allowed to make additional random oracle queries during the phase in which the challenger provides its responses.
The challenger must answer all such queries consistently throughout the interaction.
\end{remark}

\iitem{$\mannipfe$ for subset policies} 
In this paper, we primarily focus on constructing an $\mannipfe$ scheme for the class of subset policies. In this setting, the ciphertext is encrypted under a set of authorities (attributes) $A\subseteq\AU$, which uniquely determines the access structure under the subset policy setting. Consequently, the access structure 
$\mathbb{A}$ does not need to be explicitly specified in the encryption algorithm, as the policy is implicitly defined by the set $A$. The decryption algorithm succeeds only if a user possesses secret keys associated with a set of authorities $B$ such that $A\subseteq B$. 

The flexibility of the subset policy model provides an efficient framework for expressing access control rules, making it well-suited for a wide range of applications that require dynamic access restrictions based on user roles or permissions. By combining the definition of the $\maabe$ scheme for subset policies, as introduced in \cite{WWW22}, with the definition of $\maipfe$ scheme in \cite{AGT21,DP23}, we formalize the following definition.
 
\begin{definition}[$\mannipfe$ for subset policies]
Let $\lambda$ be a security parameter and $n=n(\lambda),q=q(\lambda)$ be lattice parameters. Let $\AU$ be the universe of authority identifiers. Define the class of subset policies $\P$ as $$\P=\{\mathbb{A}: \mathbb{A}=\{B: B\supseteq A\} \text{ where } A\subseteq \AU\}.$$ An access structure $\mathbb{A}$ for a subset policy is determined by the set $A\subseteq \AU$. For an $(B_0(\lambda),B_1(\lambda))$-$\mannipfe$ scheme $\Pi_{\mannipfe}=\{\globalset,\authset,\keygen,\enc,\dec\}$ for subset policies over the vector space $\Z_q^n$, we can omit the explicit specification of the access structure $\mathbb{A}$ in the encryption algorithm, since it is implicitly defined by the set of public keys $\{\pk_{\aid}\}_{\aid\in A}$. Precisely, the encryption algorithm $\enc$ in Definition \ref{def1} can be simplified as follows (with other algorithms unchanged):
\begin{itemize}
    \item $\enc(\gp,\{\pk_{\aid}\}_{\aid\in A},\vecu)\rightarrow \ct$: The encryption algorithm takes as input the global parameters $\gp$, the set of public keys $\{\pk_{\aid}\}$ corresponding to authorities $\aid\in A$, and a plaintext vector $\vecu\in\Z_q^n$. It outputs a ciphertext $\ct$.
\end{itemize}
\end{definition}

In the subset policy setting, admissibility (Definition~\ref{admiss}) of an adversary imposes specific structural constraints on secret-key queries. We now make this more explicit by classifying such queries into two types. Let $\mathcal{Q}=\{
(\gid,A,\vecv)\}$ denote the set of secret-key queries and let $A^*$ denote the set of attributes related to the challenging ciphertext. According to Definition \ref{admiss}, an admissible adversary is additionally required to satisfy $A^*\cap\mathcal{C}\subsetneqq A^*$, and each query $(\gid,A,\vecv)\in\mathcal{Q}$ must satisfy at least one of the following conditions:
\begin{itemize}
     \item  \textbf{Unauthorized access}: $(A\cup \mathcal{C})\cap A^*\subsetneqq A^*$.
    \item \textbf{Functional indishtinguishability}: $\|(\vecu_1-\vecu_0)^\top\vecv\|\leq B_1$.
\end{itemize}
To facilitate the analysis in our security proof, we classify secret-key queries into the following types:
\begin{itemize}
    \item \textbf{Type I secret-key query}: A query satisfying $(A\cup \mathcal{C})\cap A^*\subsetneqq A^*$.
    \item \textbf{Type II secret-key query}: A query satisfying $(A\cup\mathcal{C})\cap A^*=A^*$ and $\|(\vecu_1-\vecu_0)^\top\vecv\|\leq B_1$.
\end{itemize}
It follows directly that each secret-key query must belong to one of these two types.

\begin{remark}
    As mentioned in \cite{WWW22}, the $\maabe$ scheme for subset policies implies an $\maabe$ scheme for access structure supporting access structures defined by a polynomial-size conjunction or a \DNF formula. Similarly, our construction of the $\mannipfe$ scheme for subset policies can also be extended to an $\mannipfe$ scheme for access structure decided by conjunction or a \DNF formula by following an analogous argument.
\end{remark}

\subsection[$\maev$]{Multi-authority Attribute-based Evasive Inner-product Functional Encryption ($\maev$)}
Inspired by the evasive \IPFE scheme introduced in~\cite{HLL24}, we define the notion of a \emph{multi-authority attribute-based evasive inner-product functional encryption} ($\maev$) scheme by relaxing the security requirement of the $\mannipfe$ scheme described above. In \cite{HLL24}, the evasive \IPFE captures a ``generic-model view'' analogous to that of evasive \LWE. The authors consider the \emph{static} (also known as \emph{very selective}~\cite{HLL24}) security model with respect to samplers producing pseudorandom noisy inner-product outputs. 

Following this approach, we consider a weakened static security game where the adversary is no longer required to fully specify the entire Type II secret-key queries or the challenge plaintexts. Instead, these components are jointly determined by a predefined sampler and the adversary. We formalize the definition of an $(B_0,\chi_s)$-$\maev$ scheme accordingly, adopting most of the notations following the convention of \cite{HLL24}. The syntax of the scheme and its approximate correctness guarantee follow the definition of the $(B_0, B_1)$-$\mannipfe$ scheme given in Definition~\ref{def1}, where the parameter $B_0$ retains the same role of bounding the decryption noise. We omit restating this part in the $\maev$ definition below for clarity. We formalize the static security model for $\maev$ as follows.

\iitem{Static security} Let $\samp=(\samp_{\vecv},\samp_{\vecu})$ be two algorithms with the following syntax.

\begin{itemize}
    \item $\samp_{\vecv}(1^\lambda;\vecr_{\pub})$: The algorithm takes as input the security parameter $\lambda$ and public randomness $\vecr_{\pub}$, and outputs a sequence of vectors $\vecv_1',\vecv_2',\ldots,\vecv_{Q'}'\in\Z_q^n$.

    \item $\samp_{\vecu}(1^\lambda,\vecr_{\pub};\vecr_{\pri})$: The algorithm takes as input the security parameter $\lambda$, public randomness $\vecr_{\pub}$, and private randomness $\vecr_{\pri}$,  and outputs a single vector $\vecu\in\Z_q^n$.
\end{itemize}

For a security parameter $\lambda\in\N$, we define the static security game between a challenger and an (admissible) adversary $\mathcal{A}$ for an $\maev$ scheme, parameterized by the sampler $\samp=(\samp_{\vecv},\samp_{\vecu})$, as follows:

\dotitem{Global Setup} The challenger runs $\gp\leftarrow\globalset(1^\lambda)$ to generate the global parameters and sends $\gp$ to the adversary $\A$.

\dotitem{Adversary's Queries} The adversary specifies the following queries:
\begin{itemize}
    \item A set $\mathcal{C}\subseteq \AU$ of corrupt authorities, along with their respective public keys $\pk_{\aid}$ for each corrupt authority $\aid\in \mathcal{C}$. 
    \item A set $\mathcal{N}\subseteq \AU$ of non-corrupt authorities, where $\mathcal{N}\cap\mathcal{C}=\varnothing$.

    \item A set of authorities $A^*\subseteq \mathcal{C}\cup\mathcal{N}$ such that $A^*\cap\mathcal{C}\subsetneqq A^*$.
    
    \item A set of Type I secret-key queries $\mathcal{Q}=\{(\gid,A,\vecv)\}$, where $A\subseteq \mathcal{N}$ and $(A\cup \mathcal{C})\cap A^*\subsetneqq A^*$. 

    \item A partial set of Type II secret-key queries $\mathcal{Q}_{\partset}'=\{(\gid',A')\}$, where $A'\subseteq \mathcal{N}$ and $(A'\cup\mathcal{C})\cap A^*=A^*$.
    
\end{itemize}

\dotitem{Challenger's Replies}
    Then the challenger proceeds as follows:
    \begin{enumerate}
    \item The challenger first flips a fair coin $\beta\rand\{0,1\}$ and generates key pairs $(\pk_{\aid},\msk_{\aid})\leftarrow \authset(\gp,\aid)$ for each $\aid\in \mathcal{N}$.
    \item The challenger samples $\vecr_{\pub}\rand\{0,1\}^{\kappa_1},\vecr_{\pri}\rand\{0,1\}^{\kappa_2}$, where $\kappa_1,\kappa_2$ denotes the upper bounds of the random bits used by $\samp_{\vecv}$ and $\samp_{\vecu}$, respectively.  Then it computes
    \begin{align*}
(\vecv_1',\vecv_2',\ldots,\vecv_{Q'}')\leftarrow\samp_{\vecv}(1^\lambda;\vecr_{\pub}),\\
\vecu_0\leftarrow \samp_{\vecu}(1^{\lambda};\vecr_{\pub},\vecr_{\pri}),\vecu_1\rand\Z_q^n. 
    \end{align*}
Next, the challenger builds the full set of Type II queries
$$\mathcal{Q}'=\{(\gid_1',A_1',\vecv_1'),\ldots,(\gid_{Q'}',A_{Q'}',\vecv_{Q'}')\}.$$

Before detailing the responses, we make the following assumptions without loss of generality:
\begin{itemize}
    \item Each pair $(\gid,\vecv)$ appears at most once in $\mathcal{Q}$, for the same reason as in $\mannipfe$. Likewise, each pair $(\gid',\vecv')$ appears at most once in $\mathcal{Q}'$, since two queries $(\gid',A_1',\vecv')$ and $(\gid',A_2',\vecv')$ can, without loss of generality, be replaced by a single query $(\gid',A_1'\cup A_2',\vecv')$.
     \item Each pair $(\gid_j',\vecv_j')$ in $\mathcal{Q}_{\partset}'$ does not appear in the Type I secret-key queries set $\mathcal{Q}$. If a pair $(\gid_j',\vecv_j')$ appears in both query sets, it can be safely removed from $\mathcal{Q}$. This is because the challenger will already return the secret keys associated with all authorities in $A^*\cap\mathcal{N}$---these cover all the authorities in any Type I query involving $(\gid_j',\vecv_j')$. As such, the Type I secret-key would be redundant. 
    
    \item The partial set of Type II queries $\mathcal{Q}_{\partset}'$  is assumed to contain exactly $Q'$ elements, matching the number of vectors output by the sampler $\samp_{\vecv}$. If necessary, the set $\mathcal{Q}_{\partset}'$ can be padded or truncated to ensure that its size matches $Q'$, without affecting the adversary’s view or the underlying distribution.
    \end{itemize}

\item It then responds to the adversary $\mathcal{A}$ as follows:
\begin{itemize}
    \item The public keys $\pk_{\aid}$ for each non-corrupt authorities $\aid\in\mathcal{N}$.
    \item The secret keys $\sk_{\aid,\gid,\vecv}\leftarrow\keygen(\gp,\pk_{\aid},\msk_{\aid},\gid,\vecv)$ for each Type I secret-key query $(\gid, A,\vecv)\in \mathcal{Q}$ and each $\aid\in A$.
    \item The secret keys $\sk_{\aid,\gid',\vecv'}\leftarrow\keygen(\gp,\pk_{\aid},\msk_{\aid},\gid',\vecv')$ for each Type II secret-key query $(\gid',A',\vecv')\in\mathcal{Q}'$ and each $\aid\in A'$.
    \item The challenge ciphertext $\ct_{\beta}\leftarrow\enc(\gp,\{\pk_{\aid}\}_{\aid\in A^*},\vecu_\beta)$.
    \item The public randomness $\vecr_{\pub}$.
\end{itemize}
\end{enumerate}

\dotitem{Guess phase} The adversary outputs a guess $\beta'\in\{0,1\}$ for the value of $\beta$.

The \emph{advantage} of the adversary $\mathcal{A}$ in distinguishing $\beta$ is defined as $$\Adv_{\mathcal{A}}^{\samp}(\lambda):=|\pr[\beta=\beta']-1/2|.$$ 

\begin{definition}[Static Security for an $\maev$ scheme]
    We say that the sampling algorithms $\samp=(\samp_{\vecv},\samp_{\vecu})$ in the game above produces \emph{pseudorandom noisy inner products} with noise parameter $\sigma$, if $$(1^\lambda,\vecr_{\pub},\{\vecu_0^\top\vecv_i'+e_i\}_{i\in[Q']})\overset{c}{\approx}(1^\lambda,\vecr_{\pub},\{\delta_i\}_{i\in[Q']}),$$ where $e_i\leftarrow D_{\Z,\sigma},\delta_i\rand\Z_q$ for each $i\in [Q']$. An $\maev$ scheme is said to be $\samp$-\emph{secure} if for every efficient admissible adversary $\mathcal{A}$, the advantage of $\Adv_{\mathcal{A}}^{\samp}(\lambda)$ is negligible in $\lambda$. Furthermore, an $(B_0,\chi_s)$-$\maev$ scheme is said to be \emph{statically secure} if it is $\samp$-secure for all efficient $\samp$ that produce pseudorandom noisy inner products with noise parameter $\sigma\leq \chi_s$. 
\end{definition}

\begin{remark}
    As justified in \cite{HLL24}, to ensure indistinguishability in the static security game, it is necessary to make the sampling process of $\vecu_0$ private-coin, meaning the distinguisher cannot trivially use randomness to obtain $\vecu_0$. Otherwise, the distinguisher may obtain $\vecu_0$ and $\vecv_i'$ for $i\in [Q']$, and compute the corresponding inner products $\vecu_0^\top\vecv_i'$. This would allow distinguishing the challenge ciphertext and thus break the security of the scheme.
\end{remark}

\section{$\maev$ Scheme from $\EVIPFE$ Assumption (in the Random Oracle Model)}\label{sec:6}

In this section, we present our construction of the $\maev$ scheme for subset policies in the random oracle model. The construction follows a structure similar to that in \cite{DKW21a,WWW22,HLL24}. Later, we prove the static security of this scheme based on the \LWE assumption and $\EVIPFE$ assumption, analogous to the approach introduced in \cite{WWW22, HLL24}.

\begin{construction}[$\maev$ for subset policies in the random oracle model]\label{con1}
    Let $\lambda$ be the security parameter. Let $n$ and $q$ be lattice parameters, and define $m = n\lceil\log q\rceil$. Let $m'$, $\chi$, and $\chi'$ be additional lattice parameters. 
    Let $\AU=\{0,1\}^{\lambda}$ denote the universe of authority identifiers, and let $\GID=\{0,1\}^{\lambda}$ denote the universe of global identifiers for users. Let $\H: \GID\times\Z_q^n\rightarrow\Z_q^{m'}$ be a hash function, modeled as a random oracle, similarly to Construction 5 in \cite{WWW22}. More precisely, 
    \begin{itemize}
        \item The outputs of the random oracle $\H$ are modeled as following a discrete Gaussian distribution with parameter $\chi'$. Specifically, on each input query $(\gid,\vecv)\in \GID\times\Z_q^n$, the output $\H(\gid,\vecv)$ is distributed according to discrete Gaussian $D_{\Z,\chi'}^{m'}$.
        
        \item As noted in \cite{WWW22}, such a hash function $\H$ can be instantiated using a standard random oracle $\H^{\prime}: \GID\times \Z_q^n\rightarrow \{0,1\}^{\lambda m'}$, where the outputs of $\H'(\gid,\vecv)$ are uniformly distributed over $\{0,1\}^{\lambda m'}$. According to Lemma 5.7 and Remark 5.9 of \cite{WWW22}, the outputs of $\H'$ can then be efficiently transformed into the desired discrete Gaussian distribution above using inversion sampling techniques.
    \end{itemize}
    A $(B_0,\chi_s)$-$\maev$ scheme over $\Z_q^n$ for subset policies consists of a tuple of efficient algorithms $\Pi_{\maev}=(\globalset, \authset,\keygen,\enc,\dec)$. These algorithms proceed as follows: 
\begin{itemize}
    \item $\globalset(1^\lambda)\rightarrow \gp$: The global setup algorithm takes as input the security parameter $\lambda$, and outputs the global parameters $\gp=(\lambda,n,m,m',q,\chi,\chi',\H)$.
    
    \item $\authset(\gp,\aid)\rightarrow (\pk_{\aid},\msk_{\aid})$: The authority setup algorithm takes as input the global parameters $\gp$ and an authority identifier $\aid\in\AU$. It samples $(\matA_{\aid},\td_{\aid})\leftarrow\trapgen(1^n,1^m,q),\matB_{\aid}\rand\Z_q^{n\times m'},\matP_{\aid}\rand\Z_q^{n\times m}$. 
    It outputs a public key $\pk_{\aid}=(\matA_{\aid},\matB_{\aid},\matP_{\aid})$ and a master secret key $\msk_{\aid}=\td_{\aid}$.
    
    \item $\keygen(\gp,\pk_\aid,\msk_\aid,\gid,\vecv)\rightarrow \sk_{\aid,\gid,\vecv}$: The key generation algorithm takes as input the global parameters $\gp$, the public key $\pk_{\aid}=(\matA_\aid,\matB_{\aid},\matP_{\aid})$, the authority's master secret key $\msk_{\aid}=\td_\aid$, the user identifier $\gid\in \GID$, the key vector $\vecv\in\Z_q^n$, the key generation algorithm computes $\vecr\leftarrow \H(\gid,\vecv)$ and uses the trapdoor $\td_\aid$ for $\matA_\aid$ to sample $\veck\leftarrow\samplepre(\matA_{\aid},\td_{\aid},\matP_\aid\matG^{-1}(\vecv)+\matB_\aid\vecr,\chi)$. It outputs a secret key $\sk_{\aid,\gid,\vecv}=\veck$.
    
    \item $\enc(\gp,\{\pk_{\aid}\}_{\aid\in A},\vecu)\rightarrow\ct$: The encryption algorithm takes as input the global parameters $\gp$, a set of public keys $\pk_{\aid}=(\matA_{\aid},\matB_{\aid},\matP_{\aid})$ associated with authorities $A\subseteq \AU$, a plaintext vector $\vecu\in \Z_q^n$. The encryption algorithm samples $\vecs_{\aid}\rand\Z_q^n, \vece_{1,\aid}\leftarrow D_{\Z,\chi}^{m}$ for each $\aid\in A$, $\vece_2\leftarrow D_{\Z,\chi}^{m'}$, and $\vece_3\leftarrow D_{\Z,\chi}^{m}$. It outputs a ciphertext $\ct\in\Z_q^m\times\Z_q^{m'}\times\Z_q^{m}$, where $$\ct=\left(\left\{\vecs_{\aid}^\top\matA_{\aid}+\vece_{1,\aid}^\top\right\}_{\aid\in A},\sum_{\aid\in A}\vecs_{\aid}^\top\matB_{\aid}+\vece_2^{\top},\sum_{\aid\in A}\vecs_{\aid}^{\top}\matP_{\aid}+\vece_3^{\top}+\vecu^\top\matG\right).$$
    
    \item $\dec(\gp,\{\sk_{\aid,\gid,\vecv}\}_{\aid\in A},\gid,\vecv,\ct)\rightarrow \Gamma$: The decryption algorithm takes as input the global parameters $\gp$, a collection of secret keys $\sk_{\aid,\gid,\vecv}=\veck_{\aid,\gid,\vecv}$ associated with 
    authorities $\aid\in A$, a user identifier $\gid$, a key vector $\vecv$, and a ciphertext $\ct=(\{\vecc_{1,\aid}^\top\}_{\aid\in A},\vecc_2^\top,\vecc_3^\top)$. The decryption algorithm first computes $\vecr\leftarrow \H(\gid, \vecv)$ and outputs 
$$\Gamma=\vecc_3^{\top}\matG^{-1}(\vecv)+\vecc_2^\top\vecr-\sum_{\aid\in A}\vecc_{1,\aid}^\top\veck_{\aid,\gid,\vecv}.$$
\end{itemize}
\end{construction}

\subsection{Correctness}
\begin{theorem}[Correctness]\label{correct1}
    Let $L$ be an upper bound on the number of attributes associated with a ciphertext.  Let $\chi_0=\sqrt{n\log q}\cdot \omega(\sqrt{\log n})$ be a polynomial such that Lemma \ref{preimage} holds. Suppose that the lattice parameters $n,q,\chi$ are such that $\chi\geq \chi_0(n,q)$. Then the scheme $\Pi_{\maev}$ in Construction \ref{con1} is \emph{correct} as a $(B_0,\chi_s)$-$\maev$ scheme, where the parameter $B_0=\sqrt{\lambda}\chi m+\lambda \chi\chi' m'+\lambda \chi^2 mL.$
\end{theorem}

\begin{proof}
    Let $\vecu \in \Z_q^n$ be any plaintext, $\vecv \in \Z_q^n$ a key vector, $\gid \in \GID$ a user identifier, and $A$ a set of authorities. First, sample the global parameters $\gp\leftarrow\globalset(1^{\lambda})$, the authority key pairs $(\pk_{\aid},\msk_{\aid})\leftarrow\authset(\gp,\aid)$ for each $\aid\in A$, the secret keys $\sk_{\aid,\gid,\vecv}\leftarrow\keygen(\gp,\pk_{\aid},\msk_{\aid},\gid,\vecv)$ for each $\aid\in A$, and the ciphertext $\ct\leftarrow\enc(\gp,\{\pk_{\aid}\}_{\aid\in A},\vecu)$. 
    
    In the following, we verify the correctness by expanding the computation of the decryption process $\dec(\gp,\{\sk_{\aid,\gid,\vecv}\}_{\aid\in A},\gid,\vecv,\ct)$ in detail:
    \begin{itemize}
        \item The global parameters $\gp=(\lambda,n,m,m',q,\chi,\chi',\H)$ consist of the security parameter, the lattice parameters, and the description of a hash function $\H:\GID\times\Z_q^n\rightarrow\Z_q^{m'}$.
        \item For the plaintext vector $\vecu\in \Z_q^n$, the ciphertext  associated with authorities $\{\aid\}_{\aid\in A}$ is constructed as $\left(\left\{\vecc_{1,\aid}^\top\right\}_{\aid\in A},\vecc_2^\top,\vecc_3^{\top}\right)$, where $$\vecc_{1,\aid}^\top=\vecs_{\aid}^\top\matA_{\aid}+\vece_{1,\aid}^\top,\vecc_{2}^\top=\sum_{\aid\in A}\vecs_{\aid}^\top\matB_{\aid}+\vece_2^\top,\vecc_{3}^\top=\sum_{\aid\in A}\vecs_{\aid}^\top\matP_{\aid}+\vece_3^\top+\vecu^\top\matG,$$
        where $(\matA_{\aid},\matB_{\aid},\matP_{\aid})$ is the public key associated with the authority $\aid\in A$.
        
        \item Each secret key is generated as $\sk_{\aid,\gid,\vecv}=\veck_{\aid,\gid,\vecv}\leftarrow\samplepre(\matA_{\aid},\td_{\aid},\matP_\aid\matG^{-1}(\vecv)+\matB_\aid\vecr,\chi)$, where $\vecr\leftarrow\H(\gid,\vecv)$.    By Lemma \ref{trapdoor}, $\veck_{\aid,\gid,\vecv}$ is distributed according to $D_{\Z,\chi}^m$ conditioned on $$\matA_{\aid}\veck_{\aid,\gid,\vecv}=\matP_{\aid}\matG^{-1}(\vecv)+\matB_{\aid}\vecr.$$ By Lemma \ref{truncated}, it holds with overwhelming probability that $\|\veck_{\aid,\gid,\vecv}\|\leq \sqrt{\lambda}\chi$ and $\|\vecr\|\leq \sqrt{\lambda}\chi'$.
       
        \item Using the secret key $\veck_{\aid,\gid,\vecv}$, the decryption algorithm computes $$\vecc_{1,\aid}^\top\veck_{\aid,\gid,\vecv}=\vecs_{\aid}^\top\matA_{\aid}\veck_{\aid,\gid,\vecv}+\vece_{1,\aid}^\top\veck_{\aid,\gid,\vecv}=\vecs_{\aid}^\top\matP_{\aid}\matG^{-1}(\vecv)+\vecs_{\aid}^{\top}\matB_{\aid}\vecr+\vece_{1,\aid}^\top\veck_{\aid,\gid,\vecv}$$ for each $\aid\in A$.  Substituting the above expression into the decryption formula, we obtain: $$\vecc_3^{\top}\matG^{-1}(\vecv)+\vecc_2^\top\vecr-\sum_{\aid\in A}\vecc_{1,\aid}^\top\veck_{\aid,\gid,\vecv}=\vecu^\top\vecv+\vece_3^\top\matG^{-1}(\vecv)+\vece_2^\top\vecr-\sum_{\aid\in A}\vece_{1,\aid}^\top\veck_{\aid,\gid,\vecv},$$ where the error term is $$\tilde{e}:=\vece_3^\top\matG^{-1}(\vecv)+\vece_2^\top\vecr-\sum_{\aid\in A}\vece_{1,\aid}^\top\veck_{\aid,\gid,\vecv}.$$ We conclude that the total error term $\tilde{e}$ satisfies $|\tilde{e}|\leq B_0$, as required.

        \item To bound $|\tilde{e}|$, we analyze each of its components: By Lemma \ref{truncated}, the following bounds hold with overwhelming probability: \begin{align*}
            &\|\vece_{1,\aid}\|\leq \sqrt{\lambda}\chi, \quad \text{for each } \aid\in A, \\
            &\|\vece_2\|\leq  \sqrt{\lambda}\chi,\quad \|\vece_3\|\leq  \sqrt{\lambda}\chi.
        \end{align*} Using these bounds, we have
        \begin{align*}
        \|\vece_3^\top\matG^{-1}(\vecv)\|&\leq m \sqrt{\lambda}\chi,\\
        \|\vece_2^\top\vecr\|&\leq m'\lambda\chi\chi',\\
\|\vece_{1,\aid}^\top\veck_{\aid,\gid,\vecv}\|&\leq m\lambda\chi^2.
        \end{align*}
Combining these results, we obtain the bound for the error term $\tilde{e}$:$$\|\tilde{e}\|\leq \sqrt{\lambda}\chi m+\lambda \chi\chi' m'+\lambda \chi^2 m\ell.$$ Thus, the scheme satisfies correctness, as claimed.
\end{itemize}
\end{proof}

\subsection{Security}
\begin{theorem}\label{security1}
 Let $\chi_0(n,q)=\sqrt{n\log q}\cdot\omega(\sqrt{\log n})$ be a polynomial such that Lemma \ref{preimage} holds. Let $Q_0$ be the upper bound on the total number of secret-key queries (including those in the partial set) submitted by the adversary. Suppose that the following parameter conditions hold:
\begin{itemize}
    \item $\chi'=\Omega(\sqrt{n\log q})$.
    \item Let $\chi$ be an error distribution parameter such that $\chi>\chi_0$ and $\chi\geq \lambda^{\omega(1)}\cdot (\sqrt{\lambda}(m+\ell)\chi_s+\lambda m\chi'\chi_s)$ for all $\ell=\poly(\lambda)$, where $\chi_s$ is a noise parameter such that  $\LWE_{n,m_1,q,\chi_s}$ assumption holds for some $m_1=\poly(m,m',Q_0)$.
    
    \item The assumption $\EVIPFE_{n,m,m',q,\chi,\chi}$ holds.
\end{itemize}
   
   Then Construction \ref{con1} is \emph{statically secure} as a $(B_0,\chi_s)$-$\EVIPFE$ scheme. 
\end{theorem}

\begin{proof}
    We prove the \emph{static security} of Construction \ref{con1} by defining a sequence of hybrid experiments. We begin by fixing a public sampling algorithm tuple $\samp=(\samp_{\vecv},\samp_{\vecu})$ producing pseudorandom noisy inner products with noise parameter $\chi_s$.

 \iitem{Game $\H_0$} This experiment corresponds to the real static security game, in which the challenger encrypts the plaintext $\vecu$ sampled from $\samp_{\vecu}$, exactly as specified in Construction \ref{con1}. The experiment proceeds as follows. 
 
 At the beginning of the experiment, the adversary specifies the following queries:
\begin{itemize}
    \item A set of corrupt authorities $\mathcal{C}\subseteq\AU$, along with their public keys: $\pk_{\aid}=(\matA_{\aid},\matB_{\aid},\matP_{\aid})$ for each $\aid\in \mathcal{C}$.
    
    \item A set of non-corrupt authorities $\mathcal{N}\subseteq \AU$, satisfying $\mathcal{N}\cap\mathcal{C}=\varnothing$.
    
    \item A set of authorities $A^*\subseteq \mathcal{C}\cup\mathcal{N}$, satisfying $(A^*\cap\mathcal{C})\subsetneqq A^*$.
   
    \item A set of Type I secret-key queries $\mathcal{Q}=\{(\gid,A,\vecv)\}$, where $A\subseteq \mathcal{N}$ and $(A\cup \mathcal{C})\cap A^*\subsetneqq A^*$.

    \item A partial set of Type II secret-key queries $\mathcal{Q}_{\partset}'=\{(\gid',A')\}$, where $A'\subseteq \mathcal{N}$ and $(A'\cup\mathcal{C})\cap A^*=A^*$.
\end{itemize}

To simulate the random oracle, the challenger initializes an empty table $\T:\GID\times\Z_q^n\rightarrow\Z_q^{m'}$, which will be used to consistently store and respond to all random oracle queries throughout the experiment. Then the challenger with respect to samplers $\samp=(\samp_{\vecv},\samp_{\vecu})$ processes the adversary's queries as follows.
\begin{itemize}
 \item \textbf{Public keys for non-corrupt authorities:} For each non-corrupt authority $\aid\in\mathcal{N}$, the challenger samples $(\matA_{\aid},\td_{\aid})\leftarrow \trapgen(1^n,1^m,q),\matB_{\aid}\rand\Z_q^{n\times m'}$, and $\matP_{\aid}\rand\Z_q^{n\times m}$. The public key associated with $\aid$ is then set as $\pk_{\aid}\leftarrow(\matA_{\aid},\matB_{\aid},\matP_{\aid})$.
    
\item \textbf{Secret-key queries:} 
The challenger handles secret-key queries in two types as follows:
    \begin{itemize}
\item \textbf{Type I}: For each Type I secret-key query $(\gid,A,\vecv)\in\mathcal{Q}$, the challenger first computes $\vecr_{\gid,\vecv}\leftarrow\H(\gid,\vecv)$, and samples $$\veck_{\aid,\gid,\vecv}\leftarrow\samplepre(\matA_{\aid},\td_{\aid},\matP_\aid\matG^{-1}(\vecv)+\matB_\aid\vecr_{\gid,\vecv},\chi)$$ for each $\aid\in A$. Then it sets the secret key as $\sk_{\aid,\gid,\vecv}\leftarrow\veck_{\aid,\gid,\vecv}$.
        
\item \textbf{Type II}: The Type II secret-key queries are determined jointly by the partial set $\mathcal{Q}_{\partset}'$ and by the vectors $\{\vecv_j'\}$ output by the sampler $\samp_{\vecv}$. The challenger samples public randomness  $\vecr_{\pub}\rand\{0,1\}^{\kappa_1},\vecr_{\pri}\rand\{0,1\}^{\kappa_2}$, where $\kappa_1,\kappa_2$ are the upper bounds of the random bits used by $\samp_{\vecv}$ and $\samp_{\vecu}$, respectively. Then it samples $$\vecv_1',\ldots,\vecv_{Q'}'\leftarrow \samp_{\vecv}(1^\lambda;\vecr_{\pub})\text{ and } \vecu\leftarrow \samp_{\vecu}(1^\lambda,\vecr_{\pub};\vecr_{\pri}).$$ The full set of Type II secret-key queries is then given by $$\mathcal{Q}'=\{(\gid_1',A_1',\vecv_1'),\ldots,(\gid_{Q'}',A_{Q'}',\vecv_{Q'}')\}.$$ 
For each $(\gid',A',\vecv')\in \mathcal{Q}'$, the challenger computes $\vecr_{\gid',\vecv'}\leftarrow \H(\gid',\vecv')$, and samples $$\veck_{\aid,\gid',\vecv'}\leftarrow\samplepre(\matA_{\aid},\td_{\aid},\matP_\aid\matG^{-1}(\vecv')+\matB_\aid\vecr_{\gid',\vecv'},\chi)$$ for each $\aid\in A'$. It then sets the secret key as $\sk_{\aid,\gid',\vecv'}\leftarrow\veck_{\aid,\gid',\vecv'}$. 
        
\item The challenger sends the public randomness $\vecr_{\pub}$ and all the secret keys generated above to the adversary.
\end{itemize}
    
\item \textbf{Challenge ciphertext:} For each $\aid\in A^*$, the challenger samples $\vecs_{\aid}\rand\Z_{q}^n$ and $\vece_{1,\aid}\leftarrow D_{\Z,\chi}^m$. Then, it samples $\vece_2\leftarrow D_{\Z,\chi}^{m'},\vece_3\leftarrow D_{\Z,\chi}^{m}$. Finally, the challenge ciphertext is constructed as $\ct=\left(\{\vecc_{1,\aid}^\top\}_{\aid\in A^*},\vecc_2^\top,\vecc_3^\top\right)$, where    $$\vecc_{1,\aid}^\top=\vecs_{\aid}^\top\matA_{\aid}+\vece_{1,\aid}^\top,\vecc_{2}^\top=\sum_{\aid\in A^*}\vecs_{\aid}^\top\matB_{\aid}+\vece_2^\top,\vecc_{3}^\top=\sum_{\aid\in A^*}\vecs_{\aid}^\top\matP_{\aid}+\vece_3^\top+\vecu^\top\matG.$$
    
\item \textbf{Random oracle queries:} Upon receiving a query $(\gid,\vecv)\in\GID\times\Z_q^n$, the challenger checks whether the input $(\gid,\vecv)$ has been queried before---either during the adversary's direct random oracle queries or implicitly through processing the secret-key queries. If so, the challenger retrieves and responds with the stored value $\vecr_{\gid,\vecv}$ from the table $\T$. If the input is new, the challenger samples $\vecr_{\gid,\vecv}\leftarrow D_{\Z,\chi}^{m'}$, records the mapping $(\gid,\vecv)\mapsto \vecr_{\gid,\vecv}$ in the table $\T$, and then replies with $\vecr_{\gid,\vecv}$.
\end{itemize}
At the end of the experiment, the adversary outputs a bit $b'\in\{0,1\}$, which is taken as the output of the experiment. 

\iitem{Game $\H_{1}$} This experiment is identical to $\H_0$, except for how the challenger generates the challenge ciphertext.
\begin{itemize}
    \item  For each $\aid\in A^*\cap\mathcal{C}$, the challenger samples $\vecs_{\aid}\rand\Z_q^n$ and $\vece_{1,\aid}\leftarrow D_{\Z,\chi}^m$, and computes $\vecc_{1,\aid}\leftarrow\vecs_{\aid}^\top\matA_{\aid}+\vece_{1,\aid}^\top$. Then for each $\aid^*\in A^*\cap\mathcal{N}$, the challenger samples $\boxed{\vecc_{1,\aid^*}\rand\Z_q^m}$. It also samples $\boxed{\vecc_2\rand\Z_q^{m'},\vecc_3\rand\Z_q^m}$. Finally, it outputs the challenge ciphertext $$(\{\vecc_{1,\aid}^\top\}_{\aid\in A^*},\vecc_2^\top,\vecc_3^\top).$$
\end{itemize}

\iitem{Game $\H_{2}$} This experiment is identical to $\H_1$, except for how the challenger generates the third component of the challenge ciphertext. Specifically, 
\begin{itemize}
    \item The challenger samples $\vecu_{\delta}\rand\Z_q^n$, and computes the challenge ciphertext: $$(\{\vecc_{1,\aid}^\top\}_{\aid\in A^*},\vecc_2^\top,\boxed{\vecc_3^\top+\vecu_\delta^\top\matG}).$$
\end{itemize}

\iitem{Game $\H_{3}$} This experiment corresponds to the real static security game, in which the challenger encrypts the plaintext $\vecu+\vecu_{\delta}\in\Z_q^n$, where $\vecu$ is sampled from $\samp_{\vecu}$ as in $\H_1$ and $\vecu_{\delta}\rand\Z_q^n$. Specifically, this experiment is identical to $\H_0$, except that the embedded plaintext is shifted by $\vecu_{\delta}$. 
In particular, the third component of the challenge ciphertext is computed as $$\vecc_{3}^\top=\sum_{\aid\in A^*}\vecs_{\aid}^\top\matP_{\aid}+\vece_3^\top+(\vecu^\top+\vecu_{\delta}^\top)\matG.$$

\iitem{Game $\H_{4}$} This experiment corresponds to the real static security game, in which the challenger encrypts the uniformly random plaintext $\vecu_{\delta}$. Specifically, this experiment is identical to $\H_0$, except that the embedded plaintext is $\vecu_{\delta}$. 
In particular, $$\vecc_{3}^\top=\sum_{\aid\in A^*}\vecs_{\aid}^\top\matP_{\aid}+\vece_3^\top+\boxed{\vecu_{\delta}^\top\matG}.$$

\begin{lemma}\label{eva0}
Let $Q_0$ be the upper bound on the total number of secret-key queries (including those in the partial set $\mathcal{Q}_{\partset}'$) submitted by the adversary. Suppose that the following parameter conditions hold:
\begin{itemize}
    \item $m'>6n\log q$.
    \item $\chi'=\Omega(\sqrt{n\log q})$.
    \item Let $\chi$ be an error distribution parameter such that  $\chi\geq \lambda^{\omega(1)}\cdot (\sqrt{\lambda}(m+\ell)\chi_s+\lambda m\chi'\chi_s)$, where $\chi_s$ is a noise parameter such that  $\LWE_{n,m_1,q,\chi_s}$ assumption holds for some $m_1=\poly(m,m',Q_0)$.
    
    \item The assumption $\EVIPFE_{n,m,m',q,\chi,\chi}$ holds.
\end{itemize}
   Then we have that $\H_{0}\overset{c}{\approx}\H_1$.
\end{lemma}

\begin{proof}
    Suppose that there exists an \emph{efficient} adversary $\mathcal{A}$ that distinguishes $\H_{0}$ from $\H_{1}$ with non-negligible advantage. Based on adversary $\mathcal{A}$, we define a pair of sampling algorithms $\mathcal{S}_{\mathcal{A}}=(\mathcal{S}_{\mathcal{A},\vecu},\mathcal{S}_{\mathcal{A},\vecv})$ for global parameters $\gp=(1^\lambda,q,1^n,1^m,1^{m'},1^\chi,1^{\chi})$, with respect to the $\EVIPFE$ assumption, as follows.
    
    \begin{itemize}
        \item $\mathcal{S}_{\mathcal{A},\vecv}(\gp;(\vecr_{\pub_1},\vecr_{\pub_2}))$: This algorithm takes as inputs the global parameter $\gp$ and public randomness $\vecr_{\pub_1},\vecr_{\pub_2}\in\{0,1\}^*$, and proceeds as follows.
        \begin{itemize}
            \item Run adversary $\mathcal{A}(1^\lambda;\vecr_{\pub_1})$ with randomness $\vecr_{\pub_1}$, and extract from its output: a set of Type I secret-key queries $\mathcal{Q}=\{(\gid,A,\vecv)\}$, a partial set of Type II secret-key queries $\mathcal{Q}_{\partset}'=\{(\gid_1',A_1'),\ldots,(\gid_{Q'}', A_{Q'}')\}$, and a set of ciphertext authorities $A^*=\{\aid_1^*,\ldots,\aid_{\ell}^*\}$.
 \item Denote $|\mathcal{Q}|=Q$, and index the queries as $(\gid_1,A_1,\vecv_1),\ldots,(\gid_Q,A_Q,\vecv_Q)$.  Then for each $i\in [\ell]$, let $N_i\in[Q]$ denote the number of Type I secret-key queries in which the challenge authority $\aid_{i}^*$ appears. Suppose that authority $\aid_i^*$ is contained in the set $A_{j_1},\ldots, A_{j_{N_i}}$ for some indices $j_1,\ldots,j_{N_i}\in [Q]$, listed in increasing order. Define the mapping $\rho_i:[N_i]\rightarrow [Q]$ by setting $\rho_{i}(k)=j_k$. That is,  $\aid_{i}^*$ appears exactly in the set $A_{\rho_i(1)},\ldots, A_{\rho_i(N_i)}$ with the indices ordered increasingly.

\item Run the sampler $\samp_{\vecv}(1^\lambda;\vecr_{\pub_2})$ with randomness $\vecr_{\pub_2}$, which outputs the vectors $\vecv_1',\ldots,\vecv_{Q'}'$.

 \item For each $i\in [Q]$, sample $\vecr_i\leftarrow D_{\Z,\chi'}^{m'}$, and for each $j\in [Q']$, sample $\vecr_{j}'\leftarrow D_{\Z,\chi'}^{m'}$.

\item Finally, $\mathcal{S}_{\mathcal{A},{\vecv}}$ outputs 
\begin{align*}
&1^{\ell},\{1^{N_i+Q'}\}_{i\in [\ell]};\\
             &(\vecr_{\rho_1(1)},\vecv_{\rho_1(1)}),\ldots,(\vecr_{\rho_1(N_1)},\vecv_{\rho_{1}(N_1)}),(\vecr_1',\vecv_1'),\ldots,(\vecr_{Q'}',\vecv_{Q'}');\\
            &(\vecr_{\rho_2(1)},\vecv_{\rho_2(1)}),\ldots,(\vecr_{\rho_2(N_2)},\vecv_{\rho_{2}(N_2)}),(\vecr_1',\vecv_1'),\ldots,(\vecr_{Q'}',\vecv_{Q'}'); \\
            &\cdots;\\
            &(\vecr_{\rho_\ell(1)},\vecv_{\rho_\ell(1)}),\ldots,(\vecr_{\rho_\ell(N_\ell)},\vecv_{\rho_{\ell}(N_\ell)}),(\vecr_1',\vecv_1'),\ldots,(\vecr_{Q'}',\vecv_{Q'}').
            \end{align*}
         The output is structured as shown above, where each row (excluding the first) corresponds to a challenge authority.
        \end{itemize}
\item $\mathcal{S}_{\mathcal{A},\vecu}(\gp,(\vecr_{\pub_1},\vecr_{\pub_2});\vecr_{\pri})$: takes as input the global parameter $\gp$, public randomness $\vecr_{\pub_1},\vecr_{\pub_2}$ and private randomness $\vecr_{\pri}$. It computes $\vecu\leftarrow\samp_{\vecu}(1^\lambda,\vecr_{\pub_2};\vecr_{\pri})$ and outputs $\vecu$.
\end{itemize}

With respect to the sampling algorithm $\mathcal{S}_{\mathcal{A}}$, We invoke Claim~\ref{evpre} below, and defer its proof to the end of this section.
\begin{claim}\label{evpre}
Let $Q_0$ be the upper bound on the total number of secret-key queries (including those in the partial set $\mathcal{Q}_{\partset}'$) submitted by the adversary. Suppose that the following parameter conditions hold:
\begin{itemize}
    \item $m'>6n\log q$.
    \item $\chi'=\Omega(\sqrt{n\log q})$.
    \item Let $\chi$ be an error distribution parameter such that  $\chi\geq \lambda^{\omega(1)}\cdot (\sqrt{\lambda}(m+\ell)\chi_s+\lambda m\chi'\chi_s)$, where $\chi_s$ is a noise parameter such that  $\LWE_{n,m_1,q,\chi_s}$ assumption holds for some $m_1=\poly(m,m',Q_0)$.
    
    \item The sampling algorithm $\samp=(\samp_{\vecv},\samp_{\vecu})$ produces pseudorandom noisy inner products with noise parameter $\chi_s$.
\end{itemize}
 Then, for every efficient distinguisher $\mathcal{D}$, there exists a negligible function $\negl(\cdot)$ such that for all $\lambda\in \N^*$, we have $\Adv_{\mathcal{D},\mathcal{S}_{\mathcal{A}}}^\pre(\lambda)=\negl(\lambda)$, where $\Adv_{\mathcal{D},\mathcal{S}_{\mathcal{A}}}^\pre$ denotes the advantage of the distinguisher $\mathcal{D}$ (as defined in Assumption \ref{Evasiveipfe}) in the $\EVIPFE$ assumption with respect to the sampling algorithm $\mathcal{S}_{\mathcal{A}}=(\mathcal{S}_{\mathcal{A},\vecv},\mathcal{S}_{\mathcal{A},\vecu})$.
\end{claim}

\noindent\emph{Proof of Lemma \ref{eva0} (Continued)}. 
To complete the proof, we prove that if there exists an \emph{efficient} adversary $\mathcal{A}$ that can distinguish between $\H_{0}$ and $\H_1$ with non-negligible advantage, then we can construct an \emph{efficient} algorithm  $\mathcal{B}$ such that $\Adv_{\mathcal{S}_{\mathcal{A}},\mathcal{B}}^\post$ is non-negligible in the evasive \IPFE assumption with respect to the sampling algorithm $\mathcal{S}_{\mathcal{A}}=(\mathcal{S}_{\mathcal{A},\vecu},\mathcal{S}_{\mathcal{A},\vecv})$. The algorithm $\mathcal{B}$ proceeds as follows:

\begin{enumerate}
    \item Algorithm $\mathcal{B}$ begins by receiving an $\EVIPFE$ challenge $$(1^\lambda,(\vecr_{\pub_1},\vecr_{\pub_2}), \{\matA_{i},\vecy_{1,i}^\top\}_{i\in[\ell]},[\matB\mid \matP],\vecy_2^\top,\{\matK_i\}_{i\in[\ell]}),$$ where $\vecr_{\pub_1},\vecr_{\pub_2}\in\{0,1\}^*,\matA_{i}\in\Z_q^{n\times m}, \vecy_{1,i}\in\Z_q^m, \matK_i\in\Z_q^{m\times (N_i+Q')}$ for each $i\in [\ell]$, $\matB\in\Z_q^{n\ell\times m'},\matP\in\Z_q^{n\ell\times m},\vecy_2\in\Z_q^{m'+m}$.
    
    \item For each $i\in [\ell]$, algorithm $\mathcal{B}$ parses the matrix $\matK_i$ as $$\matK_i=[\matK_i^{(1)}\mid \matK_i^{(2)}],$$ where $\matK_{i}^{(1)}\in \Z_q^{m\times N_i}$ and $\matK_{i}^{(2)}\in \Z_q^{m\times Q'}$. Let $\veck_{i,j},\veck_{i,j}'$ denote the $j$-th column vectors of $\matK_i^{(1)}$ and $\matK_{i}^{(2)}$, respectively.
    
    \item Algorithm $\mathcal{B}$ runs algorithm $\mathcal{A}(1^\lambda;\vecr_{\pub_1})$, feeding in the randomness $\vecr_{\pub_1}$. The adversary $\mathcal{A}$ outputs the following queries:
        \begin{itemize}
     \item A set of corrupt authorities $\mathcal{C}\subseteq \mathcal{AU}$, along with their  public keys $\pk_{\aid}=(\matA_{\aid},\matB_{\aid},\matP_{\aid})$ for all $\aid\in\mathcal{C}$.
    
    \item A set of non-corrupt authorities $\mathcal{N}\subseteq \AU$, satisfying $\mathcal{N}\cap\mathcal{C}=\varnothing$.

    \item A ciphertext authority set $A^*\subseteq \mathcal{C}\cup \mathcal{N}$, satisfying $(A^*\cap\mathcal{C})\subsetneqq A^*$.
             
 \item A set of Type I secret-key queries $\mathcal{Q}=\{(\gid,A,\vecv)\}$, where $A\subseteq \mathcal{N}$ and $(A\cup \mathcal{C})\cap A^*\subsetneqq A^*$.

\item A partial set of Type II secret-key queries $\mathcal{Q}_{\partset}'=\{(\gid',A')\}$, where $A'\subseteq \mathcal{N}$ and $(A'\cup\mathcal{C})\cap A^*=A^*$.
\end{itemize}

 \item Algorithm $\mathcal{B}$ runs sampling algorithm $\mathcal{S}_{\mathcal{A},\vecv}(\gp;(\vecr_{\pub_1},\vecr_{\pub_2}))$, which outputs \begin{align*}
    & 1^{\ell},\{1^{N_i+Q'}\}_{i\in [\ell]}; \\
 & (\vecr_{\rho_1(1)},\vecv_{\rho_1(1)}),\ldots,(\vecr_{\rho_1(N_1)},\vecv_{\rho_{1}(N_1)}),(\vecr_1',\vecv_1'),\ldots,(\vecr_{Q'}',\vecv_{Q'}');                   \\
 & (\vecr_{\rho_2(1)},\vecv_{\rho_2(1)}),\ldots,(\vecr_{\rho_2(N_2)},\vecv_{\rho_{2}(N_2)}),(\vecr_1',\vecv_1'),\ldots,(\vecr_{Q'}',\vecv_{Q'}');                   \\
 & \cdots; \\
& (\vecr_{\rho_\ell(1)},\vecv_{\rho_\ell(1)}),\ldots,(\vecr_{\rho_\ell(N_\ell)},\vecv_{\rho_{\ell}(N_\ell)}),(\vecr_1',\vecv_1'),\ldots,(\vecr_{Q'}',\vecv_{Q'}').
\end{align*} 
The full set of Type II secret-key queries is then given by $\mathcal{Q}'=\{(\gid_1',A_1',\vecv_1'),\ldots,(\gid_{Q'}',A_{Q'}',\vecv_{Q'}')\}.$

Since the randomness used to simulate $\mathcal{A}$ in Step 3) comes from $\vecr_{\pub_1}$, the values produced here---namely, the vectors $\vecv_1,\ldots,\vecv_Q$---align exactly with those generated by $\mathcal{A}$'s simulation under the public randomness in Step 3). This ensures that algorithm $\mathcal{B}$'s internal simulation of $\mathcal{A}$ is consistent with the actual input instance of the $\EVIPFE$ challenge. 
        
     \item From the construction of $\mathcal{S}_{\mathcal{A}}$, we have that $\ell=|A^*\cap \mathcal{N}|$, i.e., the number of non-corrupt authorities appearing in the challenge ciphertext. Let $A^*\cap \mathcal{N}=\{\aid_{1}^*,\ldots,\aid_{\ell}^*\}$. Algorithm $\mathcal{B}$ first sets $\matA_{\aid_i^*}\leftarrow \matA_{i}$ for each $i\in [\ell]$, and parses $[\matB\mid \matP]$ as $$[\matB\mid\matP]=\left[\begin{array}{c|c}
        \matB_{\aid_1^*}&\matP_{\aid_1^*}\\
        \vdots&\vdots\\
        \matB_{\aid_{\ell}^*} &\matP_{\aid_{\ell}^*}
    \end{array}\right]\in\Z_q^{n\ell\times (m'+m)},$$ where $\matB_{\aid_i^*}\in\Z_q^{n\times m'}$ and $\matP_{\aid_i^*}\in\Z_{q}^{n\times m}$ for each $\aid_{i}^*\in A^*\cap\mathcal{N}$. 

    \item Let $|\mathcal{Q}|=Q$ and $|\mathcal{Q}'|=Q'$. For each $i\in [Q]$, algorithm $\mathcal{B}$ partitions $A_i=A_{i,\chal}\cup\bar{A}_{i,\chal}\subseteq \mathcal{N}$, where $A_{i,\chal}$ consists of the authorities in $A_i$ that appear in the ciphertext, i.e., $A_{i,\chal}=A_i\cap A^*$.

    \item Algorithm $\mathcal{B}$ initializes an empty table $\T:\GID\times\Z_q^n\rightarrow\Z_q^{m'}$. This table will be used to store and  consistently respond to all queries made to the random oracle during the experiment.
    
    \item The algorithm $\mathcal{B}$ responds to the queries as follows:
    \begin{itemize}
        \item \textbf{Public keys for non-corrupt authorities}: 
        \begin{itemize}
    \item For each $\aid_i^*\in  A^*\cap\mathcal{N}$, set $\pk_{\aid_i^*}\leftarrow(\matA_{\aid_i^*},\matB_{\aid_i^*},\matP_{\aid_i^*})$.
    
    \item For authorities $\aid\in\mathcal{N}\setminus A^*$, algorithm $\mathcal{B}$ samples $(\matA_{\aid},\td_{\aid})\leftarrow \trapgen(1^n,1^m,q)$, $\matB_{\aid}\rand\Z_q^{n\times m'},\matP_{\aid}\rand \Z_{q}^{n\times m}$ and sets the public key $\pk_{\aid}\leftarrow (\matA_{\aid},\matB_{\aid},\matP_{\aid})$.
            \end{itemize}
            \item \textbf{Secret keys}: The algorithm $\mathcal{B}$ responds to each secret-key query depending on its type:
\begin{itemize}
\item \textbf{Type I}: For a Type I secret-key query $(\gid_k,A_k,\vecv_k)$, recall that $A_k$ is partitioned as  $A_k=A_{k,\chal}\cup \bar{A}_{k,\chal}$, where $A_{k,\chal}=A_k\cap A^*$.

\begin{itemize}
\item For each $\aid_{i^*}\in A^*\cap \mathcal{N}$, recall that the number of secret-key queries involving $\aid_{i}^*$ is the parameter $N_i$ given in the $\EVIPFE$ challenge. Let $\rho(\cdot)$ be the index mapping previously defined in the proof of Lemma \ref{eva0}. For each $j\in [N_i]$, set $\sk_{\aid_i^*,\gid_{\rho_i(j)},\vecv_{\rho_i(j)}}\leftarrow\veck_{i,j}$. Then the algorithm $\mathcal{B}$ checks if the table $\T$ has ever recorded the image of $(\gid_{\rho_i(j)},\vecv_{\rho_i(j)})$, if not, store the mapping $(\gid_{\rho_i(j)},\vecv_{\rho_i(j)})\mapsto\vecr_{\rho_i(j)}$ to the table.

\item For each $k\in [Q]$, if $A_{k,\chal}=\varnothing$, then sample $\vecr_{\gid_k,\vecv_k}\leftarrow D_{\Z,\chi'}^{m'}$, and add the mapping $(\gid_{k},\vecv_{k})\mapsto\vecr_{\gid_k,\vecv_k}$ to the table. At this point, the table contains the image of all pairs $(\gid_k,\vecv_k)$ for each $k\in [Q]$.

\item For each $k\in [Q]$, for each $\aid\in \bar{A}_{k,\chal}$, compute $$\sk_{\aid,\gid_k,\vecv_k}\leftarrow\samplepre(\matA_{\aid},\td_{\aid},\matP_{\aid}\matG^{-1}(\vecv_k)+\matB_{\aid}\H(\gid_k,\vecv_k)),$$ where the value $\H(\gid_k,\vecv_k)$ is retrieved from the table $\T$.
\end{itemize}

\item \textbf{Type II}: For a Type II query $(\gid_j', A_j', \vecv_j')$,  recall that $A^*\cap\mathcal{N}=\{\aid_{1}^*,\ldots,\aid_{\ell}^{*}\}\subseteq A_{j}'$. 
\begin{itemize}
	\item For each $i\in [\ell],j\in [Q']$, set $\sk_{\aid_{i}^*,\gid_j',\vecv_j'}\leftarrow\veck_{i,j}'\in \Z_{q}^m$. Next, algorithm $\mathcal{B}$ adds the mapping $(\gid_j',\vecv_j')\mapsto \vecr_{j}'$ to the table $\T$.
	
	\item For each $\aid\in A_j'\setminus A^*$, algorithm $\mathcal{B}$ computes  $$\sk_{\aid,\gid_j',\vecv_j'}\leftarrow \samplepre(\matA_{\aid},\td_{\aid},\matP_{\aid}\matG^{-1}(\vecv_j')+\matB_{\aid}\vecr_{j}',\chi)$$ \emph{efficiently}.
\end{itemize}
\end{itemize}   

\item \textbf{Challenge ciphertext}: Algorithm $\mathcal{B}$ parses $\vecy_2$ as $\vecy_2^\top=[\hat{\vecy}_2^\top\mid \hat{\vecy}_{3}^\top]$ where $\hat{\vecy}_2\in\Z_q^{m'},\hat{\vecy}_{3}\in\Z_q^{m}$. It constructs the ciphertext as follows:
            For each $\aid\in A^*\cap\mathcal{C}$, sample  $\vecs_{\aid}\rand\Z_q^n$ and $\vece_{1,\aid}\leftarrow D_{\Z,\chi}^m$, then set $\vecc_{1,\aid}^\top\leftarrow \vecs_{\aid}^\top\matA_{\aid}+\vece_{1,\aid}^\top\in\Z_q^m$. For $\aid_i^*\in A^*\cap\mathcal{N}$, set $\vecc_{1,\aid_i^*}\leftarrow \vecy_{1,i}$. Finally, algorithm $\mathcal{B}$ outputs the ciphertext $$\ct=\left(\{\vecc_{1,\aid}^\top\}_{\aid\in A^*},\sum_{\aid\in A^*\cap\mathcal{C}}\vecs_{\aid}^\top\matB_{\aid}+\hat{\vecy}_2^\top,\sum_{\aid\in A^*\cap\mathcal{C}}\vecs_{\aid}^\top\matP_{\aid}+\hat{\vecy}_{3}^\top\right).$$
        \item \textbf{Random oracle queries:} Upon receiving a query $(\gid,\vecv)\in\GID\times\Z_q^n$, algorithm $\mathcal{B}$ checks whether the input $(\gid,\vecv)$ has been queried before---either during the adversary's direct random oracle queries or implicitly through processing the secret-key queries. If so, algorithm $\mathcal{B}$ retrieves and responds with the stored value $\vecr_{\gid,\vecv}$ from the table $\T$. If the input is new, the challenger samples $\vecr_{\gid,\vecv}\leftarrow D_{\Z,\chi}^{m'}$, records the mapping $(\gid,\vecv)\mapsto \vecr_{\gid,\vecv}$ in the table $\T$, and then replies with $\vecr_{\gid,\vecv}$.
\end{itemize}
    \item Algorithm $\mathcal{B}$ outputs whatever algorithm $\mathcal{A}$ outputs.
\end{enumerate}

The distributions of the public keys for non-corrupt authorities are exactly the same as those in $\H_{0}$ and $\H_1$, as they are uniformly generated. We now analyze the responses to the secret-key queries:
\begin{itemize}
    \item For each $\aid_i^*\in A^*\cap\mathcal{N}$, we have \begin{align*}
\veck_{i,j}&\leftarrow (\matA_{\aid_i^*})_{\chi}^{-1}(\matP_{\aid_i^*}\matG^{-1}(\vecv_{\rho_{i}(j)})+\matB_{\aid_i^*}\vecr_{\rho_i(j)}),\quad \text{for all } j\in [N_{i}],\\
        \veck_{i,j}'&\leftarrow (\matA_{\aid_i^*})_{\chi}^{-1}(\matP_{\aid_i^*}\matG^{-1}(\vecv_j')+\matB_{\aid_i^*}\vecr_{j}'),\quad \text{for all } j\in [Q'].
    \end{align*} 
    
    This exactly matches the distribution of $\sk_{\aid_i^*,\gid_{\rho_i(j)},\vecv_{\rho_i(j)}}$ and $\sk_{\aid_i^*,\gid_j',\vecv_j'}$ in the actual game, respectively. 
    
    \item For each $\aid\in A_j\setminus A^*$ for some $j\in [Q]$, the secret key $\sk_{\aid,\gid_j,\vecv_j}$ is generated using $\samplepre$, identical to the procedure in $\H_0$ and $\H_1$. The same argument also applies to $\aid\in A_j'\setminus A^*$ for $j\in [Q']$. 
\end{itemize}

Finally, we analyze the distribution of the challenge ciphertext. We consider the following two cases:
\begin{itemize}
    \item If for each $i\in[\ell]$,  $\vecy_{1,i}^\top=\vecs_{i}^\top\matA_{i}+\vece_{1,i}^\top$ and $\vecy_{2}^\top=\vecs^\top[\matB\mid\matP]+\vece_{2}^\top+[\veczero_{m'}^\top\mid \vecu^\top\matG]$ for some $\vecs^\top=[\vecs_1^\top\mid\cdots\mid\vecs_{\ell}^\top]\rand\Z_q^{n\ell},\vece_{1,i}\leftarrow D_{\Z,\chi}^m,\vece_2\leftarrow D_{\Z,\chi}^{m'+m}$. We parse $\vece_2^\top$ into two components as $[\hat{\vece}_2^\top\mid \hat{\vece}_3^\top]$ where $\hat{\vece}_2\in\Z_q^{m'}, \hat{\vece}_3\in\Z_q^{m}$. Then 
\begin{align*}
&\vecc_{1,\aid_i^*}^\top=\vecy_{1,i}^\top=\vecs_{i}^\top\matA_{\aid_i^*}+\vece_{1,i}^\top, \quad  \text{ for each }\aid_{i}^*\in A^*\cap\mathcal{N},\\
&\hat{\vecy}_{2}^\top=\sum_{\aid_i^*\in A^{*}\cap \mathcal{N}}\vecs_i^\top\matB_{\aid_i^*}+\hat{\vece}_{2}^\top,\\
&\hat{\vecy}_{3}^\top= \sum_{\aid_i^*\in A^{*}\cap \mathcal{N}}\vecs_i^\top\matP_{\aid_i^*}+\hat{\vece}_{3}^\top+\vecu^\top\matG.
\end{align*}
Then the ciphertext is constructed as
\begin{align*}
&\vecc_{1,\aid_{i}^*}^\top=\vecs_{i}^\top\matA_{\aid_i^*}+\vece_{1,i}^\top, \quad \text{ for each }\aid_{i}^*\in A^*\cap\mathcal{N},\\
&\vecc_{1,\aid}^\top=\vecs_{\aid}^\top\matA_{\aid}+\vece_{1,\aid}^\top, \quad \text{ for each }\aid \in A^*\cap\mathcal{C},\\
    &\vecc_{2}^\top=\sum_{\aid\in A^*\cap\mathcal{C}}\vecs_{\aid}^\top\matB_{\aid}+\sum_{\aid_i^*\in A^{*}\cap \mathcal{N}}\vecs_i^\top\matB_{\aid_i^*}+\hat{\vece}_{2}^\top,\\
 &\vecc_{3}^\top=\sum_{\aid\in A^*\cap\mathcal{C}}\vecs_{\aid}^\top\matP_{\aid}+ \sum_{\aid_i^*\in A^{*}\cap \mathcal{N}}\vecs_i^\top\matP_{\aid_i^*}+\hat{\vece}_{3}^\top+ \vecu^\top\matG,
\end{align*} where $\vecs_{\aid}\rand\Z_q^n$ and $\vece_{1,\aid}\leftarrow D_{\Z,\chi}^m$ for each $\aid\in A^*\cap \mathcal{C}$.
Since all randomness is sampled exactly as in $\H_0$, the resulting ciphertext generated by $\mathcal{B}$ is identically distributed to that in the experiment $\H_0$.

\item If for each $i\in [\ell]$, $\vecy_{1,i}\rand\Z_q^m$, and $\vecy_2\rand\Z_q^{m'+m}$. Then from the construction, the ciphertext is constructed as \begin{align*}
&\vecc_{1,\aid_i^*}^\top=
\vecy_{1,i}^\top, \quad \text{ for each }\aid_{i}^*\in A^*\cap\mathcal{N},\\
&\vecc_{1,\aid}^\top=\vecs_{\aid}^\top\matA_{\aid}+\vece_{1,\aid}^\top, \quad \text{ for each }\aid\in A^*\cap\mathcal{C},\\
    &\vecc_{2}^\top=\sum_{\aid\in A^*\cap\mathcal{C}}\vecs_{\aid}^\top\matB_{\aid}+\hat{\vecy}_2^\top,\\
 &\vecc_{3}^\top=\sum_{\aid\in A^*\cap\mathcal{C}}\vecs_{\aid}^\top\matP_{\aid}+\hat{\vecy}_3^\top,
\end{align*} where $\vecs_{\aid}\rand\Z_q^n$ and $\vece_{1,\aid}\leftarrow D_{\Z,\chi}^m$ for each $\aid\in A^*\cap\mathcal{C}$. Since $\{\vecy_{1,i}\}_{i\in [\ell]},\hat{\vecy}_2,\hat{\vecy}_3$ are all uniform and independent of each other, the overall distribution of $\{\vecc_{1,\aid}\}_{\aid\in A^*\cap\mathcal{N}},\vecc_2,\vecc_3$ is also uniform, which matches the distribution in experiment $\H_1$.
\end{itemize}

In either case, the algorithm $\mathcal{B}$ constructed perfectly simulates experiments $\H_0$ and $\H_1$. The advantage $\Adv_{\mathcal{S}_{\mathcal{A}},\mathcal{B}}^\post$ in the $\EVIPFE$ assumption with respect to the sampling algorithm $\mathcal{S}_{\mathcal{A}}$ is the same as the non-negligible advantage of $\mathcal{A}$ distinguishing between $\H_0$ and $\H_1$. By the definition of the $\EVIPFE$ assumption, the existence of such an efficient post-challenge adversary $\mathcal{B}$ with non-negligible advantage would imply the existence of another efficient adversary $\mathcal{B}'$ breaking the pre-challenge security, i.e., $\Adv^\pre_{\mathcal{S}_{\mathcal{A}},\mathcal{B}'}$ is non-negligible, contradicting Claim \ref{evpre}. Hence, under the $\EVIPFE$ assumption, no efficient adversary can distinguish between $\H_0$ and $\H_1$ with non-negligible advantage. This completes the proof.
\end{proof}

\begin{lemma}\label{eva1}
    We have that $\H_1\equiv\H_2$.
\end{lemma}

\begin{proof}
    The proof follows directly from the fact that in $\H_2$, the third component of the ciphertext is generated as $\vecc_3^\top + \vecu_\delta^\top \matG$, where $\vecu_\delta$ is sampled uniformly from $\Z_q^n$ and is independent of all other components. The resulting component remains uniform and thus identical to that in $\H_1$.
\end{proof}

\begin{lemma}\label{eva2}
    Under the same assumptions as in Lemma \ref{eva0}, we have that $\H_2 \overset{c}{\approx} \H_3$.
\end{lemma}

\begin{proof}
    The proof follows essentially the same argument as in Lemma \ref{eva0}.
\end{proof}

\begin{lemma}\label{eva3}
    We have that $\H_{3}\equiv \H_4$.
\end{lemma}

\begin{proof}
    The result follows directly from the fact that both $\vecu + \vecu_\delta$ in $\H_3$ and $\vecu_\delta$ in $\H_4$ are uniformly distributed over $\Z_q^n$, as $\vecu_\delta$ is sampled uniformly and independently of $\vecu$. Hence, $\vecc_3^\top$ is identically distributed in both experiments.
\end{proof}

\noindent\emph{Proof of Theorem \ref{security1} (Continued).} By Lemmas \ref{eva0}, \ref{eva1}, \ref{eva2}, and \ref{eva3}, and a standard hybrid argument, we conclude that under the given parameter constraints, $\H_0\overset{c}{\approx}\H_4$. This implies that ciphertexts encrypting the vector $\vecu$ generated by $\samp_{\vecv}$ are computationally indistinguishable from those encrypting a uniformly random vector $\vecu_{\delta}$. Therefore, the theorem follows.
\end{proof}

To complete the proof of Theorem~\ref{security1}, it remains to establish the validity of Claim~\ref{evpre}.
\begin{proof}[Proof of Claim \ref{evpre}]
We prove the claim by defining a sequence of hybrid experiments.

\iitem{Game $\H_0^\pre$} 
    On input the security parameter $\lambda$, the challenger proceeds as follows: 
    \begin{enumerate}
        \item Let $\kappa=\kappa(\lambda)$ be an upper bound on the number of random bits used by the adversary $\mathcal{A}$, and the sampling algorithms $\samp_{\vecu}$, $\samp_{\vecv}$. The challenger samples $\vecr_{\pub_1},\vecr_{\pub_2},\vecr_{\pri}\rand\{0,1\}^{\kappa}$ and run the sampling algorithms $\mathcal{S}_{\mathcal{A}}=(\mathcal{S}_{\mathcal{A},\vecv},\mathcal{S}_{\mathcal{A},\vecu})$ as follows.
        \begin{itemize}
            \item Run $\mathcal{S}_{\mathcal{A},\vecv}(\gp;(\vecr_{\pub_1},\vecr_{\pub_2}))$, which outputs \begin{align*}
            &1^{\ell},\{1^{N_i+Q'}\}_{i\in [\ell]};\\
             &(\vecr_{\rho_1(1)},\vecv_{\rho_1(1)}),\ldots,(\vecr_{\rho_1(N_1)},\vecv_{\rho_{1}(N_1)}),(\vecr_1',\vecv_1'),\ldots,(\vecr_{Q'}',\vecv_{Q'}');\\
            &(\vecr_{\rho_2(1)},\vecv_{\rho_2(1)}),\ldots,(\vecr_{\rho_2(N_2)},\vecv_{\rho_{2}(N_2)}),(\vecr_1',\vecv_1'),\ldots,(\vecr_{Q'}',\vecv_{Q'}'); \\
            &\cdots;\\
            &(\vecr_{\rho_\ell(1)},\vecv_{\rho_\ell(1)}),\ldots,(\vecr_{\rho_\ell(N_\ell)},\vecv_{\rho_{\ell}(N_\ell)}),(\vecr_1',\vecv_1'),\ldots,(\vecr_{Q'}',\vecv_{Q'}').
            \end{align*}
            \item Run $\vecu\leftarrow \mathcal{S}_{\mathcal{A},\vecu}(\gp,(\vecr_{\pub_1},\vecr_{\pub_2});\vecr_{\pri})$.
        \end{itemize}
        
        \item Sample $\matB\rand\Z_q^{n\ell\times m'},\matP\rand\Z_q^{n\ell\times m}$, and parse the matrices $$\matB=\left[\begin{array}{c}
        \matB_{1}\\
        \vdots\\
        \matB_{\ell} 
\end{array}\right]\in\Z_q^{n\ell\times m'},\matP=\left[\begin{array}{c}
        \matP_{1}\\
        \vdots\\
        \matP_{\ell}
\end{array}\right]\in\Z_q^{n\ell\times m},$$ with $\matB_{i}\in \Z_q^{n\times m'},\matP_{i}\in\Z_q^{n\times m}$ for all $i\in [\ell]$. 

\item For each $i\in [\ell]$, define the matrix:
    \begin{align*}
\matQ_i^{(1)}\leftarrow\left[\matP_{i}\matG^{-1}(\vecv_{\rho_i(1)})+\matB_{i}\vecr_{\rho_i(1)}\mid\cdots\mid\matP_{i}\matG^{-1}(\vecv_{\rho_i(N_i)})+\matB_{i}\vecr_{\rho_i(N_i)}\right]\in \Z_q^{n\times N_i}, \end{align*}

and similarly define:  
$$\matQ_i^{(2)}=\left[\matP_{i}\matG^{-1}(\vecv_{1}')+\matB_{i}\vecr_{1}'\mid \cdots\mid \matP_{i}\matG^{-1}(\vecv_{Q'}')+\matB_{i}\vecr_{Q'}'\right]\in\Z_q^{n\times Q'}.$$
 Finally, set $\matQ_i=[\matQ_i^{(1)}\mid \matQ_i^{(2)}]\in\Z_{q}^{n\times (N_i+Q')}$. The construction of each matrix $\matQ_i$ follows exactly the format specified in Assumption~\ref{evIPFE}, ensuring consistency with the \EVIPFE input distribution.            

\item Then the challenger samples $\vecs_1,\ldots,\vecs_\ell\rand\Z_q^n$ and sets $\vecs^\top=\left[\vecs_1^\top\mid\cdots\mid \vecs_{\ell}^\top\right]\in \Z_q^{n\ell}$. For each $i\in [\ell]$, it samples $\vece_{1,i}\leftarrow D_{\Z,\chi}^m,\vece_{3,i}\leftarrow D_{\Z,\chi}^{N_i+Q'}$. Then it samples $\vece_2\leftarrow D_{\Z,\chi}^{m'+m}$.
            
\item The challenger samples $(\matA_{1},\td_{1}),\ldots,(\matA_{\ell},\td_{\ell})\leftarrow \trapgen(1^n,1^m,q)$. For each $i\in [
            \ell
            ]$, it computes the following values : 
            \begin{itemize}
                \item  $\vecz_{1,i}^\top\leftarrow \vecs_i^\top\matA_i+\vece_{1,i}^\top\in\Z_q^m$  for each $i\in [\ell]$.

        \item                 $\vecz_{2}^\top= \vecs^\top[\matB\mid\matP]+\vece_2^\top+[\veczero_{m'}^\top\mid \vecu^\top\matG]=\left[\sum_{i\in[\ell]}\vecs_i^\top\matB_i \,\Bigg|\,\sum_{i\in[\ell]}\vecs_i^\top\matP_i+\vecu^\top\matG\right]+\vece_2^\top\in\Z_q^{m'+m}.$
                
                \item $\vecz_{3,i}^\top\leftarrow \vecs_{i}^\top\matQ_i+\vece_{3,i}^\top\in \Z_q^{N_i+Q'}$ for each $i\in [\ell]$.
\end{itemize}  
The component $\vecz_{3,i}$ and $\vece_{3,i}$ are parsed as $$\vecz_{3,i}^\top=[\vecz_{3,i}^{(1)\top}\mid \vecz_{3,i}^{(2)\top}]\in \Z_q^{N_i+Q'},\quad \vece_{3,i}^{\top}=[\vece_{3,i}^{(1)\top}\mid \vece_{3,i}^{(2)\top}]\in \Z_q^{N_i+Q'},$$  where $\vecz_{3,i}^{(1)},\vece_{3,i}^{(1)}\in\Z_q^{N_i}$ and $\vecz_{3,i}^{(2)},\vece_{3,i}^{(2)}\in\Z_q^{Q'}$. For clarity, define $\vect_{i,j}=\matP_{i}\matG^{-1}(\vecv_{\rho_i(j)})+\matB_{i}\vecr_{\rho_i(j)}\in \Z_q^n$ and $y_{i,j}=\vecs_i^\top\vect_{i,j}\in\Z_q$ for each $i\in [\ell]$ and $j\in[N_i]$. Similarly define $\vect_{i,j}'=\matP_{i}\matG^{-1}(\vecv_{j}')+\matB_{i}\vecr_{j}'\in \Z_q^n$ and $y_{i,j}'=\vecs_i^\top\vect_{i,j}'\in\Z_q$ for each $i\in [\ell]$ and $j\in[Q']$.

In summary, we obtain \begin{align*}
\vecz_{3,i}^{(1)\top} & =\vecs_i^\top\matQ_i^{(1)}+\vece_{3,i}^{(1)\top}                =\left[\vecs_i^\top\vect_{i,1}\mid\cdots\mid\vecs_i^\top\vect_{i,N_i}\right]+\vece_{3,i}^{(1)\top}=\left[y_{i,1}\mid\cdots\mid y_{i,N_i}\right]+\vece_{3,i}^{(1)\top}\in\Z_q^{N_i}, \\
\vecz_{3,i}^{(2)\top} & =\vecs_i^\top\matQ_i^{(2)}+\vece_{3,i}^{(2)\top}=\left[\vecs_i^\top\vect_{i,1}'\mid\cdots\mid\vecs_i^\top\vect_{i,Q'}'\right]+\vece_{3,i}^{(2)\top}=\left[y_{i,1}'\mid\cdots\mid y_{i,Q'}'\right]+\vece_{3,i}^{(2)\top}\in\Z_q^{Q'}.
\end{align*}

\item The challenger outputs the tuple $(1^\lambda,(\vecr_{\pub_1},\vecr_{\pub_2}),\{(\matA_i,\vecz_{1,i}^\top)\}_{i\in [\ell]},[\matB\mid\matP],\vecz_2^\top,\{\vecz_{3,i}^\top\}_{i\in[\ell]})$ and sends it to the distinguisher $\mathcal{D}$.

\item The distinguisher $\mathcal{D}$ outputs a bit $\hat{b}\in\{0,1\}$, which is taken as the output of the experiment.
\end{enumerate}
    
\iitem{Game $\H_1^{\pre}$} The experiment is identical to $\H_0^\pre$, except for the procedure how the challenger samples $\vecz_{1,i},\vecz_2,\vecz_{3,i}$, which are now smudged with additional noise.
    \begin{itemize}
        \item $\vecz_{1,i}$: For each $i\in [\ell]$, sample $\tilde{\vece}_{1,i}\leftarrow D_{\Z,\chi_s}^m$, and set $\boxed{\vecz_{1,i}^\top\leftarrow \vecs_i^\top\matA_i+\tilde{\vece}_{1,i}^\top+\vece_{1,i}^\top}$.
        
        \item $\vecz_2$: For each $i\in [\ell]$, sample $\tilde{\vece}_{2,i}\leftarrow D_{\Z,\chi_s}^{m'}$ and $\tilde{\vece}_{2,i}'\leftarrow D_{\Z,\chi_s}^m$, and compute $\tilde{\vecz}_{2,i}^\top\leftarrow\vecs_i^\top\matB_i+\tilde{\vece}_{2,i}^\top$ and $\tilde{\vecz}_{2,i}'^\top\leftarrow\vecs_i^\top\matP_i+\tilde{\vece}_{2,i}'^\top$. Then set $$\boxed{\vecz_2^\top\leftarrow \left[\sum_{i\in [\ell]}\tilde{\vecz}_{2,i}^\top\,\Bigg|\,\sum_{i\in [\ell]}\tilde{\vecz}_{2,i}'^\top+\vecu^\top\matG\right]+\vece_2^\top=\left[\sum_{i\in [\ell]}(\vecs_i^\top\matB_i+\tilde{\vece}_{2,i}^\top)\,\Bigg|\,\sum_{i\in [\ell]}(\vecs_i^\top\matP_i+\tilde{\vece}_{2,i}'^\top)+\vecu^\top\matG\right]+\vece_2^\top}.$$
        
        \item $\vecz_{3,i}^{(1)}$: For each $i\in [\ell]$ and $j\in [N_i]$, compute 
        \begin{center}
             \fbox{
        \parbox{0pt}{\begin{align*}
            y_{i,j}&\leftarrow (\vecs_i^\top\matP_i+\tilde{\vece}_{2,i}'^\top)\matG^{-1}(\vecv_{\rho_i(j)})+(\vecs_i^\top\matB_i+\tilde{\vece}_{2,i}^\top)\vecr_{\rho_i(j)}\\
            &=\tilde{\vecz}_{2,i}'^\top\matG^{-1}(\vecv_{\rho_i(j)})+\tilde{\vecz}_{2,i}^\top\vecr_{\rho_i(j)},
        \end{align*} }} 
        \end{center}
        and set $$\vecz_{3,i}^{(1)\top}\leftarrow \left[y_{i,1}\mid\cdots\mid y_{i,N_i}\right]+\vece_{3,i}^{(1)\top}\in\Z_q^{N_i}. $$
        
         \item $\vecz_{3,i}^{(2)}$: For each $i\in [\ell]$ and $j\in [Q']$, compute 
         \begin{center}
             \fbox{
        \parbox{0pt}{\begin{align*}
         y_{i,j}'&\leftarrow (\vecs_i^\top\matP_i+\tilde{\vece}_{2,i}'^\top)\matG^{-1}(\vecv_{j}')+(\vecs_i^\top\matB_i+\tilde{\vece}_{2,i}^\top)\vecr_{j}'\\
            &=\tilde{\vecz}_{2,i}'^\top\matG^{-1}(\vecv_{j}')+\tilde{\vecz}_{2,i}^\top\vecr_{j}',
        \end{align*}}}
         \end{center}
         and set $$\vecz_{3,i}^{(2)\top}\leftarrow \left[y_{i,1}'\mid\cdots\mid y_{i,Q'}'\right]+\vece_{3,i}^{(2)\top}\in\Z_q^{Q'}. $$
    \end{itemize}
    
  \iitem{Game $\H_2^{\pre}$} The experiment is identical to $\H_1^\pre$, except for the procedure how the values $\vecz_{1,i},\vecz_2,\vecz_{3,i}$ are sampled.
    \begin{itemize}
        \item $\vecz_{1,i}$: For each $i\in [\ell]$, sample $\boxed{\vecz_{1,i}\rand\Z_q^m}$.
        
        \item $\vecz_2$: For each $i\in [\ell]$, sample $\boxed{\tilde{\vecz}_{2,i}\rand\Z_q^{m'}}$ and $\boxed{\tilde{\vecz}_{2,i}'\rand\Z_q^m}$. Then set $$\vecz_2^\top\leftarrow \left[\sum_{i\in [\ell]}\tilde{\vecz}_{2,i}^\top\,\Bigg|\,\sum_{i\in [\ell]}\tilde{\vecz}_{2,i}'^\top+\vecu^\top\matG\right]+\vece_2^\top.$$
        
        \item $\vecz_{3,i}^{(1)}$: For each $i\in [\ell]$ and $j\in [N_i]$, compute 
        $$y_{i,j}\leftarrow \tilde{\vecz}_{2,i}'^\top\matG^{-1}(\vecv_{\rho_i(j)})+\tilde{\vecz}_{2,i}^\top\vecr_{\rho_i(j)},$$ then set $$\vecz_{3,i}^{(1)\top}\leftarrow \left[y_{i,1}\mid\cdots\mid y_{i,N_i}\right]+\vece_{3,i}^{(1)\top}\in\Z_q^{N_i}. $$
        
         \item $\vecz_{3,i}^{(2)}$: For each $i\in [\ell]$ and $j\in [Q']$, compute 
        $$y_{i,j}'\leftarrow \tilde{\vecz}_{2,i}'^\top\matG^{-1}(\vecv_{j}')+\tilde{\vecz}_{2,i}^\top\vecr_{j}',$$ then set $$\vecz_{3,i}^{(2)\top}\leftarrow \left[y_{i,1}'\mid\cdots\mid y_{i,Q'}'\right]+\vece_{3,i}^{(2)\top}\in\Z_q^{Q'}. $$
        \end{itemize}
        
    \iitem{Game $\H_3^\pre$} The experiment is identical to $\H_2^\pre$, except it samples $\boxed{\vecz_{3,1}^{(1)}\rand\Z_q^{N_1}}$.

    \iitem{Game $\H_4^\pre$} The experiment is identical to $\H_3^\pre$, except for how the components $\vecz_2$ and $\vecz_{3,1}^{(2)}$ are sampled.
        \begin{itemize}
            \item $\vecz_2$: For each $i\in [\ell]$, sample $\tilde{\vecz}_{2,i}\rand\Z_q^{m'},\tilde{\vecz}_{2,i}'\rand\Z_q^{m}$. Then set $$\boxed{\vecz_2^\top\leftarrow \left[\sum_{i\in [\ell]}\tilde{\vecz}_{2,i}^\top\,\Bigg|\,\sum_{i\in [\ell]}\tilde{\vecz}_{2,i}'^\top\right]+\vece_2^\top},$$ which omits the additive term $[\veczero_{m'} \mid \vecu^\top\matG]$ compared to the construction in $\H_3^\pre$.
            \item $\vecz_{3,i}^{(2)}$: For $i=1$, sample $\boxed{\vecz_{3,1}^{(2)}\rand\Z_q^{Q'}}$. For $i\neq 1$, sample $\vecz_{3,i}^{(2)}$ in the same manner as in $\H_{3}^\pre$.
        \end{itemize}

        \iitem{Game $\H_5^\pre$} The experiment is identical to $\H_4^\pre$, except it samples $\vecz_{2}\rand\Z_q^{m'+m}$, $\vecz_{3,i}^{(1)}\rand\Z_q^{N_i},\vecz_{3,i}^{(2)}\rand\Z_q^{Q'}$ for each $i\in [\ell]$.

\begin{lemma}\label{evpre0}
    Suppose that $\chi\geq\lambda^{\omega(1)}\cdot(\sqrt{\lambda}(m+\ell)\chi_s+\lambda m'\chi'\chi_s)$. Then we have $\H_{0}^{\pre}\overset{s}{\approx}\H_1^\pre$.
\end{lemma}

\begin{proof}
     The only difference between $\H_{0}^{\pre}$ and $\H_{1}^{\pre}$ lies in the way the error terms in $\vecz_{1,i},\vecz_{2},\vecz_{3,i}$ are generated. We analyze each of these components individually below:
    \begin{itemize}
        \item \textbf{Error term in $\vecz_{1,i}$}: In experiment $\H_{1}^{\pre}$, the error term in $\vecz_{1,i}$ is given by $\vece_{1,i}+\tilde{\vece}_{1,i}$, where $\vece_{1,i}\leftarrow D_{\Z,\chi}^m$ and $\tilde{\vece}_{1,i}\leftarrow D_{\Z,\chi_s}^m$. By Lemma \ref{truncated}, we have $\|\tilde{\vece}_{1,i}\|\leq \sqrt{\lambda}\chi_s$ with overwhelming probability. Given that $\chi>\lambda^{\omega(1)}\cdot\sqrt{\lambda}\chi_s$, Lemma \ref{smudge} implies that $\vece_{1,i}+\tilde{\vece}_{1,i}\overset{s}{\approx}\vece_{1,i}$. 
        
        \item \textbf{Error term in $\vecz_{2}$}:  In $\H_{1}^\pre$, the error term in $\vecz_2$ is given by $$\vece_2^\top+\tilde{\vece}_2^\top:=\vece_2^\top+\left[\sum_{i\in[\ell]}\tilde{\vece}_{2,i}^\top\,\Bigg|\,\sum_{i\in [\ell]}\tilde{\vece}_{2,i}'^\top\right],$$ where $\tilde{\vece}_{2,i}\leftarrow D_{\Z,\chi_s}^{m'}$ and $\tilde{\vece}_{2,i}'\leftarrow D_{\Z,\chi_s}^{m}$. Lemma \ref{truncated} implies that $\|\tilde{\vece}_2\|\leq \ell\cdot\sqrt{\lambda}\chi_s$ with overwhelming probability. Given that $\chi>\lambda^{\omega(1)}\cdot\ell\sqrt{\lambda}\chi_s$, Lemma \ref{smudge} ensures that $\vece_2+\tilde{\vece}_2\overset{s}{\approx}\vece_2$. Hence, the distributions of $\vecz_{2}$ in $\H_{0}^{\pre}$ and $\H_{1}^{\pre}$ are statistically indistinguishable. 
        
        \item \textbf{Error term in $\vecz_{3,i}$}: In experiment $\H_{1}^{\pre}$, the error term of $y_{i,j}$ is given by 
     $$\tilde{\vece}_{2,i}'^\top\matG^{-1}(\vecv_{\rho_i(j)})+\tilde{\vece}_{2,i}^\top\vecr_{\rho_i(j)}.$$ By Lemma \ref{truncated}, the infinity norm of this term is at most $\sqrt{\lambda}m\chi_s+\lambda m'\chi'\chi_s$ with overwhelming probability. Since $\chi>\lambda^{\omega(1)}\cdot(\sqrt{\lambda}m\chi_s+\lambda m'\chi'\chi_s)$,  Lemma \ref{smudge} implies that the distributions of $\vecz_{3,i}^{(1)}$ in the two experiments are statistically indistinguishable. The same argument applies to the component $\vecz_{3,i}^{(2)}$ as well.
    \end{itemize}

    In conclusion, under the constraints of the stated parameters, all error terms in $\vecz_{1,i},\vecz_2,\vecz_{3,i}$ 
are statistically close in the two experiments. Hence, we have
$\H_0^{\pre} \overset{s}{\approx} \H_1^{\pre}$,
as claimed.
\end{proof}

\begin{lemma}\label{evpre1}
    Suppose that the assumption $\LWE_{n,m_1,q,\chi_s}$ holds for $m_1=\poly(m,m')$. Then we have $\H_{1}^\pre\overset{c}{\approx}\H_2^\pre$.
\end{lemma}

\begin{proof}
    We prove this lemma by defining a sequence of intermediate hybrid experiments $\H_{1,d}^\pre$ for each $0\leq d\leq \ell$.
    
    \iitem{Game $\H_{1,d}^{\pre}$} This experiment is identical to $\H_1^\pre$, except for how the value $\vecz_{1,i},\vecz_2,\vecz_{3,i}$ are generated:
        \begin{itemize}
            \item $\vecz_{1,i}$:
            \begin{itemize}
                \item For each $i\leq d$, sample $\vecz_{1,i}\rand\Z_q^m$.
                \item For each $i>d$, sample $\vece_{1,i}\leftarrow D_{\Z,\chi}^m, \tilde{\vece}_{1,i}\leftarrow D_{\Z,\chi_s}^m$, and set  $\vecz_{1,i}^\top\leftarrow \vecs_i^\top\matA_i+\tilde{\vece}_{1,i}^\top+\vece_{1,i}^\top$.
            \end{itemize} 
            
            \item $\vecz_{2}$:
            \begin{itemize}
                \item For each $i\leq d$, sample $\tilde{\vecz}_{2,i}\rand\Z_q^{m'}$ and $\tilde{\vecz}_{2,i}'\rand\Z_q^m$.
                \item For each $i>d$, sample $\tilde{\vece}_{2,i}\leftarrow D_{\Z,\chi_s}^{m'}, \tilde{\vece}_{2,i}'\leftarrow D_{\Z,\chi_s}^m$, and set $\tilde{\vecz}_{2,i}^\top\leftarrow\vecs_i^\top\matB_i+\tilde{\vece}_{2,i}^\top$ and $\tilde{\vecz}_{2,i}'^\top\leftarrow\vecs_i^\top\matP_i+\tilde{\vece}_{2,i}'^\top$.
                \item 
               Then sample $\vece_2\leftarrow D_{\Z,\chi}^{m'+m}$, and set $$\vecz_2^\top\leftarrow \left[\sum_{i\in [\ell]}\tilde{\vecz}_{2,i}^\top\,\Bigg|\,\sum_{i\in [\ell]}\tilde{\vecz}_{2,i}'^\top+\vecu^\top\matG\right]+\vece_2^\top.$$
            \end{itemize}  
           
           \item $\vecz_{3,i}^{(1)}$: For each $i\in [\ell]$, sample $\vece_{3,i}^{(1)}\leftarrow D_{\Z,\chi}^{N_i}$. Then for each $j\in [N_i]$, compute 
        $$y_{i,j}\leftarrow \tilde{\vecz}_{2,i}'^\top\matG^{-1}(\vecv_{\rho_i(j)})+\tilde{\vecz}_{2,i}^\top\vecr_{\rho_i(j)},$$ and let $$\vecz_{3,i}^{(1)\top}\leftarrow \left[y_{i,1}\mid\cdots\mid y_{i,N_i}\right]+\vece_{3,i}^{(1)\top}\in\Z_q^{N_i}. $$
         \item $\vecz_{3,i}^{(2)}$: Similarly, for each $i\in [\ell]$, sample $\vece_{3,i}^{(2)}\leftarrow D_{\Z,\chi}^{Q'}$. Then for $j\in [Q']$, compute 
        $$y_{i,j}'\leftarrow \tilde{\vecz}_{2,i}'^\top\matG^{-1}(\vecv_{j}')+\tilde{\vecz}_{2,i}^\top\vecr_{j}',$$ and let $$\vecz_{3,i}^{(2)\top}\leftarrow \left[y_{i,1}'\mid\cdots\mid y_{i,Q'}'\right]+\vece_{3,i}^{(2)\top}\in\Z_q^{Q'}. $$ 
        \end{itemize}
    
    It follows naturally by the construction that $\H_{1,0}^\pre\equiv \H_1^\pre$, and $\H_{1,\ell}^{\pre}\equiv \H_2^\pre$. We now demonstrate that, for each $d\in [\ell]$, the hybrids $\H_{1,d-1}^\pre$ and $\H_{1,d}^\pre$ are computationally indistinguishable under the $\LWE_{n,m_1,q,\chi_s}$ assumption for $m_1=\poly(m,m')$. 
    
    Suppose, for contradiction, that there exists an efficient distinguisher $\mathcal{D}$ that can distinguish between $\H_{1,d-1}^\pre$ and $\H_{1,d}^\pre$ with non-negligible advantage. We use $\mathcal{D}$ to construct an adversary $\mathcal{D}'$ that breaks the \LWE assumption. The adversary $\mathcal{D}'$ proceeds as follows:
    \begin{enumerate}
        \item Algorithm $\mathcal{D}'$ begins by receiving an $\LWE$ challenge $(\matD,\vecy)$, where $\matD\in\Z_q^{n\times 
 (2m+m')},\vecy\in \Z_q^{2m+m'}$. Then it parses $\matD$ and $\vecy$ as $$\matD=[\matA_d\mid\matB_d\mid\matP_d], \quad\vecy^\top=[\vecy_1^\top\mid\vecy_2^\top\mid\vecy_2'^\top],$$ where $\matA_d,\matP_d\in \Z_q^{n\times m},\matB_d\in \Z_q^{n\times m'}$, and $\vecy_1,\vecy_2'\in\Z_q^m,\vecy_2\in\Z_q^{m'}$.
 
 \item Algorithm $\mathcal{D}'$ simulates $\mathcal{S}_{\mathcal{A}}=(\mathcal{S}_{\mathcal{A},\vecv},\mathcal{S}_{\mathcal{A},\vecu})$ as follows:
 \begin{itemize}
     \item It samples $\vecr_{\pub_1},\vecr_{\pub_2},\vecr_{\pri}\rand\{0,1\}^\kappa$ as the randomness for algorithm $\mathcal{S}_{\mathcal{A}}$. 
     \item It computes $$\left(\begin{array}{l}
          1^{\ell},1^{N_i+Q'};\\
    (\vecr_{\rho_1(1)},\vecv_{\rho_1(1)}),\ldots,(\vecr_{\rho_1(N_1)},\vecv_{\rho_{1}(N_1)}),(\vecr_1',\vecv_1'),\ldots,(\vecr_{Q'}',\vecv_{Q'}');                   \\
  (\vecr_{\rho_2(1)},\vecv_{\rho_2(1)}),\ldots,(\vecr_{\rho_2(N_2)},\vecv_{\rho_{2}(N_2)}),(\vecr_1',\vecv_1'),\ldots,(\vecr_{Q'}',\vecv_{Q'}');                   \\
 \cdots; \\
(\vecr_{\rho_\ell(1)},\vecv_{\rho_\ell(1)}),\ldots,(\vecr_{\rho_\ell(N_\ell)},\vecv_{\rho_{\ell}(N_\ell)}),(\vecr_1',\vecv_1'),\ldots,(\vecr_{Q'}',\vecv_{Q'}')       \end{array}\right)\leftarrow\mathcal{S}_{\mathcal{A},\vecv}(\gp;(\vecr_{\pub_1},\vecr_{\pub_2})).$$ 
     \item It computes $\vecu\leftarrow\mathcal{S}_{\mathcal{A},\vecu}(\gp,(\vecr_{\pub_1},\vecr_{\pub_2});\vecr_{\pri})$.
     \end{itemize}
     \item Algorithm $\mathcal{D}'$ samples $\matA_i,\matP_i\rand\Z_q^{n\times m}$ and $\matB_i\rand\Z_q^{n\times m'}$ for each $i\neq d$. It also sets $\matA_d$, $\matB_d$, and $\matP_d$ to be the respective components provided in the \LWE challenge instance. Note that the distribution of the matrix $\{\matA_i\}_{i\in [\ell]},[\matB\mid\matP]$ constructed here exactly matches the setup process in the $\EVIPFE$ assumption.
     
     \item Algorithm $\mathcal{D}'$ constructs the remaining components $\matQ_1,\ldots,\matQ_\ell$ as described in the $\EVIPFE$ assumption.
     
     \item For each $i\in [\ell]$ such that $i>d$, algorithm $\mathcal{D}'$ samples $\vecs_i\rand\Z_q^n$. It then proceeds as follows:
     \begin{itemize}
          \item $\vecz_{1,i}$:
          \begin{itemize}
              \item For each $i< d$, sample $\vecz_{1,i}\rand\Z_q^m$.
              \item  For $i=d$, sample $\vece_{1,d}\leftarrow D_{\Z,\chi}^m$ and compute $\vecz_{1,d}^\top\leftarrow \vecy_1^\top+\vece_{1,d}^\top$.
              \item For each $i>d$, sample $\vece_{1,i}\leftarrow D_{\Z,\chi}^m, \tilde{\vece}_{1,i}\leftarrow D_{\Z,\chi_s}^m$, and compute $\vecz_{1,i}^\top\leftarrow \vecs_i^\top\matA_i+\tilde{\vece}_{1,i}^\top+\vece_{1,i}^\top$.
          \end{itemize}  
        
        \item $\vecz_{2}$: 
        \begin{itemize}
            \item For each $i< d$, sample $\tilde{\vecz}_{2,i}\rand\Z_q^{m'}, \tilde{\vecz}_{2,i}'\rand\Z_q^m$.
            \item For $i=d$, set $\tilde{\vecz}_{2,d}\leftarrow\vecy_2$ and $\tilde{\vecz}_{2,d}'\leftarrow\vecy_2'$.
            
            \item For each $i>d$, sample $\tilde{\vece}_{2,i}\leftarrow D_{\Z,\chi_s}^{m'}$ and $\tilde{\vece}_{2,i}'\leftarrow D_{\Z,\chi_s}^m$, and compute $\tilde{\vecz}_{2,i}^\top\leftarrow\vecs_i^\top\matB_i+\tilde{\vece}_{2,i}^\top$ and $\tilde{\vecz}_{2,i}'^\top\leftarrow\vecs_i^\top\matP_i+\tilde{\vece}_{2,i}'^\top$. Then sample $\vece_2\leftarrow D_{\Z,\chi}^{m'+m}$ and set $$\vecz_2^\top\leftarrow \left[\sum_{i\in [\ell]}\tilde{\vecz}_{2,i}^\top\,\Bigg|\,\sum_{i\in [\ell]}\tilde{\vecz}_{2,i}'^\top+\vecu^\top\matG\right]+\vece_2^\top.$$
        \end{itemize}  
        \item $\vecz_{3,i}^{(1)},\vecz_{3,i}^{(2)}$: Sample $\vecz_{3,i}^{(1)},\vecz_{3,i}^{(2)}$ via the same procedure as in $\H_{1,d}^\pre$. 
        \end{itemize}
     \item Finally, algorithm $\mathcal{D}'$ sends the challenge $\{1^\lambda,(\vecr_{\pub_1},\vecr_{\pub_2}),
    \{\matA_i,\vecz_{1,i}^\top\}_{i\in [\ell]},[\matB\mid\matP],\vecz_2^\top,\{\vecz_{3,i}^\top\}_{i\in [\ell]}\}$ to distinguisher $\mathcal{D}$ and outputs whatever $\mathcal{D}$ outputs.
    \end{enumerate}
We now analyze the distribution of the challenge given to $\mathcal{D}$. Note that the distributions of matrices $\matA_d,\matB_d,\matP_d$ match exactly those in the original game. Therefore, it suffices to analyze the distribution of vectors $\vecz_{1,i},\vecz_2,\vecz_{3,i}$. Observe that for all $i\neq d$, the components $\vecz_{1,i},\vecz_{2},\vecz_{3,i}$ are constructed identically in both $\H_{1,d-1}^\pre$ and $\H_{1,d}^{\pre}$. Hence, we only need to analyze the case of $i=d$. We now distinguish between the two possible forms of the \LWE challenge, depending on whether the challenge is structured or uniformly random.
\begin{itemize}
    \item If $\vecy_1=\vecs^\top\matA_d+\vece_1^\top,\vecy_2=\vecs^\top\matB_d+\vece_2^\top,\vecy_2'=\vecs^\top\matP_d+\vece_2'^\top$ for some $\vecs\rand\Z_q^n,\vece_1,\vece_2'\leftarrow D_{\Z,\chi_s}^m,\vece_2\leftarrow D_{\Z,\chi_s}^{m'}$. In this case, the values $\vecz_{1,d},\tilde{\vecz}_{2,d},\tilde{\vecz}_{2,d}'$ are distributed identically to those in $\H_{1,d-1}^\pre$.
    \item If $\vecy_1,\vecy_2'\rand\Z_q^m,\vecy_2\rand\Z_q^{m'}$, then $\vecz_{1,d},\tilde{\vecz}_{2,d},\tilde{\vecz}_{2,d}'$ are sampled according to the same distributions as in $\H_{1,d}^\pre$.
\end{itemize}
Consequently, $\mathcal{D}'$ provides a perfect simulation of the distribution of the corresponding hybrid. As a result, algorithm $\mathcal{D}'$ can distinguish the $\LWE$ instance from the uniform instance with the same advantage as $\mathcal{D}$ distinguishes between $\H_{1,d-1}^\pre$ and $\H_{1,d}^\pre$. This leads to a contradiction with the $\LWE$ assumption. We can conclude $\H_{1,d-1}^\pre\overset{c}{\approx}\H_{1,d}^\pre$ for each $d\in [\ell]$. By a standard hybrid argument, it follows that $\H_{1}^\pre\overset{c}{\approx}\H_{2}^\pre$, as desired.
\end{proof}

\begin{lemma}\label{evpre2}
    Suppose $m'>6n\log q$ and $\chi'=\Omega(\sqrt{n\log q})$. Let $\chi_s$ be an error parameter such that $\chi\geq \lambda^{\omega(1)}\cdot\sqrt{\lambda}\chi_s$ and the $\LWE_{n,m_1,q,\chi_s}$ assumption holds for $m_1=\poly(Q_0)$, where $Q_0$ is the upper bound on the total number of secret-key queries (including those in the partial set $\mathcal{Q}_{\partset}'$) submitted by the adversary. Then we have that $\H_{2}^\pre\overset{c}{\approx}\H_3^\pre$.
\end{lemma}

\begin{proof}
    We prove this claim by defining a sequence of intermediate hybrid games between $\H_2^{\pre}$ and $\H_{3}^\pre$. For each $d\in [N_1]$, we define the following games:
    
    \iitem{Game $\H_{2,d,0}^\pre$} This experiment is identical to $\H_{2}^\pre$, except for how the components $y_{1,j}$ are sampled.
    \begin{itemize}
        \item For each $j<d$, the challenger samples $\boxed{y_{1,j}\rand\Z_q}$.
        \item For $j\geq d$, the challenger computes $y_{1,j}$ exactly as in $\H_{2}^\pre$, i.e., $y_{1,j}\leftarrow \tilde{\vecz}_{2,1}'^\top\matG^{-1}(\vecv_{\rho_1(j)})+\tilde{\vecz}_{2,1}^\top\vecr_{\rho_1(j)}$.
    \end{itemize}

    \iitem{Game $\H_{2,d,1}^\pre$} This experiment is identical to $\H_{2,d,0}^\pre$, except for how it samples the values $\tilde{\vecz}_{2,i}$. Most of the notation follows the convention used in \cite{WWW22}.
    \begin{itemize}
        \item Denote $\gamma=\rho_1(d)$. Precisely, $\gamma$ is the index of the $d$-th Type I secret-key query involving authority $\aid_1^*$.
        \item We consider the $\gamma$-th Type I secret-key query $(\gid_\gamma,A_\gamma,\vecv_\gamma)$ submitted by the adversary $\mathcal{A}$. Since the restriction of the Type I secret-key query guarantees that $A_{\gamma}\subsetneqq A^*\cap\mathcal{N}$, there exists some index $1<i^*\leq \ell$ such that $\aid_{i^*}^*\notin A_\gamma$. Let $i^*$ be the smallest such index. In particular, $\aid_{i^*}^*\notin A_{\gamma}$.
        
        \item The challenger samples $\tilde{\tilde{\vecz}}_{2,1}\rand\Z_q^{m'},\tilde{\tilde{\vecz}}_{2,i^*}\rand\Z_q^{m'}$. Then sample $\vecs_0\rand\Z_q^{m'}$, and compute $\boxed{\tilde{\vecz}_{2,1}\leftarrow \tilde{\tilde{\vecz}}_{2,1}+\vecs_0}$, $\boxed{\tilde{\vecz}_{2,i^*}\leftarrow\tilde{\tilde{\vecz}}_{2,i^*}-\vecs_0}$. For $i\notin\{1,i^*\}$, the challenger samples $\tilde{\vecz}_{2,i}\rand\Z_q^{m'}$, as in $\H_{2,d,0}^\pre$. The components are then constructed as follows:
        \begin{itemize}
            \item $\vecz_2$: Set $$\vecz_2^\top\leftarrow \left[\sum_{i\in [\ell]}\tilde{\vecz}_{2,i}^\top\,\Bigg|\,\sum_{i\in [\ell]}\tilde{\vecz}_{2,i}'^\top+\vecu^\top\matG\right]+\vece_2^\top.$$ 
            \item $y_{i,j}$:
            For $i=1$ and $j<d$, sample $y_{1,j}\rand\Z_q$. Otherwise (i.e., for $i\neq 1$ or $j\geq d$), compute $y_{i,j}\leftarrow\tilde{\vecz}_{2,i}'^\top\matG^{-1}(\vecv_{\rho_i(j)})+\tilde{\vecz}_{2,i}^\top\vecr_{\rho_i(j)}$.
            \item $y_{i,j}'$: For each $i\in [\ell]$ and $j\in [Q']$, compute $y_{i,j}'$ as in $\H_{2,d,0}^\pre$, i.e., $y_{i,j}'\leftarrow\tilde{\vecz}_{2,i}'^\top\matG^{-1}(\vecv_{j}')+\tilde{\vecz}_{2,i}^\top\vecr_{j}'$.
             \end{itemize}
            \item In particular, we have the affected components in the following:
            \begin{itemize}
                \item For $i=1$ and $d\leq j\leq N_1$, we have \begin{align*}
y_{1,j}\leftarrow\tilde{\vecz}_{2,1}'^\top\matG^{-1}(\vecv_{\rho_1(j)})+\tilde{\vecz}_{2,1}^\top\vecr_{\rho_1(j)}=\tilde{\vecz}_{2,1}'^\top\matG^{-1}(\vecv_{\rho_1(j)})+\boxed{\tilde{\tilde{\vecz}}_{2,1}^\top\vecr_{\rho_1(j)}+\vecs_0^\top\vecr_{\rho_1(j)}}.\end{align*}
\item For $i=1$ and $j\in [Q']$, we have $$y_{1,j}'=\tilde{\vecz}_{2,1}'^\top\matG^{-1}(\vecv_{j}')+\tilde{\vecz}_{2,1}^\top\vecr_{j}'=\tilde{\vecz}_{2,1}'^\top\matG^{-1}(\vecv_{j}')+\boxed{\tilde{\tilde{\vecz}}_{2,1}^\top\vecr_{j}'+\vecs_0^\top\vecr_{j}'}.$$

\item For $i=i^*$ and $j\in [N_{i^*}]$, we have $$y_{i^*,j}=\tilde{\vecz}_{2,i^*}'^\top\matG^{-1}(\vecv_{\rho_{i^*}(j)})+\tilde{\vecz}_{2,i^*}^\top\vecr_{\rho_{i^*}(j)}=\tilde{\vecz}_{2,i^*}'^\top\matG^{-1}(\vecv_{\rho_{i^*}(j)})+\boxed{\tilde{\tilde{\vecz}}_{2,i^*}^\top\vecr_{\rho_{i^*}(j)}-\vecs_0^\top\vecr_{\rho_{i^*}(j)}}.$$

\item For $i=i^*$ and $j\in [Q']$, we have that  $$y_{i^*,j}'=\tilde{\vecz}_{2,i^*}'^\top\matG^{-1}(\vecv_{j}')+\tilde{\vecz}_{2,i^*}^\top\vecr_{j}'=\tilde{\vecz}_{2,i^*}'^\top\matG^{-1}(\vecv_{j}')+\boxed{\tilde{\tilde{\vecz}}_{2,i^*}^\top\vecr_{j}'-\vecs_0^\top\vecr_{j}'}.$$
\end{itemize}       
\end{itemize}

 \iitem {Game $\H_{2,d,2}^\pre$} The experiment is identical to $\H_{2,d,1}^\pre$, except for how $y_{i,j},y_{i,j}'$ are sampled. Specifically,
    \begin{itemize}
        \item The challenger samples $\tilde{e}_i \leftarrow D_{\Z,\chi_s}$ for each $i \in [Q]$, and $\tilde{e}_j' \leftarrow D_{\Z,\chi_s}$ for each $j \in [Q']$.
        \item For $i=1$ and $j\geq d$, set $$y_{1,j}\leftarrow\tilde{\vecz}_{2,1}'^\top\matG^{-1}(\vecv_{\rho_1(j)})+\tilde{\tilde{\vecz}}_{2,1}^\top\vecr_{\rho_1(j)}+(\vecs_0^\top\vecr_{\rho_1(j)}+\boxed{\tilde{e}_{\rho_1(j)}}).$$
        \item For $i=1$ and each $j\in [Q']$, set $$y_{1,j}'\leftarrow\tilde{\vecz}_{2,1}'^\top\matG^{-1}(\vecv_{j}')+\tilde{\tilde{\vecz}}_{2,1}^\top\vecr_{j}'+(\vecs_0^\top\vecr_{j}'+\boxed{\tilde{e}_{j}'}).$$
        \item For $i=i^*$, set \begin{align*}
y_{i^*,j}&\leftarrow\tilde{\vecz}_{2,i^*}'^\top\matG^{-1}(\vecv_{\rho_{i^*}(j)})+\tilde{\tilde{\vecz}}_{2,i^*}^\top\vecr_{\rho_{i^*}(j)}-(\vecs_0^\top\vecr_{\rho_{i^*}(j)}+\boxed{\tilde{e}_{\rho_{i^*}(j)}}),\quad\text{for all } j\in [N_{i^*}]\\
y_{i^*,j}'&\leftarrow\tilde{\vecz}_{2,i^*}'^\top\matG^{-1}(\vecv_{j}')+\tilde{\tilde{\vecz}}_{2,i^*}^\top\vecr_{j}'-(\vecs_0^\top\vecr_{j}'+\boxed{\tilde{e}_j'}),\quad\text{for all } j\in [Q'].
            \end{align*}
    \item For all other pairs $(i,j)$, the components $y_{i,j},y_{i,j}'$ are computed in the same way as in  $\H_{2,d,1}^\pre$. 
    \end{itemize}

    \iitem{Game $\H_{2,d,3}^\pre$} The experiment is identical to $\H_{2,d,2}^\pre$, except for how  $y_{i,j},y_{i,j}'$ are sampled. Specifically
    \begin{itemize}
        \item The challenger samples $\boxed{\delta_{i}\rand\Z_q}$ for each $i\in [Q]$, and $\boxed{\delta_{j}'\rand\Z_q}$ for each $j\in [Q']$.
        
        \item For $i=1$ and $j\geq d$, set $$y_{1,j}\leftarrow\tilde{\vecz}_{2,1}'^\top\matG^{-1}(\vecv_{\rho_1(j)})+\tilde{\tilde{\vecz}}_{2,1}^\top\vecr_{\rho_1(j)}+\boxed{\delta_{\rho_1(j)}}.$$
        
        \item For $i=1$ and each $j\in [Q']$, set $$y_{1,j}'\leftarrow\tilde{\vecz}_{2,1}'^\top\matG^{-1}(\vecv_{j}')+\tilde{\tilde{\vecz}}_{2,1}^\top\vecr_{j}'+\boxed{\delta_j'}.$$
        \item For $i=i^*$, set \begin{align*}
y_{i^*,j}&\leftarrow\tilde{\vecz}_{2,i^*}'^\top\matG^{-1}(\vecv_{\rho_{i^*}(j)})+\tilde{\tilde{\vecz}}_{2,i^*}^\top\vecr_{\rho_{i^*}(j)}+\boxed{\delta_{\rho_{i^*}(j)}},\quad\text{for all } j\in [N_{i^*}],\\
y_{i^*,j}'&\leftarrow\tilde{\vecz}_{2,i^*}'^\top\matG^{-1}(\vecv_{j}')+\tilde{\tilde{\vecz}}_{2,i^*}^\top\vecr_{j}'+\boxed{\delta_{j}'},\quad\text{for all } j\in [Q'].
            \end{align*}
    \item For all other pairs $(i,j)$, the components $y_{i,j},y_{i,j}'$ are computed in the same way as in $\H_{2,d,2}^\pre$. 
    \end{itemize}

\iitem{Game $\H_{2,d,4}^\pre$} The experiment is identical to $\H_{2,d,3}^\pre$ except it samples $\boxed{y_{1,d}\rand\Z_q}$.

\iitem {Game $\H_{2,d,5}^\pre$} The experiment is identical to $\H_{2,d,4}^\pre$, except for how the values $y_{i,j},y_{i,j}'$ are sampled. Specifically,
    \begin{itemize}
        \item The challenger samples $\tilde{e}_i \leftarrow D_{\Z,\chi_s}$ for each $i \in [Q]$, and $\tilde{e}_j' \leftarrow D_{\Z,\chi_s}$ for each $j \in [Q']$.
        \item For $i=1$ and $j=d$, the challenger samples $y_{1,d}\rand\Z_q$.
        \item For $i=1$ and $j> d$, set $$y_{1,j}\leftarrow\tilde{\vecz}_{2,1}'^\top\matG^{-1}(\vecv_{\rho_1(j)})+\tilde{\tilde{\vecz}}_{2,1}^\top\vecr_{\rho_1(j)}+\boxed{(\vecs_0^\top\vecr_{\rho_1(j)}+\tilde{e}_{\rho_1(j)})}.$$
        
        \item For $i=1$ and each $j\in [Q']$, set $$y_{1,j}'\leftarrow\tilde{\vecz}_{2,1}'^\top\matG^{-1}(\vecv_{j}')+\tilde{\tilde{\vecz}}_{2,1}^\top\vecr_{j}'+\boxed{(\vecs_0^\top\vecr_{j}'+\tilde{e}_{j}')}.$$
        
        \item For $i=i^*$, set \begin{align*}
y_{i^*,j}&\leftarrow\tilde{\vecz}_{2,i^*}'^\top\matG^{-1}(\vecv_{\rho_{i^*}(j)})+\tilde{\tilde{\vecz}}_{2,i^*}^\top\vecr_{\rho_{i^*}(j)}-\boxed{(\vecs_0^\top\vecr_{\rho_{i^*}(j)}+\tilde{e}_{\rho_{i^*}(j)})},\quad\text{for all } j\in [N_{i^*}],\\
y_{i^*,j}'&\leftarrow\tilde{\vecz}_{2,i^*}'^\top\matG^{-1}(\vecv_{j}')+\tilde{\tilde{\vecz}}_{2,i^*}^\top\vecr_{j}'-\boxed{(\vecs_0^\top\vecr_{j}'+\tilde{e}_j')},\quad\text{for all } j\in [Q'].
            \end{align*}
    \item For all other pairs $(i,j)$, the components $y_{i,j},y_{i,j}'$ are computed in the same way as in $\H_{2,d,4}^\pre$. 
    \end{itemize}

    \iitem{Game $\H_{2,d,6}^\pre$} This experiment is identical to $\H_{2,d,5}^\pre$, except for how the values $y_{i,j},y_{i,j}'$ are sampled. 
    \begin{itemize}
     \item For $i=1$ and $j>d$, set \begin{align*}
y_{1,j}\leftarrow\tilde{\vecz}_{2,1}'^\top\matG^{-1}(\vecv_{\rho_1(j)})+\tilde{\tilde{\vecz}}_{2,1}^\top\vecr_{\rho_1(j)}+\boxed{\vecs_0^\top\vecr_{\rho_1(j)}}.\end{align*}
    \item For $i=1$ and each $j\in [Q']$, set
$$y_{1,j}'\leftarrow\tilde{\vecz}_{2,1}'^\top\matG^{-1}(\vecv_{j}')+\tilde{\tilde{\vecz}}_{2,1}^\top\vecr_{j}'+\boxed{\vecs_0^\top\vecr_{j}'}.$$
 \item For $i=i^*$, set \begin{align*}
y_{i^*,j}&=\tilde{\vecz}_{2,i^*}'^\top\matG^{-1}(\vecv_{\rho_{i^*}(j)})+\tilde{\tilde{\vecz}}_{2,i^*}^\top\vecr_{\rho_{i^*}(j)}+\boxed{\vecs_0^\top\vecr_{\rho_{i^*}(j)}},\quad\text{for all } j\in [N_{i^*}],\\
y_{i^*,j}'&=\tilde{\vecz}_{2,i^*}'^\top\matG^{-1}(\vecv_{j}')+\tilde{\tilde{\vecz}}_{2,i^*}^\top\vecr_{j}'+\boxed{\vecs_0^\top\vecr_{j}'},\quad\text{for all } j\in [Q'].
            \end{align*}
\item  For all other pairs $(i,j)$, the components $y_{i,j},y_{i,j}'$ are computed in the same way as in $\H_{2,d,5}^\pre$. 
    \end{itemize}

    In the following, we present the proof that each adjacent pair of these intermediate hybrid games defined above is indistinguishable for all $d\in [N_1]$. By the hybrid argument, this will imply that $\H_2^\pre \overset{c}{\approx} \H_3^\pre$.

    \begin{claim}\label{evpre20}
        We have $\H_{2,d,0}^\pre\equiv \H_{2,d,1}^\pre$, for all $d\in [N_1]$. 
    \end{claim}

    \begin{proof}
        The only difference between the two experiments lies in the way the vectors $\tilde{\vecz}_{2,1}$ and $\tilde{\vecz}_{2,i^*}$ are sampled. In $\H_{2,d,0}^\pre$, they are sampled independently and uniformly from $\Z_q^{m'}$. In $\H_{2,d,1}^\pre$, they are constructed as $$\tilde{\vecz}_{2,1}^\top\leftarrow\tilde{\tilde{\vecz}}_{2,1}+\vecs_0, \tilde{\vecz}_{2,i^*}^\top\leftarrow\tilde{\tilde{\vecz}}_{2,i^*}-\vecs_0,$$ where $\tilde{\tilde{\vecz}}_{2,1},\tilde{\tilde{\vecz}}_{2,i^*},\vecs_0\rand\Z_q^{m'}$. It follows that the joint distribution of the pair $(\tilde{\vecz}_{2,1},\tilde{\vecz}_{2,i^*})$ in experiment $\H_{2,d,1}^\pre$ is distributed identically to the pair sampled in $\H_{2,d,0}^\pre$. This completes the proof.
    \end{proof}

    \begin{claim}\label{evpre21}
        Suppose that $\chi\geq \lambda^{\omega(1)}\cdot\sqrt{\lambda}\chi_s$. We have that $\H_{2,d,1}^{\pre}\overset{s}{\approx}\H_{2,d,2}^\pre$ for all $d\in [N_1]$.
    \end{claim}

    \begin{proof}
        The only difference between $\H_{2,d,1}^\pre$ and $\H_{2,d,2}^\pre$ lies in the modification of the error terms added to certain components $y_{i,j}$ and $y_{i,j}'$. We focus first on the error terms appearing in $\vecz_{3,i}^{(1)}$, as the reasoning for $\vecz_{3,i}^{(2)}$ is analogous. Let $z_{3,i,j}^{(1)}$ and $e_{3,i,j}^{(1)}$ denote the $j$-th entries of $\vecz_{3,i}^{(1)}$ and $\vece_{3,i}^{(1)}$, respectively. Then 
        \begin{itemize}
            \item In $\H_{2,d,1}^\pre$, for the affected components, we have \begin{align*}
    z_{3,1,j}^{(1)}&=\tilde{\vecz}_{2,1}'^\top\matG^{-1}(\vecv_{\rho_1(j)})+\tilde{\vecz}_{2,1}^\top\vecr_{\rho_1(j)}+e_{3,1,j}^{(1)},\quad \text{for } j>d\\
    z_{3,i^*,j}^{(1)}&=\tilde{\vecz}_{2,i^*}'^\top\matG^{-1}(\vecv_{\rho_{i^*}(j)})+\tilde{\vecz}_{2,i^*}^\top\vecr_{\rho_{i^*}(j)}+e_{3,i^*,j}^{(1)},\quad\text{for all } j\in [N_{i^*}].
    \end{align*}
            
    \item In $\H_{2,d,2}^\pre$, the corresponding components are computed as:\begin{align*}
z_{3,1,j}^{(1)}&=\tilde{\vecz}_{2,1}'^\top\matG^{-1}(\vecv_{\rho_1(j)})+\tilde{\vecz}_{2,1}^\top\vecr_{\rho_1(j)}+e_{3,1,j}^{(1)}+\tilde{e}_{\rho_1(j)},\quad \text{for } j>d\\
    z_{3,i^*,j}^{(1)}&=\tilde{\vecz}_{2,i^*}'^\top\matG^{-1}(\vecv_{\rho_{i^*}(j)})+\tilde{\vecz}_{2,i^*}^\top\vecr_{\rho_{i^*}(j)}+e_{3,i^*,j}^{(1)}-\tilde{e}_{\rho_{i^*}(j)}\quad \text{for all } j\in [N_{i^*}].
    \end{align*} 
      \end{itemize}
    Since $\chi\geq\lambda^{\omega(1)}\cdot\sqrt{\lambda}\chi_s$, Lemma \ref{truncated} ensures that $e_{3,1,j}^{(1)}+\tilde{e}_{\rho_1(j)}\overset{s}{\approx}e_{3,1,j}^{(1)}$, and similarly for $e_{3,i^*,j}^{(1)}-\tilde{e}_{\rho_{i^*}(j)}\overset{s}{\approx}e_{3,i^*,j}^{(1)}$. 
      
    For all other $(i,j)$, the distribution of $z_{3,i,j}^{(1)}$ remains unchanged between the two experiments. Hence we can conclude that the distribution of $\vecz_{3,i}^{(1)}$ is identical in both experiments for all $i\in [\ell]$. The same argument also applies to $\vecz_{3,i}^{(2)}$,  as the additional error terms are smudged in an analogous manner. This completes the proof.
    \end{proof}

    \begin{claim}\label{evpre22}
        Suppose that $m'>6n\log q$ and $\chi'=\Omega(\sqrt{n\log q})$. Suppose the $\LWE_{n,m_1,q,\chi_s}$ assumption holds for some $m_1=\poly(Q_0)$. Then we have that $\H_{2,d,2}^\pre\overset{c}{\approx}\H_{2,d,3}^\pre$ for all $d\in [N_1]$.
    \end{claim}

    \begin{proof}
        Assume for contradiction that there exists an efficient distinguisher $\mathcal{D}$ that can distinguish between $\H_{2,d,2}^\pre$ and $\H_{2,d,3}^\pre$ with non-negligible advantage. We construct an algorithm $\mathcal{D}'$ for the $\flipLWE_{m',m_1,q,\chi',\chi_s}$ assumption. Algorithm $\mathcal{D}'$ proceeds as follows: 
        \begin{enumerate}
 \item Algorithm $\mathcal{D}'$ begins by receiving a \flipLWE challenge $(\matD,\vecdelta)$, where $\matD\in\Z_q^{m'\times (Q+Q')}$ and $\vecdelta\in \Z_q^{Q+Q'}$. Then algorithm $\mathcal{D}'$ parses $\matD$ and $\vecdelta$ as: $$\matD=[\matD_1\mid \matD_1'], \quad\vecdelta^\top=[\vecdelta_1^\top\mid\vecdelta_1'^\top],$$ where $\matD_1\in\Z_q^{m'\times Q},\matD_1'\in\Z_q^{m'\times Q'}$, and $\vecdelta_1\in\Z_q^Q,\vecdelta_1'\in\Z_q^{Q'}$. Let $\vecd_{j},\vecd_j'$ denote the $j$-th column vectors of $\matD_1,\matD_1'$, respectively. Let $\delta_{j},\delta_j'$ denote the $j$-th entries of $\vecdelta_1,\vecdelta_1'$, respectively.

    \item Algorithm $\mathcal{D}'$ simulates $\mathcal{S}_{\mathcal{A},\vecv}$ as follows:
 \begin{itemize}
     \item It samples $\vecr_{\pub_1},\vecr_{\pub_2}\rand\{0,1\}^\kappa$ as the randomness for algorithm $\mathcal{S}_{\mathcal{A},\vecv}$. 
     \item Run algorithm $\mathcal{A}(1^\lambda;\vecr_{\pub_1})$ with randomness $\vecr_{\pub_1}$, and extract from its output: a set of Type I secret-key queries $\mathcal{Q}=\{(\gid,A,\vecv)\}$ and a partial set of Type II secret-key queries $\mathcal{Q}_{\partset}'=\{(\gid', A')\}$.
     
     \item Run $\samp_{\vecv}(1^\lambda;\vecr_{\pub_2})$ with randomness $\vecr_{\pub_2}$. It outputs $\vecv_1',\ldots,\vecv_{Q'}'$.
            
    \item Finally, outputs 
    \begin{align*}
             &1^{\ell},\{1^{N_i+Q'}\}_{i\in [\ell]};\\
             &(\vecd_{\rho_1(1)},\vecv_{\rho_1(1)}),\ldots,(\vecd_{\rho_1(N_1)},\vecv_{\rho_{1}(N_1)}),(\vecd_1',\vecv_1'),\ldots,(\vecd_{Q'}',\vecv_{Q'}');\\
            &(\vecd_{\rho_2(1)},\vecv_{\rho_2(1)}),\ldots,(\vecd_{\rho_2(N_2)},\vecv_{\rho_{2}(N_2)}),(\vecd_1',\vecv_1'),\ldots,(\vecd_{Q'}',\vecv_{Q'}'); \\
            &\cdots;\\
            &(\vecd_{\rho_\ell(1)},\vecv_{\rho_\ell(1)}),\ldots,(\vecd_{\rho_\ell(N_\ell)},\vecv_{\rho_{\ell}(N_\ell)}),(\vecd_1',\vecv_1'),\ldots,(\vecd_{Q'}',\vecv_{Q'}').
            \end{align*}
    \end{itemize}
    The only difference between the simulation and the original specification of $\mathcal{S}_{\mathcal{A},\vecv}$ lies in how the vectors $\vecr_i$ and $\vecr_i'$ are generated. In the simulation, the challenger replaces these vectors with $\vecd_i,\vecd_i'$ which are the columns of the \flipLWE matrix $\matD$. However, since both $\vecd_i,\vecd_i'$ are sampled from the same distribution $D_{\Z,\chi'}^{m'}$ as the original $\vecr_i,\vecr_i'$, this substitution does not affect the validity of the simulation.
     
     \item Algorithm $\mathcal{D}'$ samples $\vecr_{\pri}\rand\{0,1\}^\kappa$ and computes 
$\vecu\leftarrow\mathcal{S}_{\mathcal{A},\vecu}(1^\lambda,(\vecr_{\pub_1},\vecr_{\pub_2});\vecr_{\pri})$.
     
\item For each $i\in [\ell]$, it samples $\matA_i,\matP_i\rand\Z_q^{n\times m},\matB_i\rand\Z_q^{n\times m'}$. It then sets $$[\matB\mid\matP]\leftarrow\left[\begin{array}{c|c}
         \matB_1 &\matP_1  \\
          \vdots&\vdots \\
          \matB_\ell&\matP_{\ell}
     \end{array}\right]\in\Z_q^{n\ell\times (m'+m)}.$$ 
     
     \item Algorithm $\mathcal{D}'$ samples $\tilde{\vecz}_{2,i}\rand\Z_q^{m'},\tilde{\vecz}_{2,i}'\rand\Z_q^{m}$ for each $i\in [\ell]$, and constructs the following components:
     \begin{itemize}
         \item For $i=1,j<d$, sample $y_{1,j},y_{1,j}'\rand\Z_q$.
         \item For $i=1,j\geq d$, compute $$y_{1,j}\leftarrow \tilde{\vecz}_{2,1}'^\top\matG^{-1}(\vecv_{\rho_1(j)})+\tilde{\vecz}_{2,1}^\top\vecd_{\rho_1(j)}+\delta_{\rho_1(j)}.$$
         \item For $i=1,j\in [Q']$, compute  $$y_{1,j}'=\tilde{\vecz}_{2,1}'^\top\matG^{-1}(\vecv_{j}')+\tilde{\vecz}_{2,1}^\top\vecd_{j}'+\delta_{j}'.$$
         \item For $i=i^*$, compute \begin{align*}
             y_{i^*,j}&=\tilde{\vecz}_{2,i^*}'^\top\matG^{-1}(\vecv_{\rho_{i^*}(j)})+\tilde{\vecz}_{2,i^*}^\top\vecd_{\rho_{i^*}(j)}+\delta_{\rho_{i^*}(j)},\quad\text{ for all } j\in [N_{i^*}],\\
y_{i^*,j}'&=\tilde{\vecz}_{2,i^*}'^\top\matG^{-1}(\vecv_{j}')+\tilde{\vecz}_{2,i^*}^\top\vecd_j'+\delta_j',\quad\text{ for all } j\in [Q'].
         \end{align*}
\item For $i\notin \{1,i^*\}$, compute $y_{i,j}\leftarrow \tilde{\vecz}_{2,i}'^\top\matG^{-1}(\vecv_{\rho_{i}(j)})+\tilde{\vecz}_{2,i}^\top\vecd_{\rho_{i}(j)}$ and $y_{i,j}'\leftarrow \tilde{\vecz}_{2,i}'^\top\matG^{-1}(\vecv_{j}')+\tilde{\vecz}_{2,i}^\top\vecd_{j}'$.
     \end{itemize}
     \item Algorithm $\mathcal{D}'$ constructs the components $\vecz_{1,i},\vecz_2,\vecz_{3,i}$ as follows:
     \begin{itemize}
         \item $\vecz_{1,i}$: For each $i\in [\ell]$, sample $\vecz_{1,i}\rand\Z_q^m$.

         \item $\vecz_2$: Sample $\vece_2\leftarrow D_{\Z,\chi}^{m'+m}$, and set $$\vecz_2^\top\leftarrow\left[\sum_{i\in [\ell]}\tilde{\vecz}_{2,i}^\top\,\Bigg|\,\sum_{i\in [\ell]}\tilde{\vecz}_{2,i}'^\top+\vecu^\top\matG\right]+\vece_2^\top.$$

          \item $\vecz_{3,i}^\top=[\vecz_{3,i}^{(1)\top}\mid \vecz_{3,i}^{(2)\top}]$: For each $i\in [\ell]$, sample $\vece_{3,i}^{(1)}\leftarrow D_{\Z,\chi}^{N_i},\vece_{3,i}^{(2)}\leftarrow
    D_{\Z,\chi}^{Q'}$, and set 
        \begin{align*}
    \vecz_{3,i}^{(1)\top}&\leftarrow [y_{i,1}\mid\cdots\mid y_{i,N_i}]+\vece_{3,i}^{(1)\top},\\
            \vecz_{3,i}^{(2)\top}&\leftarrow [y'_{i,1}\mid\cdots\mid y'_{i,Q'}]+\vece_{3,i}^{(2)\top}.
        \end{align*}
     \end{itemize} 
    
     \item Finally, algorithm $\mathcal{D}'$ sends the challenge $\{1^\lambda,(\vecr_{\pub_1},\vecr_{\pub_2}),(\matA_i,\vecz_{1,i}^\top)_{i\in [\ell]},[\matB\mid\matP],\vecz_2^\top,\{\vecz_{3,i}^\top\}_{i\in [\ell]}\}$ to distinguisher $\mathcal{D}$ and outputs whatever $\mathcal{D}$ outputs.
        \end{enumerate}
        
    By inspecting the construction of $\mathcal{D}'$, we can observe that the components $$\{1^\lambda,(\vecr_{\pub_1},\vecr_{\pub_2}),(\matA_i,\vecz_{1,i}^\top)_{i\in [\ell]},[\matB\mid\matP]\}$$ are identically distributed in both experiments $\H_{2,d,2}^\pre$ and $\H_{2,d,3}^\pre$. Therefore, it suffices to analyze the distributions of the components $\vecz_2$ and $\vecz_{3,i}$ in two experiments. We distinguish between the following two cases depending on the form of the flipped $\LWE$ challenge:
    \begin{itemize}
        \item If $\vecdelta^\top=\vecs_0^\top\matD+\vece^\top$ for some $\vecs_0\rand\Z_q^{m'},\vece\leftarrow D_{\Z,\chi_s}^{Q+Q'}$, then by the construction procedure of $\vecz_{2}$ and $\vecz_{3,i}$, the simulated distribution is identical to that in experiment $\H_{2,d,2}^\pre$.
        
        \item If $\vecdelta\rand\Z_q^{Q+Q'}$,  then the resulting distributions of $\vecz_2,\vecz_{3,i}$ are exactly the same as in experiment $\H_{2,d,3}^\pre$. Hence the simulation is exactly as in experiment $\H_{2,d,3}^\pre$. 
    \end{itemize}
    In conclusion, distinguisher $\mathcal{D}'$ constructed as above can distinguish flipped \LWE sample from uniform samples with the same advantage as $\mathcal{D}$ distinguishes between $\H_{2,d,2}^\pre$ and $\H_{2,d,3}^\pre$. However, by corollary \ref{flipped}, under the assumption $\LWE_{n,m_1,q,\chi_s}$ and given parameter constraints, the assumption $\flipLWE_{m',m_1,q,\chi',\chi_s}$ must hold. This leads to a contradiction, implying that $\H_{2,d,2}^\pre\overset{c}{\approx}\H_{2,d,3}^\pre$, which completes the proof.
    \end{proof}

    \begin{claim}\label{evpre23}
     We have that $\H_{2,d,3}^\pre\equiv \H_{2,d,4}^\pre$ for all $d\in [N_1]$.
    \end{claim}

    \begin{proof}
        The only difference between experiments $\H_{2,d,3}^\pre$ and $\H_{2,d,4}^\pre$ lies in the distribution of $y_{1,d}$. First, we argue that $\rho_1(d)\neq \rho_{i^*}(j)$ for all $j\in [N_{i^*}]$. This follows from the construction of $i^*$. Recall $i^*$ was defined as the smallest index such that $\aid_{i^*}\notin A_{\rho_{1}(d)}$. Since $\rho_{i^*}(j)$ indexes secret-key queries that contain $\aid_{i^*}$, it follows that $A_{\rho_{i^*}(j)}\neq A_{\rho_1(d)}$ for all $j\in [N_{i^*}]$. In particular, $\rho_{i^*}(j)\neq \rho_1(d)$ for all $j\in [N_{i^*}]$.  This implies that the only component that involves $\delta_{\rho_1(d)}$ in experiment $\H_{2,d,3}^\pre$ is $y_{1,d}$. Since $\delta_{\rho_1(d)}$ is sampled from the uniform distribution and is independent of other elements, the replacement of $y_{1,d}$ with a uniformly random value in $\H_{2,d,4}^\pre$ does not affect the overall distribution of the experiment, completing the proof.
    \end{proof}

    \begin{claim}\label{evpre24}
        Suppose that $m'>6n\log q$ and $\chi'=\Omega(\sqrt{n\log q})$. Suppose the $\LWE_{n,m_1,q,\chi_s}$ assumption holds for some $m_1=\poly(Q_0)$. Then we have that $\H_{2,d,4}^\pre\equiv \H_{2,d,5}^\pre$ for all $d\in [N_1]$.
    \end{claim}

    \begin{proof}\label{evpre25}
        The proof follows essentially the same argument as that of Claim \ref{evpre22}.
    \end{proof}

    \begin{claim}
        Suppose that $\chi\geq \lambda^{\omega(1)}\cdot\sqrt{\lambda}\chi_s$. We have that $\H_{2,d,5}^\pre\overset{s}{\approx}\H_{2,d,6}^\pre$ for all $d\in [N_1]$.
    \end{claim}

    \begin{proof}
        The proof follows essentially the same argument as that Claim \ref{evpre21}.
    \end{proof}

    \begin{claim}\label{pre6}
        For all $d\in [N_1]$, we have $\H_{2,d-1,6}^\pre\equiv \H_{2,d,0}^\pre$.
    \end{claim}

    \begin{proof}
        The proof follows essentially the same argument as that of Claim \ref{evpre20}.
    \end{proof}

    \noindent\emph{Proof of Lemma \ref{evpre2} (Continued)}. 
   By Claims \ref{evpre20} through \ref{pre6}, we obtain that for all $d \in [N_1]$, $\H_{2,d-1,0}^\pre\overset{c}{\approx}\H_{2,d,0}^\pre$ for all $d\in [N_1]$. Applying a standard hybrid argument over the sequence of hybrids, $\H_{2,1,0}^\pre\overset{c}{\approx}\H_{2,N_1,0}^\pre$.
    By definition of these hybrid experiments above, $\H_{2,1,0}^\pre\equiv\H_{2}^\pre$ and $\H_{2,N_1,0}^\pre\equiv \H_3^\pre$, which together yields $\H_2^\pre\overset{c}{\approx}\H_3^\pre$, as claimed.
\end{proof}

\begin{lemma}\label{evpre3}
    Suppose that $\chi\geq \lambda^{\omega(1)}\cdot\sqrt{\lambda}\chi_s$ and $\samp=(\samp_{\vecv},\samp_{\vecu})$ produces pseudorandom noisy inner product with noise parameter $\chi_s$. We have that $\H_{3}^\pre\overset{c}{\approx}\H_4^\pre$.
\end{lemma}

\begin{proof}
    We prove the lemma by defining a sequence of hybrid experiments between $\H_{3}^\pre$ and $\H_4^\pre$.
    
    \iitem{Game $\H_{3,0}^\pre$}: The experiment is identical to $\H_3^\pre$, except for how it samples $\tilde{\vecz}_{2,1}'$. Specifically, the challenger first samples $\tilde{\tilde{\vecz}}_{2,1}'\rand\Z_q^m$, and sets $\boxed{\tilde{\vecz}_{2,1}'\leftarrow \tilde{\tilde{\vecz}}_{2,1}'-\vecu^\top\matG}$. This modification affects the computation of the following components:
    \begin{itemize}
        \item $\vecz_{2}$: We have $$\boxed{\vecz_2^\top\leftarrow\left[\sum_{i\in [\ell]}\tilde{\vecz}_{2,i}^\top\,\Bigg|\,\sum_{i\in [\ell]}\tilde{\vecz}_{2,i}'^\top+\vecu^\top\matG\right]=\left[\sum_{i\in [\ell]}\tilde{\vecz}_{2,i}^\top\,\Bigg|\,\ \tilde{\tilde{\vecz}}_{2,1}'^\top+\sum_{2\leq i\leq\ell}\tilde{\vecz}_{2,i}'^\top\right]}.$$
        \item $y_{1,j}'$: We have $$y_{1,j}'=\tilde{\vecz}_{2,1}'^\top\matG^{-1}(\vecv_j')+\tilde{\vecz}_{2,1}^\top\vecr_{j}'=\tilde{\tilde{\vecz}}_{2,1}'^\top\matG^{-1}(\vecv_j')+\tilde{\vecz}_{2,1}^\top\vecr_{j}'-\vecu^\top\vecv_j'$$ for each $j\in [Q']$. 
    \end{itemize}

    \iitem{Game $\H_{3,1}^\pre$} The experiment is identical to $\H_{3,0}^\pre$, except for the procedure how it samples $y_{1,j}'$. Specifically, the term $\vecu^\top\vecv_{j}$ is perturbed by an additional noise term.
    \begin{itemize}
        \item For each $j\in [Q']$, the challenger first samples $e_j\rand D_{\Z,\chi_s}$, and sets $$y_{1,j}'\leftarrow\tilde{\tilde{\vecz}}_{2,1}'^\top\matG^{-1}(\vecv_j')+\tilde{\vecz}_{2,1}^\top\vecr_{j}'-(\vecu^\top\vecv_j'+\boxed{e_j}).$$
    \end{itemize}
    \iitem {Game $\H_{3,2}^\pre$} The experiment is identical to $\H_{3,1}^\pre$ except for the procedure how it samples $y_{1,j}'$.
    \begin{itemize}
        \item For each $j\in [Q']$, the challenger first samples $\omega_j\rand\Z_q$, and sets $$y_{1,j}'\leftarrow\tilde{\tilde{\vecz}}_{2,1}'^\top\matG^{-1}(\vecv_j')+\tilde{\vecz}_{2,1}^\top\vecr_{j}'+\boxed{\omega_j}.$$
    \end{itemize}
    \iitem{Game $\H_{3,3}^\pre$} The experiment is identical to $\H_{3,2}^\pre$ except it samples $\boxed{y_{1,j}'\rand\Z_q}$ for each $j\in [Q']$.
\\

    From the constructions above, we first observe that $\H_{3,0}^\pre \equiv \H_{3}^\pre$. The only difference between these two experiments lies in the way $\tilde{\vecz}_{2,1}'$ is sampled, which is set as $\tilde{\tilde{\vecz}}_{2,1}' - \vecu^\top\matG$ in $\H_{3,0}^\pre$. However, since $\tilde{\tilde{\vecz}}_{2,1}'$ is sampled uniformly from $\Z_q^m$, it follows that $\tilde{\vecz}_{2,1}'$ remains uniformly distributed over $\Z_q^m$ and independent of all other components. Therefore, the joint distribution of all variables in $\H_{3,0}^\pre$ remains identical to that in $\H_{3}^\pre$. 
    
    Similarly, we have that $\H_{3,3}^\pre \equiv \H_{4}^\pre$. This equivalence follows directly from the fact that the only symbolic change between the two experiments is the replacement of $\tilde{\tilde{\vecz}}_{2,1}'$ with $\tilde{\vecz}_{2,1}'$. Since both variables are sampled uniformly from $\Z_q^m$, this substitution does not affect the distribution of any component in the experiment. Since all values $y_{1,j}'$ are sampled uniformly at random in experiments $\H_{3,3}^\pre$ and $\H_4^\pre$, the overall distributions of the two experiments are identical.

    In the following, we present the proof that each adjacent pair of these intermediate hybrids $\H_{3,i}^\pre$ and $\H_{3,i+1}^\pre$ are indistinguishable for $i=0,1,2$. By applying a standard hybrid argument, these results yield that $\H_{3}^\pre\overset{c}{\approx}\H_4^\pre$, as desired. 

    \begin{claim}\label{evpre30}
        Suppose that $\chi\geq\lambda^{\omega(1)}\cdot\chi_s$. Then we have that $\H_{3,0}^\pre\overset{s}{\approx}\H_{3,1}^\pre$.
    \end{claim}

    \begin{proof}
        The proof follows essentially the same noise smudging argument as that of Claim \ref{evpre21}.
    \end{proof}
    
    \begin{claim}\label{evpre31}
        Suppose that the sampler $\samp=(\samp_{\vecv},\samp_{\vecu})$ produces pseudorandom noisy inner products with noise parameter $\chi_s$. Then we have $\H_{3,1}^\pre\overset{c}{\approx}\H_{3,2}^\pre$.
    \end{claim}

    \begin{proof}
        Assume for contradiction that there exists an efficient distinguisher $\mathcal{D}$ that can distinguish $\H_{3,1}^\pre$ and $\H_{3,2}^\pre$ with non-negligible advantage. We use $\mathcal{D}$ to construct an algorithm $\mathcal{D}'$ that breaks the pseudorandomness of noisy inner products produced by the sampler $\samp$.
        Algorithm $\mathcal{D}'$ proceeds as follows:
        \begin{enumerate}
            \item Algorithm $\mathcal{D}'$ begins by receiving a challenge instance $(\vecr_{\pub},\{\beta_j\}_{j\in [Q']})$.

             \item Algorithm $\mathcal{D}'$ simulates the sampler $\mathcal{S}_{\mathcal{A},\vecv}$ as follows:
 \begin{itemize}
     \item It samples $\vecr_{\pub_1}\rand\{0,1\}^\kappa$ as part of the randomness for algorithm $\mathcal{S}_{\mathcal{A},\vecv}$. 
     \item It runs algorithm $\mathcal{A}(1^\lambda;\vecr_{\pub_1})$ with randomness $\vecr_{\pub_1}$, and extracts from its output: a set of Type I secret-key queries $\mathcal{Q}=\{(\gid,A,\vecv)\}$ and a partial set of Type II secret-key queries $\mathcal{Q}_{\partset}'=\{(\gid', A')\}$.

     \item For each $i\in [Q]$, sample $\vecr_{i}\leftarrow D_{\Z,\chi'}^{m'}$, and for each $j\in [Q']$, sample $\vecr_{j}'\leftarrow D_{\Z,\chi'}^{m'}$.
     
    \item Run $\samp_{\vecv}(1^\lambda;\vecr_{\pub})$ with randomness $\vecr_{\pub}$. It outputs $\vecv_1',\ldots,\vecv_{Q'}'$.
            
    \item Finally, outputs 
    \begin{align*}
             &1^{\ell},\{1^{N_i+Q'}\}_{i\in [\ell]};\\
             &(\vecr_{\rho_1(1)},\vecv_{\rho_1(1)}),\ldots,(\vecr_{\rho_1(N_1)},\vecv_{\rho_{1}(N_1)}),(\vecr_1',\vecv_1'),\ldots,(\vecr_{Q'}',\vecv_{Q'}');\\
            &(\vecr_{\rho_2(1)},\vecv_{\rho_2(1)}),\ldots,(\vecr_{\rho_2(N_2)},\vecv_{\rho_{2}(N_2)}),(\vecr_1',\vecv_1'),\ldots,(\vecr_{Q'}',\vecv_{Q'}'); \\
            &\cdots;\\
            &(\vecr_{\rho_\ell(1)},\vecv_{\rho_\ell(1)}),\ldots,(\vecr_{\rho_\ell(N_\ell)},\vecv_{\rho_{\ell}(N_\ell)}),(\vecr_1',\vecv_1'),\ldots,(\vecr_{Q'}',\vecv_{Q'}').
            \end{align*}
    \end{itemize}
    The only difference between this simulation and the original specification of $\mathcal{S}_{\mathcal{A},\vecv}$ lies in how the randomness $\vecr_{\pub}$ is generated. In the simulation, algorithm $\mathcal{D}'$ uses the challenge-provided $\vecr_{\pub}$ instead of a freshly sampled $\vecr_{\pub_2}$. However, since $\vecr_{\pub}$ is uniformly distributed, this substitution does not alter the distribution of any component. Thus, the overall output remains consistent with that of the original specification of $\mathcal{S}_{\mathcal{A},\vecv}$.

    \item Algorithm $\mathcal{D}'$ samples $\matA_i,\matP_i\rand\Z_q^{n\times m},\matB_i\rand\Z_q^{n\times m'}$  for all $i\in [\ell]$. It then sets $$[\matB\mid\matP]\leftarrow\left[\begin{array}{c|c}
         \matB_1 &\matP_1  \\
          \vdots&\vdots \\
          \matB_\ell&\matP_{\ell}
     \end{array}\right]\in\Z_q^{n\ell\times (m'+m)}.$$ 
    
    \item Then algorithm $\mathcal{D}'$ constructs the components as follows:
        \begin{itemize}
         \item $\vecz_{1,i}$: For each $i\in [\ell]$, sample $\vecz_{1,i}\rand\Z_q^m$.
        
        \item $\vecz_2$: For each $i\in [\ell]$, sample $\tilde{\vecz}_{2,i}\rand\Z_q^{m'}$. Then for each $2\leq i\leq \ell$, sample $\tilde{\vecz}_{2,i}'\rand\Z_q^m$, and sample $\tilde{\tilde{\vecz}}_{2,1}'\rand\Z_q^m$. Set $$\vecz_2^\top\leftarrow\left[\sum_{i\in [\ell]}\tilde{\vecz}_{2,i}^\top\,\Bigg|\,\ \tilde{\tilde{\vecz}}_{2,1}'^\top+\sum_{2\leq i\leq\ell}\tilde{\vecz}_{2,i}'^\top\right].$$
        
        \item $\vecz_{3,i}^{(1)}$: It is computed in the same way as in experiment $\H_{3,1}^\pre$.
        
        \item $\vecz_{3,i}^{(2)}$: 
        \begin{itemize}
            \item For $i=1$, compute $$y_{1,j}'\leftarrow\tilde{\tilde{\vecz}}_{2,1}'^\top\matG^{-1}(\vecv_j')+\tilde{\vecz}_{2,i}^\top\vecr_{j}'+\beta_j$$ for each $j\in [Q']$. 
            \item For $i>1$, the sampling procedure for $y_{i,j}'$ is the same as in experiment $\H_{3,1}^\pre$ and $\H_{3,2}^\pre$.
            \item  Finally, it samples $\vece_{3,i}^{(2)}\leftarrow D_{\Z,\chi}^{Q'}$ for each $i\in [\ell]$, and sets $$\vecz_{3,i}^{(2)\top}\leftarrow [y_{i,1}'\mid\cdots\mid y_{i,Q'}']+\vece_{3,i}^{(2)\top}.$$
        \end{itemize}
    \end{itemize}
        
    \item The algorithm $\mathcal{D}'$ sends the challenge $\{1^\lambda,(\vecr_{\pub_1},\vecr_{\pub}), (\matA_i,\vecz_{1,i}^\top)_{i\in [\ell]},[\matB\mid\matP],\vecz_2^\top,\{\vecz_{3,i}^\top\}_{i\in [\ell]}\}$ to the distinguisher $\mathcal{D}$ and outputs whatever $\mathcal{D}$ outputs.
        \end{enumerate}
     From the definition of the algorithm $\mathcal{D}$, it is clear that the components $$\{1^\lambda,(\vecr_{\pub_1},\vecr_{\pub}), (\matA_i,\vecz_{1,i}^\top)_{i\in [\ell]},[\matB\mid \matP],\vecz_{2}^\top,\{\vecz_{3,i}^{(1)\top}\}_{i\in[\ell]},\{\vecz_{3,i}^{(2)}\}_{2\leq i\leq\ell}\}$$ are exactly distributed as in $\H_{3,1}^\pre$ and $\H_{3,2}^\pre$. It remains to consider the distribution of $\vecz_{3,1}^{(2)}$ in two experiments. 
    \begin{itemize}
        \item If $\beta_j=\vecu^\top\vecv_j'+e_j$ for some $\vecu\leftarrow \mathcal{S}_{\vecu}(1^\lambda,(\vecr_{\pub_1},\vecr_{\pub});\vecr_{\pri}),e_j\leftarrow D_{\Z,\chi_s}$, then by the construction procedure, each $y_{1,j}'$ is computed exactly as in $\H_{3,1}^\pre$, and thus the resulting distribution of  $\vecz_{3,1}^{(2)}$ matches that in $\H_{3,1}^\pre$.
        
        \item If $\beta_j\rand\Z_q$, then the resulting distribution of  $\vecz_{3,1}^{(2)}$ is identical to that in experiment $\H_{3,2}^\pre$.
    \end{itemize}
    Therefore, algorithm $\mathcal{D}'$ breaks the pseudorandomness of the noisy inner product produced by $\samp$ with the same advantage as $\mathcal{D}$ distinguishes between $\H_{3,1}^\pre$ and $\H_{3,2}^\pre$. This leads to a contradiction with the assumption that $\samp$ produces pseudorandom noisy inner products. Hence $\H_{3,1}^\pre\overset{c}{\approx}\H_{3,2}^\pre$, as claimed.
    \end{proof}

    \begin{claim}\label{evpre32}
       We have that $\H_{3,2}^\pre\equiv \H_{3,3}^\pre$. 
    \end{claim}

    \begin{proof}
        The claim follows directly from the observation that in experiment $\H_{3,2}^\pre$, the only component depending on $\omega_j$ is $y_{1,j}'$. Since $\omega_j$ is sampled uniformly and independent of all other components, it follows that $y_{1,j}'$ is also uniform and independent of all other components. This matches precisely the distribution of $y_{1,j}'$ in $\H_{3,3}^\pre$.
    \end{proof}

     \noindent\emph{Proof of Lemma \ref{evpre3} (Continued)}. 
    By applying a sequence of hybrid arguments from Claim \ref{evpre30} to Claim \ref{evpre32}, we obtain that $\H_{3,0}^\pre\overset{c}{\approx}\H_{3,3}^\pre$. Since we have already established that $\H_{3,0}^\pre\equiv\H_3^\pre$ and $\H_{3,3}^\pre\equiv \H_4^\pre$, which completes the argument that $\H_3^\pre\overset{c}{\approx}\H_4^\pre$.
\end{proof}

\begin{lemma}\label{evhyb4}
     Suppose $m'>6n\log q$ and $\chi'=\Omega(\sqrt{n\log q})$. Let $\chi_s$ be an error parameter such that $\chi\geq \lambda^{\omega(1)}\cdot\sqrt{\lambda}\chi_s$ and suppose that the $\LWE_{n,m_1,q,\chi_s}$ assumption holds for $m_1=\poly(Q_0)$, where $Q_0$ is the upper bound on the total number of secret-key queries (including those in the partial set $\mathcal{Q}_{\partset}'$) submitted by the adversary. Then we have that $\H_{4}^\pre\overset{c}{\approx}\H_5^\pre$.
\end{lemma}

\begin{proof}
     We prove this lemma by defining a sequence of intermediate hybrid experiments between $\H_4^{\pre}$ and $\H_{5}^\pre$. For each $1<d\leq\ell+1$, we introduce the following sequence of games:
     
     \iitem{Game $\H_{4,d,0}^\pre$} The experiment is identical to $\H_4^\pre$, except for how it samples $\vecz_2$ and $y_{i,j},y_{i,j}'$. Specifically, the challenger first samples $\boxed{\vecz_2\rand\Z_q^{m'+m}}$, and then sets $y_{i,j}$ and $y_{i,j}'$ as follows:
     \begin{itemize}
         \item For each $i<d$, sample $\boxed{y_{i,j}\rand \Z_q}$ for each $j\in [N_i]$ and $\boxed{y_{i,j}'\rand\Z_q}$ for each $j\in [Q']$.
         \item For each $i\geq d$, set \begin{align*}
             y_{i,j}&\leftarrow \tilde{\vecz}_{2,i}'^\top\matG^{-1}(\vecv_{\rho_i(j)})+\tilde{\vecz}_{2,i}^\top\vecr_{\rho_{i}(j)},\quad \text{for all } j\in [N_i],\\
             y_{i,j}'&\leftarrow \tilde{\vecz}_{2,i}'^\top\matG^{-1}(\vecv_{j}')+\tilde{\vecz}_{2,i}^\top\vecr_{j}',\quad \text{for all } j\in [Q'].
         \end{align*}
     \end{itemize}

     \iitem{Game $\H_{4,d,1}^\pre$} The experiment is identical to $\H_{4,d,0}^\pre$, except for how it samples $y_{d,j},y_{d,j}'$. The challenger first samples $e_{i}\leftarrow D_{\Z,\chi_s}$ for each $i\in [Q]$ and $e_j'\leftarrow D_{\Z,\chi_s}$ for each $j\in [Q']$, then sets $y_{i,j}$ and $y_{i,j}'$ as follows:
     \begin{itemize}
         \item For each $i<d$, sample $y_{i,j}\rand \Z_q$ for each $j\in [N_i]$ and $y_{i,j}'\rand\Z_q$ for each $j\in [Q']$ as in $\H_{4,d,0}^\pre$.
         
         \item For $i=d$, set \begin{align*}
             y_{d,j}&\leftarrow \tilde{\vecz}_{2,d}'^\top\matG^{-1}(\vecv_{\rho_d(j)})+(\tilde{\vecz}_{2,d}^\top\vecr_{\rho_{d}(j)}+\boxed{e_{\rho_d(j)}}),\quad \text{for all } j\in [N_d],\\
         y_{d,j}'&\leftarrow \tilde{\vecz}_{2,d}'^\top\matG^{-1}(\vecv_{j}')+(\tilde{\vecz}_{2,d}^\top\vecr_{j}'+\boxed{e_j'}),\quad \text{for all } j\in [Q'].
         \end{align*}
         
         \item For each $i> d$, compute $y_{i,j},y_{i,j}'$ in the same way as in $\H_{4,d,0}^\pre$.
        \end{itemize}
        
         \iitem{Game $\H_{4,d,2}^\pre$} The experiment is identical to $\H_{4,d,1}^\pre$ except that it replaces the error terms in $y_{d,j},y_{d,j}'$ with uniformly random values. The challenger first samples $
         \delta_i\rand\Z_q$ for each $i\in [Q]$, and $\delta_j'\rand\Z_q$ for each $j\in [Q']$. Then it sets $y_{d,j}$ and $y_{d,j}'$ as follows:
         \begin{align*}
                 y_{d,j}&\leftarrow \tilde{\vecz}_{2,d}'^\top\matG^{-1}(\vecv_{\rho_d(j)})+\boxed{\delta_{\rho_d(j)}},\quad \text{for all } j\in [N_d],\\
                 y_{d,j}'&\leftarrow \tilde{\vecz}_{2,d}'^\top\matG^{-1}(\vecv_{j}')+\boxed{\delta_{j}'},\quad \text{for all } j\in [Q'].
             \end{align*}
             
    \iitem{Game $\H_{4,d,3}^\pre$} The experiment is identical to $\H_{4,d,2}^\pre$ except for the values $y_{d,j},y_{d,j}'$ are sampled uniformly at random. Specifically,
     \begin{itemize}
         \item The challenger samples $\boxed{y_{d,j}\rand \Z_q}$ for each $j\in [N_d]$ and $\boxed{y_{d,j}'\rand\Z_q}$ for each $j\in [Q']$.
     \end{itemize}
    
    From the construction above, it follows that $\H_{4,2,0}^\pre\equiv \H_{4}^\pre$. Since the only component dependent on $\tilde{\vecz}_{2,1}$ and $\tilde{\vecz}_{2,1}'$ is $\vecz_2$, directly sampling $\vecz_2 \rand \Z_q^{m'+m}$ does not affect the overall distribution. Moreover, it is straightforward to verify that $\H_{4,\ell+1,0}^\pre \equiv \H_5^\pre$. In the following, we prove that each adjacent pair of hybrids is computationally indistinguishable for all $2\leq d\leq \ell$.

    \begin{claim}\label{evpre40}
        Suppose that $\chi\geq \lambda^{\omega(1)}\cdot\sqrt{\lambda}\chi_s$. Then we have that $\H_{4,d,0}^\pre\overset{s}{\approx}\H_{4,d,1}^\pre$ for $2\leq d\leq\ell$.
    \end{claim}

    \begin{proof}
        The proof follows essentially the same argument as that of Claim \ref{evpre21}.
    \end{proof}

    \begin{claim}\label{evpre41}
        Suppose that $m'>6n\log q$ and $\chi'=\Omega(\sqrt{n\log q})$, and suppose the $\LWE_{n,m_1,q,\chi_s}$ assumption holds for some $m_1=\poly(Q_0)$. Then we have $\H_{4,d,1}^\pre\overset{c}{\approx}\H_{4,d,2}^\pre$ for $2\leq d\leq \ell$.
    \end{claim}

    \begin{proof}
        The proof follows essentially the same argument as that of Claim \ref{evpre22}. 
    \end{proof}

\begin{claim}\label{evpre42}
    For all $2\leq d\leq \ell$, we have $\H_{4,d,2}^\pre\equiv\H_{4,d,3}^\pre$.
\end{claim}

\begin{proof}
    This follows directly from the fact that $\delta_j,\delta_j'$ are sampled uniformly and independently over $\Z_q$. Thus, replacing them with fresh uniform values does not alter the distributions of $y_{d,j}$ and $y_{d,j}'$.
\end{proof}

\begin{claim}\label{evpre43}
    For all $2\leq d\leq \ell$, we have $\H_{4,d,3}^\pre\equiv\H_{4,d+1,0}^\pre$.
\end{claim}

\begin{proof}
    This follows directly by the construction of the two experiments. $\H_{4,d,3}^\pre$ ends with $y_{d,j},y_{d,j}'$ sampled uniformly at random, which matches the initial sampling procedure of $\H_{4,d+1,0}^\pre$ for index $d$.
\end{proof}

\noindent\emph{Proof of Lemma \ref{evhyb4} (Continued)}. 
    By Claims \ref{evpre40} through \ref{evpre43}, we obtain that $\H_{4,d,0}^\pre \overset{c}{\approx} \H_{4,d+1,0}^\pre$ for all $2 \leq d \leq \ell$.
    By applying a standard hybrid argument, we can conclude that $\H_4^\pre\equiv\H_{4,2,0}^\pre\overset{c}{\approx}\H_{4,\ell+1,0}^\pre\equiv\H_5^\pre$, completing the proof.   
\end{proof}

\noindent\emph{Proof of Claim \ref{evpre} (Continued)}. By combining Lemmas \ref{evpre0}, \ref{evpre1}, \ref{evpre2}, and \ref{evpre3}, and applying a standard hybrid argument, we conclude that $\H_0^\pre\overset{c}{\approx}\H_5^\pre$, thereby completing the proof.
\end{proof}

\subsection{Parameters}
Let $\lambda$ be the security parameter.
\begin{itemize}
    \item We set the lattice dimension $n=\lambda^{1/\epsilon}$ and the modulus $q=2^{\tilde{O}(n^{\epsilon})}=2^{\tilde{O}(\lambda)}$ for some constant $\epsilon>0$, where $\tilde{O}(\cdot)$ suppresses constant and logarithmic factors. Then $m=n\lceil\log q\rceil=\tilde{O}(\lambda^{1+1/\epsilon})$.
    \item We set $L=2^\lambda$ to support ciphertext policies of arbitrary polynomial size, i.e., $\ell=\poly(\lambda)$.         
    \item For the static security (Theorem~\ref{security1}), we set $m'=\tilde{O}(\lambda^{1+1/\epsilon}),Q_0=\poly(\lambda)$ and $\chi'=\Theta(\lambda^{1+1/\epsilon})$. The construction relies on the $\LWE_{n,m_1,q,\chi_s}$ assumption and the $\INDIPFE_{n,m,m',q,\chi,\chi}$ assumption, where the noise parameter $\chi_s=\chi_s(\lambda)$ is set polynomial-bounded and $\chi$ satisfies the bound $\chi\geq \lambda^{\omega(1)}\cdot (\sqrt{\lambda}(m+\ell)\chi_s+\lambda m\chi'\chi_s)$. This requirement is satisfied by setting $\chi=2^{\tilde{O}(n^{\epsilon})}$.  
    
    \item To ensure correctness (Theorem~\ref{correct1}), we require $\chi > \sqrt{n \log q} \cdot \omega(\sqrt{\log n})$, which is also satisfied by setting $\chi = 2^{\tilde{O}(n^{\epsilon})}$. In addition, we have the relationship $B_0=\sqrt{\lambda}\chi m+\lambda \chi\chi' m'+\lambda \chi^2 mL$, which implies $B_0=2^{\tilde{O}(n^{\epsilon})}$. Then we obtain $B_0/(\sqrt{n}\chi_s)=\mathrm{superpoly}(\lambda)$, which aligns with the convention in the $\NIPFE$ scheme as in \cite{AP20}.
     \end{itemize}
     
     We summarize with the following instantiation
     \begin{corollary}[$(B_0,\chi_s)$-$\maev$ for Subset Policies in the Random Oracle Model]\label{con1par}
    Let $\lambda$ be a security parameter.  Assuming the polynomial hardness of \LWE and of the $\EVIPFE$ (defined in Section~\ref{Evasiveipfe}), both holding under a sub-exponential modulus-to-noise ratio, there exists a statically secure $(B_0,\chi_s)$-$\maev$ scheme for subset policies in the random oracle model, with $B_0/\chi_s=\superpoly(\lambda)$.
\end{corollary}

\section{$\manipfe$ Scheme from $\INDIPFE$ Assumption (in the Random Oracle Model)}\label{sec:7}
We observe that the $\maev$ construction presented in Construction \ref{con1} also naturally satisfies the syntax of a $(B_0,B_1)$-$\manipfe$ scheme over $\Z_q^n$ for subset policies. Accordingly, we reinterpret the scheme in Construction \ref{con1} as an $\manipfe$ scheme, and denote it by $$\Pi_{\manipfe}=(\globalset,\authset,\keygen,\enc,\dec),$$ where all algorithms are identical to those defined in Construction \ref{con1}.

In the following, we present the approximate correctness and static security properties of this scheme.

\subsection{Correctness}
\begin{theorem}[Correctness]\label{correct11}
    Let $\chi_0=\sqrt{n\log q}\cdot \omega(\sqrt{\log n})$ be a polynomial such that Lemma \ref{preimage} holds. Suppose that the lattice parameters $n,q,\chi$ are such that $\chi\geq \chi_0(n,q)$. Then the scheme $\Pi_{\manipfe}$ in Construction \ref{con1} is \emph{correct} as a $(B_0,B_1)$-$\manipfe$ scheme, where the parameter $B_0=\sqrt{\lambda}\chi m+\lambda \chi\chi' m'+\lambda \chi^2 m\ell.$
\end{theorem}

\begin{proof}
    The proof is essentially the same as Theorem \ref{correct1}.
\end{proof}

\subsection{Security}
\begin{theorem}[Static Security]\label{security2}
     Let $\chi_0(n,q)=\sqrt{n\log q}\cdot\omega(\sqrt{\log n})$ be a polynomial such that Lemma \ref{preimage} holds. Let $Q_0$ be an upper bound on the number of secret-key queries submitted by the adversary $\mathcal{A}$. Suppose that the following conditions hold:
\begin{itemize}
\item $m'>6n\log q$.
\item Let $\chi'$ be an error distribution parameter such that $\chi'=\Omega(\sqrt{n\log q})$.
\item Let $\chi$ be an error distribution parameter such that $\chi\geq \lambda^{\omega(1)}\cdot\max\{B_1,\sqrt{\lambda}\chi_s,\chi_0\}$, where $\chi_s$ is an noise parameter such that $\LWE_{n,m_1,q,\chi_s}$ assumption holds for some $m_1=\poly(m,m',Q_0)$.
 \item The assumption $\INDIPFE_{n,m,m',q,\chi,\chi}$ holds.
\end{itemize}
Then, Construction \ref{con1} $\Pi_{\manipfe}$ is \emph{statically secure} as a $(B_0,B_1)$-$\manipfe$ scheme.
\end{theorem}

\begin{proof}
We prove the static security of $\Pi_{\manipfe}$ under the standard static security model by defining two experiments $\H_0$ and $\H_1$, corresponding to challenge bits $0$ and $1$, respectively. The goal is to show that no efficient adversary can distinguish between these two games with non-negligible advantage. We begin by giving the definition of the structure of Game $\H_b$ for $b\in \{0,1\}$.

\iitem{Game $\H_b$} This experiment corresponds to the real static security game in which the challenger encrypts the plaintext $\vecu_{b}$ exactly as specified in Construction \ref{con1}, where $b\in\{0,1\}$. The detailed structure of the experiment is as follows. 

At the beginning of the experiment, the adversary $\mathcal{A}$ specifies the following parameters:
\begin{itemize}
    \item A set of corrupt authorities $\mathcal{C}\subseteq\AU$ along with their public keys: $\pk_{\aid}=(\matA_{\aid},\matB_{\aid},\matP_{\aid})$ for each $\aid\in \mathcal{C}$.
    \item A set of non-corrupt authorities $\mathcal{N}\subseteq \AU$, satisfying $\mathcal{N}\cap\mathcal{C}=\varnothing$.
    \item A pair of challenge messages $\vecu_0,\vecu_1\in\Z_q^n$ along with a designated authority set $A^*\subseteq \mathcal{C}\cup\mathcal{N}$ satisfying  $(A^*\cap\mathcal{C})\subsetneqq A^*$. 
    \item A set of secret-key queries $\mathcal{Q}=\{(\gid,A,\vecv)\}$, where $A\subseteq \mathcal{N}$. For each $(\gid,A,\vecv)\in\mathcal{Q}$, exactly one of the following conditions must hold:
\begin{itemize}
    \item \textbf{Type I secret-key queries}: $(A\cup \mathcal{C})\cap A^*\subsetneqq A^*$;
    \item \textbf{Type II secret-key queries}: $(A\cup \mathcal{C})\cap A^*= A^*$ and $\|(\vecu_0-\vecu_1)^\top\vecv\|\leq B_1$.
\end{itemize}
\end{itemize}
To simulate the random oracle, the challenger initializes an empty table $\T:\GID\times\Z_q^n\rightarrow\Z_q^{m'}$. This table will be used to store and  consistently respond to all queries made to the random oracle during the experiment. The challenger then processes the adversary's queries as follows.
\begin{itemize}
    \item \textbf{Public keys for non-corrupt authorities:} For each non-corrupt authority $\aid\in\mathcal{N}$, the challenger samples $(\matA_{\aid},\td_{\aid})\leftarrow \trapgen(1^n,1^m,q),\matB_{\aid}\xleftarrow{\$}\Z_q^{n\times m'}$, and $\matP_{\aid}\rand\Z_q^{n\times m}$. The public key associated with $\aid$ is then set as $\pk_{\aid}=(\matA_{\aid},\matB_{\aid},\matP_{\aid})$.
    \item \textbf{Secret-key queries:} For each secret-key query $(\gid,A,\vecv)\in\mathcal{Q}$, the challenger first computes $\vecr_{\gid,\vecv}\leftarrow\H(\gid,\vecv)$, then samples $\veck_{\aid,\gid,\vecv}\leftarrow\samplepre(\matA_{\aid},\td_{\aid},\matP_\aid\matG^{-1}(\vecv)+\matB_\aid\vecr_{\gid,\vecv},\chi)$ for each $\aid\in A$. The resulting secret key is set as $\sk_{\aid,\gid,\vecv}=\veck_{\aid,\gid,\vecv}$.
    \item \textbf{Challenge ciphertext:} The challenger first samples $\vecs_{\aid}\rand\Z_{q}^n,\vece_{1,\aid}\leftarrow D_{\Z,\chi}^m$ for each $\aid\in A^*$. Then it samples $\vece_2\leftarrow D_{\Z,\chi}^{m'}$, and $\vece_3\leftarrow D_{\Z,\chi}^{m}$. The challenge ciphertext is constructed as $\ct=\left(\{\vecc_{1,\aid}^\top\}_{\aid\in A^*},\vecc_2^\top,\vecc_3^\top\right)$, where    $$\vecc_{1,\aid}^\top=\vecs_{\aid}^\top\matA_{\aid}+\vece_{1,\aid}^\top,\vecc_{2}^\top=\sum_{\aid\in A^*}\vecs_{\aid}^\top\matB_{\aid}+\vece_2^\top,\vecc_{3}^\top=\sum_{\aid\in A^*}\vecs_{\aid}^\top\matP_{\aid}+\vece_3^\top+\vecu_b^\top\matG.$$
    \item \textbf{Random oracle queries:} Upon receiving a query $(\gid,\vecv)\in\GID\times\Z_q^n$, the challenger first checks whether the input $(\gid,\vecv)$ has been queried before, either during previous random oracle queries by the adversary or during processing the secret-key queries. If it has, the challenger retrieves and returns the stored value $\vecr_{\gid,\vecv}$ from the table $\T$. If the input is new, the challenger samples $\vecr_{\gid,\vecv}\leftarrow D_{\Z,\chi'}^{m'}$, stores the mapping $(\gid,\vecv)\mapsto \vecr_{\gid,\vecv}$ to the table $\T$, and then responds with $\vecr_{\gid,\vecv}$.
\end{itemize}
At the end of the game, the adversary outputs a bit $b'\in\{0,1\}$ as its guess for $b$, which is taken as the output of the experiment. 
\\

To establish the static security of Construction \ref{con1} as a $(B_0,B_1)$-$\manipfe$ scheme, it suffices to show that games $\H_0$ and $\H_1$ are computationally indistinguishable. Suppose, for contradiction, that there exists an \emph{efficient} adversary $\mathcal{A}$ that can distinguish between $\H_0$ and $\H_{1}$ with non-negligible advantage. We will show that such an adversary can be used to break the $\INDIPFE$ assumption, thereby contradicting its assumed hardness.

We construct a sampling algorithm $\mathcal{S}_{\mathcal{A}}$ that uses $\mathcal{A}$ as a subroutine for the $\INDIPFE$ assumption. On input the global parameter $\gp=(1^\lambda,q,1^m,1^{m'},1^\chi,1^\chi)$, the sampling algorithm $\mathcal{S}_{\mathcal{A}}(\gp)$ proceeds as follows:
\begin{enumerate}
\item Let $\kappa=\kappa(\lambda)$ be an upper bound on the number of random bits used by the adversary $\mathcal{A}$. Sample $\vecr\rand\{0,1\}^{\kappa}$ and run $\mathcal{A}(1^\lambda;\vecr)$.

\item The adversary $\mathcal{A}$ outputs a set of corrupt authorities $\mathcal{C}$, a set of non-corrupt authorities $\mathcal{N}$, a set of secret-key queries $\mathcal{Q}=\{(\gid,A,\vecv)\}$, a pair of challenge messages $\vecu_0,\vecu_1\in\Z_q^n$ and a set of authorities $A^*$ associated with the challenge ciphertext. These outputs correspond to the queries submitted by the adversary in the static security game for
$\manipfe$.

\item Let $\ell=|A^*\cap\mathcal{N}|$ and denote the set $A^*\cap\mathcal{N}=\{\aid_1^*,\ldots,\aid_{\ell}^*\}$.
			
\item Partition the set $\mathcal{Q}=\mathcal{Q}_I\cup\mathcal{Q}_{II}$, where $\mathcal{Q}_I,\mathcal{Q}_{II}$ contain all secret-key queries of Type I and Type II, respectively. Let $|\mathcal{Q}_{I}|=Q$, $|\mathcal{Q}_{II}|=Q'$ and write $$\mathcal{Q}_I=\{(\gid_1,A_1,\vecv_1),\ldots,(\gid_{Q},A_{Q},\vecv_{Q})\},\quad\mathcal{Q}_{II}=\{(\gid_1',A_1',\vecv_1'),\ldots,(\gid_{Q'}',\A_{Q'}',\vecv_{Q'}')\}.$$

\item For each $i\in [\ell]$, let $N_i\in[Q]$ denote the number of Type I secret-key queries in which the challenge authority $\aid_{i}^*$ appears. Suppose that authority $\aid_i^*$ is contained in the set $A_{j_1^{(i)}},\ldots, A_{j_{N_i}^{(i)}}$ for some indices $j_1^{(i)},\ldots,j_{N_i}^{(i)}\in [Q]$, listed in increasing order. Define the mapping $\rho_i:[N_i]\rightarrow [Q]$ by setting $\rho_{i}(k)=j_k^{(i)}$. That is,  $\aid_{i}^*$ appears exactly in the set $A_{\rho_i(1)},\ldots, A_{\rho_i(N_i)}$ with the indices ordered increasingly.

\item For each $i\in [Q]$, sample $\vecr_i\leftarrow D_{\Z,\chi'}^{m'}$. For each $j\in [Q']$, sample $\vecr_{j}'\leftarrow D_{\Z,\chi'}^{m'}$.

\item Finally, $\mathcal{S}_{\mathcal{A}}$ outputs
			\begin{align*}
			& 1^{\ell},\{1^{N_i+Q'}\}_{i\in [\ell]}; \\
			& (\vecr_{\rho_1(1)},\vecv_{\rho_1(1)}),\ldots,(\vecr_{\rho_1(N_1)},\vecv_{\rho_{1}(N_1)}),(\vecr_1',\vecv_1'),\ldots,(\vecr_{Q'}',\vecv_{Q'}');                   \\
				       & (\vecr_{\rho_2(1)},\vecv_{\rho_2(1)}),\ldots,(\vecr_{\rho_2(N_2)},\vecv_{\rho_{2}(N_2)}),(\vecr_1',\vecv_1'),\ldots,(\vecr_{Q'}',\vecv_{Q'}');                   \\
				       & \cdots      \\
				       & (\vecr_{\rho_\ell(1)},\vecv_{\rho_\ell(1)}),\ldots,(\vecr_{\rho_\ell(N_\ell)},\vecv_{\rho_{\ell}(N_\ell)}),(\vecr_1',\vecv_1'),\ldots,(\vecr_{Q'}',\vecv_{Q'}'); \\
				       & \vecu_0,\vecu_1.
			      \end{align*}
		\end{enumerate}

We remark that the sampling algorithm $\mathcal{S}_{\mathcal{A}}$ defined above does not explicitly guarantee that each tuple $(\vecr_i,\vecv_i)$ or $(\vecr_j',\vecv_j')$ is non-zero, as required in Assumption \ref{INDIPFE}. However, since each $\vecr_i,\vecr_j'$ component in the tuple is sampled from discrete Gaussian distribution $D_{\Z,\chi'}^{m'}$, the probability that any of them equals the all-zero vector is \emph{negligible}. Specifically, by a standard tail bound argument, we can upper bound the probability that a sample equals the all-zero vector as follows: \begin{align*}
\Pr[\vecr=\veczero_{m'}:\vecr\leftarrow D_{\Z,\chi'}^{m'}]&=\frac{1}{\sum_{\vecx\in\Z^{m'}}\mathrm{e}^{-\pi\|\vecx\|^2/\chi'^2}}   =\left(\frac{1}{\sum_{x\in \Z}\mathrm{e}^{-\pi|x|^2/\chi'^2}}\right)^{m'}\\
   & <\left(\frac{1}{1+2\mathrm{e}^{-\pi/\chi'^2}}\right)^{m'}\leq (1-c\mathrm{e}^{-\pi/\chi'^2})^{m'}\leq \exp\left(-c\cdot m'\mathrm{e}^{-\pi/\chi'^2}\right),
 \end{align*} for some constant $c>0$. When $\chi'$ is taken to be polynomial in $m'$, the probability is $\mathrm{e}^{-\Omega(m')}$, which is negligible in $m'$. Therefore, we may safely ignore the negligible probability that the sampling algorithm outputs an all-zero tuple $(\vecr_i,\vecv_i)$ or $(\vecr_{j}',\vecv_j')$.
		
We now claim that for all efficient distinguishers $\mathcal{D}$, the advantage $\Adv_{\mathcal{S}_{\mathcal{A}},\mathcal{D}}^\pre(\lambda)$ (cf. Assumption \ref{INDIPFE}) for sampling algorithm $\mathcal{S}_{\mathcal{A}}$ is negligible in $\lambda$. More precisely,
	
\begin{claim}\label{indpre}
Let $Q_0$ be an upper bound on the number of secret-key queries submitted by the adversary $\mathcal{A}$. Suppose that the following conditions hold:
\begin{itemize}
\item $m'>6n\log q$.
\item Let $\chi'$ be an error distribution such that $\chi'=\Omega(\sqrt{n\log q})$.
\item Let $\chi$ be an error distribution parameter such that $\chi\geq \lambda^{\omega(1)}\cdot\max\{B_1,\sqrt{\lambda}\chi_s\}$, where $\chi_s$ is an error parameter such that $\LWE_{n,m_1,q,\chi_s}$ assumption holds for some $m_1=\poly(m,m',Q_0)$.
\end{itemize}
Then, for every efficient distinguisher $\mathcal{D}$, there exists a negligible function $\negl(\cdot)$ such that for all $\lambda\in\N$, we have $\Adv_{\mathcal{S}_{\mathcal{A}},\mathcal{D}}^\pre(\lambda)=\negl(\lambda)$, where $\Adv_{\mathcal{S}_{\mathcal{A}},\mathcal{D}}^\pre$ denotes the advantage of the distinguisher $\mathcal{D}$ in the $\INDIPFE$ assumption described in Assumption \ref{INDIPFE}.
\end{claim}

Assuming Claim~\ref{evpre}, we proceed with the proof of Theorem~\ref{security2}. For clarity, the proof of Claim~\ref{evpre} is postponed to the end of this section.

\noindent\emph{Proof of Theorem \ref{security2} (Continued).}
To complete the proof, we show that if there exists an \emph{efficient} adversary $\mathcal{A}$ that can distinguish between $\H_{0}$ and $\H_1$ with non-negligible advantage, then we can construct an \emph{efficient} algorithm $\mathcal{B}$ that breaks the $\INDIPFE$ assumption with respect to the sampling algorithm $\mathcal{S}_{\mathcal{A}}$. The algorithm $\mathcal{B}$ proceeds as follows:
\begin{enumerate}
\item Algorithm $\mathcal{B}$ receives an $\INDIPFE$ challenge $$(1^\lambda,\vecr_{\pub},\{\matA_{i},\vecy_{1,i}^\top\}_{i\in[\ell]},[\matB\mid\matP],\vecy_2^\top,\{\matK_i\}_{i\in[\ell]}),$$ where $\vecr_{\pub}\in\{0,1\}^*,\matA_{i}\in\Z_q^{n\times m}, \vecy_{1,i}\in\Z_q^m,\matK_i\in\Z_q^{m\times (N_i+Q')}$ for each $i\in [\ell]$, $\matB\in\Z_q^{n\ell\times m'},\matP\in\Z_q^{n\ell\times m},\vecy_2\in\Z_q^{m'+m}$. Recall that in the public-coin model, $\vecr_{\pub}$ contains all the randomness used by sampling algorithm $\mathcal{S}_{\mathcal{A}}$.
			
\item For each $i\in [\ell]$, algorithm $\mathcal{B}$ parses the matrix $\matK_i$ as $$\matK_i=[\matK_i^{(1)}\mid \matK_i^{(2)}],$$ where $\matK_{i}^{(1)}\in \Z_q^{m\times N_i}$ and $\matK_{i}^{(2)}\in \Z_q^{m\times Q'}$. Let $\veck_{i,j},\veck_{i,j}'$ denote the $j$-th column vectors of $\matK_i^{(1)}$ and $\matK_{i}^{(2)}$, respectively.
			
\item Algorithm $\mathcal{B}$ runs algorithm $\mathcal{A}$, using the appropriate portion of the public randomness $\vecr_{\pub}$---specifically, the same portion used internally by $\mathcal{S}_{\mathcal{A}}$ to simulate $\mathcal{A}$. The adversary $\mathcal{A}$ outputs the following queries:
\begin{itemize}
\item A set of corrupt authorities $\mathcal{C}\subseteq \mathcal{AU}$,  along with their public keys $\pk_{\aid}=(\matA_{\aid},\matB_{\aid},\matP_{\aid})$ for all $\aid\in\mathcal{C}$.
				      
\item A set of non-corrupt authorities $\mathcal{N}\subseteq \AU$, satisfying $\mathcal{N}\cap\mathcal{C}=\varnothing$.

\item A pair of challenge messages $\vecu_0,\vecu_1\in \Z_q^n$, and a ciphertext authority set $A^*\subseteq \mathcal{C}\cup \mathcal{N}$ such that $(A^*\cap\mathcal{C})\subsetneqq A^*$.

\item A set of secret-key queries $\mathcal{Q}$, where each query $(\gid,A,\vecv)$ is such that $A\subseteq \mathcal{N}$ and satisfies exactly one of the following two conditions:

\begin{itemize}
\item Type I: $(A_i\cup\mathcal{C})\cap A^*\subsetneqq A^*$.
\item Type II: $(A_i\cup \mathcal{C})\cap A^*=A^*$ and $\|(\vecu_0-\vecu_1)^\top\vecv_i\|\leq B_1$.
\end{itemize}
 \end{itemize}

\item Algorithm $\mathcal{B}$ runs sampling algorithm $\mathcal{S}_{\mathcal{A}}(\gp;\vecr_{\pub})$, which outputs \begin{align*}
    & 1^{\ell},\{1^{N_i+Q'}\}_{i\in [\ell]}; \\
 & (\vecr_{\rho_1(1)},\vecv_{\rho_1(1)}),\ldots,(\vecr_{\rho_1(N_1)},\vecv_{\rho_{1}(N_1)}),(\vecr_1',\vecv_1'),\ldots,(\vecr_{Q'}',\vecv_{Q'}');                   \\
 & (\vecr_{\rho_2(1)},\vecv_{\rho_2(1)}),\ldots,(\vecr_{\rho_2(N_2)},\vecv_{\rho_{2}(N_2)}),(\vecr_1',\vecv_1'),\ldots,(\vecr_{Q'}',\vecv_{Q'}');                   \\
 & \cdots; \\
& (\vecr_{\rho_\ell(1)},\vecv_{\rho_\ell(1)}),\ldots,(\vecr_{\rho_\ell(N_\ell)},\vecv_{\rho_{\ell}(N_\ell)}),(\vecr_1',\vecv_1'),\ldots,(\vecr_{Q'}',\vecv_{Q'}').
\end{align*} Since the randomness used to simulate $\mathcal{A}$ in Step 3) comes from a specific portion of $\vecr_{\pub}$, the values produced here---namely, the vectors $\vecv_1,\ldots,\vecv_Q,\vecv_1',\ldots,\vecv_{Q'}'$---align exactly with those generated during the simulation of $\mathcal{A}$ in Step 3). This ensures that algorithm $\mathcal{B}$'s internal simulation of $\mathcal{A}$ is consistent with the actual input instance of the $\INDIPFE$ challenge.

\item From the construction of $\mathcal{S}_{\mathcal{A}}$, we have that $\ell=|A^*\cap \mathcal{N}|$, i.e., the number of non-corrupt authorities appearing in the challenge ciphertext. Let $A^*\cap \mathcal{N}=\{\aid_{1}^*,\ldots,\aid_{\ell}^*\}$. Algorithm $\mathcal{B}$ first sets $\matA_{\aid_i^*}\leftarrow \matA_{i}$ for each $i\in [\ell]$, and parses $[\matB\mid \matP]$ as $$[\matB\mid \matP]=\left[\begin{array}{c|c}
\matB_{\aid_1^*}      & \matP_{\aid_1^*}      \\
\vdots                & \vdots                \\
\matB_{\aid_{\ell}^*} & \matP_{\aid_{\ell}^*}
\end{array}\right]\in\Z_q^{n\ell\times (m'+m)},$$ where $\matB_{\aid_i^*}\in\Z_q^{n\times m'}$ and $\matP_{\aid_i^*}\in\Z_{q}^{n\times m}$ for each $\aid_i^*\in A^*\cap\mathcal{N}$.
			
\item Algorithm $\mathcal{B}$ partitions the set of secret-key queries $\mathcal{Q}=\mathcal{Q}_I\cup\mathcal{Q}_{II}$, where $\mathcal{Q}_I,\mathcal{Q}_{II}$ contain all secret-key queries of Type I and Type II, respectively. Let $|\mathcal{Q}_{I}|=Q$, $|\mathcal{Q}_{II}|=Q'$, and denote $$\mathcal{Q}_I=\{(\gid_1,A_1,\vecv_1),\ldots,(\gid_{Q},A_{Q},\vecv_{Q})\},\quad \mathcal{Q}_{II}=\{(\gid_1',A_1',\vecv_1'),\ldots,(\gid_{Q'}',A_{Q'}',\vecv_{Q'}')\}.$$ For each $i\in [Q]$, it partitions $A_i=A_{i,\chal}\cup \bar{A}_{i,\chal}$, where $A_{i,\chal}$ consists of the authorities in $A_i$ that appear in the ciphertext, i.e. $A_{i,\chal}=A_i\cap A^*$.

\item Algorithm $\mathcal{B}$ initializes an empty table $\T:\GID\times\Z_q^n\rightarrow\Z_q^{m'}$. This table will be used to store and  consistently respond to all queries made to the random oracle during the experiment.

\item The algorithm $\mathcal{B}$ responds to the queries as follows:
\begin{itemize}
\item \textbf{Public keys for non-corrupt authorities}:
\begin{itemize}
	\item For each $\aid_i^*\in A^*\cap\mathcal{N}$, algorithm $\mathcal{B}$ sets $\pk_{\aid_i^*}\leftarrow(\matA_{\aid_i^*},\matB_{\aid_i^*},\matP_{\aid_i^*})$.
	\item For authorities $\aid\in\mathcal{N}\setminus A^*$, algorithm $\mathcal{B}$ samples $(\matA_{\aid},\td_{\aid})\leftarrow \trapgen(1^n,1^m,q)$, $\matB_{\aid}\rand \Z_{q}^{n\times m'}$ and $\matP_{\aid}\rand\Z_q^{n\times m}$, and sets the public key $\pk_{\aid}\leftarrow (\matA_{\aid},\matB_{\aid},\matP_{\aid})$.
	 \end{itemize}
	
 \item \textbf{Secret keys}: The algorithm $\mathcal{B}$ responds to each secret-key query depending on its type:
\begin{itemize}
\item \textbf{Type I}: For a Type I secret-key query $(\gid_k,A_k,\vecv_k)$, recall that $A_k$ is partitioned as  $A_k=A_{k,\chal}\cup \bar{A}_{k,\chal}$, where $A_{k,\chal}=A_k\cap A^*$.

\begin{itemize}
\item For each $\aid_{i^*}\in A^*\cap \mathcal{N}$, recall that the number of secret-key queries involving $\aid_{i}^*$ is the parameter $N_i$. Let $\rho(\cdot)$ be the index mapping previously defined in the proof of Lemma \ref{security2}. For each $j\in [N_i]$, set $\sk_{\aid_i^*,\gid_{\rho_i(j)},\vecv_{\rho_i(j)}}\leftarrow\veck_{i,j}$. Then the algorithm $\mathcal{B}$ checks if the table $\T$ has ever recorded the image of $(\gid_{\rho_i(j)},\vecv_{\rho_i(j)})$, if not, store the mapping $(\gid_{\rho_i(j)},\vecv_{\rho_i(j)})\mapsto\vecr_{\rho_i(j)}$ to the table.

\item For each $k\in [Q]$, if $A_{k,\chal}=\varnothing$, then sample $\vecr_{\gid_k,\vecv_k}\leftarrow D_{\Z,\chi'}^{m'}$, and add the mapping $(\gid_{k},\vecv_{k})\mapsto\vecr_{\gid_k,\vecv_k}$ to the table. At this point, the table contains the image of all pairs $(\gid_k,\vecv_k)$ for each $k\in [Q]$.

\item For each $k\in [Q]$, for each $\aid\in \bar{A}_{k,\chal}$, compute $$\sk_{\aid,\gid_k,\vecv_k}\leftarrow\samplepre(\matA_{\aid},\td_{\aid},\matP_{\aid}\matG^{-1}(\vecv_k)+\matB_{\aid}\H(\gid_k,\vecv_k)),$$ where the value $\H(\gid_k,\vecv_k)$ is retrieved from the table $\T$.
\end{itemize}

\item \textbf{Type II}: For a Type II query $(\gid_j', A_j', \vecv_j')$,  recall that $A^*\cap\mathcal{N}=\{\aid_{1}^*,\ldots,\aid_{\ell}^{*}\}\subseteq A_{j}'$. 
\begin{itemize}
	\item For each $i\in [\ell],j\in [Q']$, set $\sk_{\aid_{i}^*,\gid_j',\vecv_j'}\leftarrow\veck_{i,j}'\in \Z_{q}^m$. Next, algorithm $\mathcal{B}$ adds the mapping $(\gid_j',\vecv_j')\mapsto \vecr_{j}'$ to the table $\T$.
	
	\item For each $\aid\in A_j'\setminus A^*$, algorithm $\mathcal{B}$ computes  $$\sk_{\aid,\gid_j',\vecv_j'}\leftarrow \samplepre(\matA_{\aid},\td_{\aid},\matP_{\aid}\matG^{-1}(\vecv_j')+\matB_{\aid}\vecr_{j}',\chi)$$ \emph{efficiently}.
\end{itemize}
\end{itemize}

\item \textbf{Challenge ciphertext}: Algorithm $\mathcal{B}$ parses $\vecy_2$ as $\vecy_2^\top=[\hat{\vecy}_2^\top\mid \hat{\vecy}_3^\top]$ where $\hat{\vecy}_2\in\Z_q^{m'},\hat{\vecy}_3\in\Z_q^{m}$. It constructs the ciphertext as follows:
For each $\aid\in A^*\cap\mathcal{C}$, sample $\vecs_{\aid}\rand\Z_q^n$ and $\vece_{1,\aid}\leftarrow D_{\Z,\chi}^m$, then set  $\vecc_{1,\aid}^\top\leftarrow \vecs_{\aid}^\top\matA_{\aid}+\vece_{1,\aid}^\top\in\Z_q^m$. For $\aid_i^*\in A^*\cap\mathcal{N}$, set $\vecc_{1,\aid_i^*}\leftarrow \vecy_{1,i}$. Finally, algorithm $\mathcal{B}$ constructs the ciphertext $$\ct\leftarrow\left(\{\vecc_{1,\aid}^\top\}_{\aid\in A^*},\sum_{\aid\in A^*\cap\mathcal{C}}\vecs_{\aid}^\top\matB_{\aid}+\hat{\vecy}_2^\top,\sum_{\aid\in A^*\cap\mathcal{C}}\vecs_{\aid}^\top\matP_{\aid}+\hat{\vecy}_3^\top\right).$$

 \item \textbf{Random oracle queries:} Upon receiving a query $(\gid,\vecv)\in\GID\times\Z_q^n$, algorithm $\mathcal{B}$ checks whether the input $(\gid,\vecv)$ has been queried before---either during the adversary's direct random oracle queries or implicitly through processing the secret-key queries. If so, algorithm $\mathcal{B}$ retrieves and responds with the stored value $\vecr_{\gid,\vecv}$ from the table $\T$. If the input is new, the challenger samples $\vecr_{\gid,\vecv}\mapsto D_{\Z,\chi}^{m'}$, records the mapping $(\gid,\vecv)\rightarrow \vecr_{\gid,\vecv}$ to the table $\T$, and then replies with $\vecr_{\gid,\vecv}$.
			      \end{itemize}
			\item Finally, algorithm $\mathcal{B}$ outputs whatever algorithm $\mathcal{A}$ outputs.
		\end{enumerate}

		The distributions of the public keys for non-corrupt authorities are exactly the same as those in $\H_{0}$ and $\H_1$, as they are uniformly generated. We now examine the responses to the secret-key queries:
\begin{itemize}
	\item For each $\aid_i^*\in A^*\cap\mathcal{N}$, we have \begin{align*}
        \veck_{i,j}&\leftarrow (\matA_{\aid_i^*})_{\chi}^{-1}(\matP_{\aid_i^*}\matG^{-1}(\vecv_{\rho_{i}(j)})+\matB_{\aid_i^*}\vecr_{\rho_i(j)}),\quad \text{for all } j\in [N_{i}],\\
        \veck_{i,j}'&\leftarrow (\matA_{\aid_i^*})_{\chi}^{-1}(\matP_{\aid_i^*}\matG^{-1}(\vecv_j')+\matB_{\aid_i^*}\vecr_{j}'),\quad \text{for all } j\in [Q'].
    \end{align*} 
    
    This perfectly matches the distribution of $\sk_{\aid_i^*,\gid_{\rho_i(j)},\vecv_{\rho_i(j)}}$ and $\sk_{\aid_i^*,\gid_j',\vecv_j'}$ in the actual game, respectively. 
    
    \item For each $\aid\in A_j\setminus A^*$ for some $j\in [Q]$, the secret key $\sk_{\aid,\gid_j,\vecv_j}$ is generated using $\samplepre$, identical to the procedure in $\H_0$ and $\H_1$. The same argument also applies to $\aid\in A_j'\setminus A^*$ for $j\in [Q']$. 
		\end{itemize}

    Finally, we analyze the distribution of the challenge ciphertext.
	Observe that for each $i\in[\ell]$,  $\vecy_{1,i}^\top=\vecs_{i}^\top\matA_{i}+\vece_{1,i}^\top$ and $\vecy_{2}^\top=\vecs^\top[\matB\mid\matP]+\vece_{2}^\top+[\veczero_{m'}^\top\mid \vecu_{b}^\top\matG]$ for some $\vecs^\top=[\vecs_1^\top\mid\cdots\mid\vecs_{\ell}^\top]\rand\Z_q^{n\ell}$, and we write $\vece_2^\top$ as $[\hat{\vece}_2^\top\mid \hat{\vece}_3^\top]$ where $\hat{\vece}_2\in\Z_q^{m'},\hat{\vece}_3\in\Z_q^{m}$. Then
		\begin{align*}
			 & \vecc_{1,\aid_i^*}^\top=\vecz_{1,i}^\top=\vecs_{i}^\top\matA_{\aid_i^*}+\vece_{1,i}^\top, \text{ for each }\aid_{i}^*\in A^*\cap\mathcal{N}, \\
			 & \hat{\vecy}_{2}^\top=\sum_{\aid_i^*\in A^{*}\cap \mathcal{N}}\vecs_i^\top\matB_{\aid_i^*}+\hat{\vece}_{2}^\top,                                               \\
			 & \hat{\vecy}_3^\top= \sum_{\aid_i^*\in A^{*}\cap \mathcal{N}}\vecs_i^\top\matP_{\aid_i^*}+\hat{\vece}_3^\top+\vecu_{b}^\top\matG.
		\end{align*}
	From the specification of algorithm $\mathcal{B}$, the resulting ciphertext is constructed as:
		\begin{align*}
&\vecc_{1,\aid_{i}^*}^\top=\vecs_{i}^\top\matA_{\aid_i^*}+\vece_{1,i}^\top, \quad \text{ for each }\aid_{i}^*\in A^*\cap\mathcal{N},\\
&\vecc_{1,\aid}^\top=\vecs_{\aid}^\top\matA_{\aid}+\vece_{1,\aid}^\top, \quad \text{ for each }\aid \in A^*\cap\mathcal{C},\\
    &\vecc_{2}^\top=\sum_{\aid\in A^*\cap\mathcal{C}}\vecs_{\aid}^\top\matB_{\aid}+\sum_{\aid_i^*\in A^{*}\cap \mathcal{N}}\vecs_i^\top\matB_{\aid_i^*}+\hat{\vece}_{2}^\top,\\
 &\vecc_{3}^\top=\sum_{\aid\in A^*\cap\mathcal{C}}\vecs_{\aid}^\top\matP_{\aid}+ \sum_{\aid_i^*\in A^{*}\cap \mathcal{N}}\vecs_i^\top\matP_{\aid_i^*}+\hat{\vece}_{3}^\top+ \vecu_b^\top\matG.
\end{align*}
		Therefore, the full ciphertext generated by $\mathcal{B}$ is identically distributed to that in the experiment $\H_b$.

    Therefore, the advantage $\Adv_{\mathcal{S_{\mathcal{A}}},\mathcal{B}}^\post$ in the $\INDIPFE$ assumption with respect to the sampling algorithm $\mathcal{S}_{\mathcal{A}}$ is non-negligible. By the $\INDIPFE$ assumption, this implies the existence of another efficient algorithm $\mathcal{B}'$, such that the advantage $\Adv_{\mathcal{S_{\mathcal{A}}},\mathcal{B}'}^\pre$ is also non-negligible, thereby contradicting the indistinguishability guarantee asserted in Claim \ref{indpre}.
    Hence, under the $\INDIPFE$ assumption, no efficient adversary can distinguish between $\H_0$ and $\H_1$ with non-negligible advantage. This concludes the proof of Theorem~\ref{security2}.
\end{proof}

We now provide the proof of Claim~\ref{indpre}, which was used in the proof of Theorem~\ref{security2}.
\begin{proof}[Proof of Claim~\ref{indpre}]
We prove the claim via a sequence of hybrid games, gradually transitioning from the real $\INDIPFE$ experiment to a game independent of the challenge bit $b$. In the following, we define a series of hybrid experiments $\H_{0,b}^\pre,\H_{1,b}^\pre,\H_{2,b}^\pre,\H_{3,b}^\pre$ for $b\in \{0,1\}$, where $\H_{0,b}^\pre$ corresponds to the original $\INDIPFE$ game, and $\H_{3,b}^{\pre}$ is constructed to be independent of $b$. By showing that each adjacent pair of hybrids is indistinguishable, we can conclude that $\H_{0,0}^\pre\overset{c}{\approx}\H_{0,1}^\pre$, thereby completing the proof.

\iitem{Game $\H_{0,b}^\pre$}
This experiment corresponds to the real game in the $\INDIPFE$ assumption. Specifically, the challenger proceeds as follows:

\begin{enumerate}
\item  Let $\kappa=\kappa(\lambda)$ be an upper bound on the number of random bits used by the adversary $\mathcal{S}_\mathcal{A}$. Sample $\vecr_{\pub}\rand\{0,1\}^\kappa$.

\item Run the sampling algorithm $\mathcal{S}_{\mathcal{A}}(\gp;\vecr_{\pub})$ which outputs \begin{align*}
    & 1^{\ell},\{1^{N_i+Q'}\}_{i\in [\ell]}; \\
 & (\vecr_{\rho_1(1)},\vecv_{\rho_1(1)}),\ldots,(\vecr_{\rho_1(N_1)},\vecv_{\rho_{1}(N_1)}),(\vecr_1',\vecv_1'),\ldots,(\vecr_{Q'}',\vecv_{Q'}');                   \\
 & (\vecr_{\rho_2(1)},\vecv_{\rho_2(1)}),\ldots,(\vecr_{\rho_2(N_2)},\vecv_{\rho_{2}(N_2)}),(\vecr_1',\vecv_1'),\ldots,(\vecr_{Q'}',\vecv_{Q'}');                   \\
 & \cdots \\
& (\vecr_{\rho_\ell(1)},\vecv_{\rho_\ell(1)}),\ldots,(\vecr_{\rho_\ell(N_\ell)},\vecv_{\rho_{\ell}(N_\ell)}),(\vecr_1',\vecv_1'),\ldots,(\vecr_{Q'}',\vecv_{Q'}');\\
&\vecu_0,\vecu_1.
\end{align*}

\item Sample $\matB\rand\Z_q^{n\ell\times m'},\matP\rand\Z_q^{n\ell\times m}$, and parse the matrices $$\matB = \begin{pmatrix} \matB_1 \\ \vdots \\ \matB_\ell \end{pmatrix}, \quad
\matP = \begin{pmatrix} \matP_1 \\ \vdots \\ \matP_\ell \end{pmatrix},
$$ where $\matB_{i}\in\Z_q^{n\times m' },\matP_{i}\in\Z_q^{n\times m}$ for each $i\in [\ell]$.

\item For each $i\in [\ell]$, define:
 \begin{align*}
\matQ_i^{(1)}\leftarrow\left[\matP_{i}\matG^{-1}(\vecv_{\rho_i(1)})+\matB_{i}\vecr_{\rho_i(1)}\mid\cdots\mid\matP_{i}\matG^{-1}(\vecv_{\rho_i(N_i)})+\matB_{i}\vecr_{\rho_i(N_i)}\right]\in \Z_q^{n\times N_i}. \end{align*}

and similarly define
$$\matQ_i^{(2)}=\left[\matP_{i}\matG^{-1}(\vecv_{1}')+\matB_{i} \vecr_{1}'\mid \cdots\mid \matP_{i}\matG^{-1}(\vecv_{Q'}')+\matB_{i} \vecr_{Q'}'\right]\in\Z_q^{n\times Q'}.$$
 
Set $\matQ_i=[\matQ_i^{(1)}\mid \matQ_i^{(2)}]\in\Z_{q}^{n\times (N_i+Q')}$. The construction of each matrix $\matQ_i$ matches that specified in Assumption~\ref{INDIPFE}.

 \item Then the challenger samples $\vecs_1,\ldots,\vecs_\ell\rand\Z_q^n$ and sets $\vecs^\top=\left[\vecs_1^\top\mid\cdots\mid \vecs_{\ell}^\top\right]\in \Z_q^{n\ell}$. For each $i\in [\ell]$, it samples $\vece_{1,i}\leftarrow D_{\Z,\chi}^m,\vece_{3,i}\leftarrow D_{\Z,\chi}^{N_i+Q'}$. Then it samples $\vece_2\leftarrow D_{\Z,\chi}^{m'+m}$.

 \item Next, the challenger samples $(\matA_{1},\td_{1}),\ldots,(\matA_{\ell},\td_{\ell})\leftarrow \trapgen(1^n,1^m,q)$. It computes the following values :
\begin{itemize}
\item  $\vecz_{1,i}^\top\leftarrow \vecs_i^\top\matA_i+\vece_{1,i}^\top\in\Z_q^m$ for each $i\in [\ell]$.

\item  $\vecz_{2}^\top\leftarrow \vecs^\top[\matB\mid\matP]+\vece_2^\top+[\veczero_{m'}^\top\mid \vecu_b^\top\matG]=\left[\sum_{i\in[\ell]}\vecs_i^\top\matB_i \,\Bigg|\,\sum_{i\in[\ell]}\vecs_i^\top\matP_i+\vecu_b^\top\matG\right]+\vece_2^\top\in\Z_q^{m'+m}.$
\item $\vecz_{3,i}^\top\leftarrow \vecs_{i}^\top\matQ_i+\vece_{3,i}^\top\in \Z_q^{N_i+Q'}$ for each $i\in [\ell]$.
\end{itemize}  
The component $\vecz_{3,i}$ and $\vece_{3,i}$ are parsed as $$\vecz_{3,i}^\top=[\vecz_{3,i}^{(1)\top}\mid \vecz_{3,i}^{(2)\top}]\in \Z_q^{N_i+Q'},\quad \vece_{3,i}^{\top}=[\vece_{3,i}^{(1)\top}\mid \vece_{3,i}^{(2)\top}]\in \Z_q^{N_i+Q'},$$ where $\vecz_{3,i}^{(1)},\vece_{3,i}^{(1)}\in\Z_q^{N_i}$ and $\vecz_{3,i}^{(2)},\vece_{3,i}^{(2)}\in\Z_q^{Q'}$. For clarity, define $\vect_{i,j}=\matP_{i}\matG^{-1}(\vecv_{\rho_i(j)})+\matB_{i}\vecr_{\rho_i(j)}\in \Z_q^n$ and $y_{i,j}=\vecs_i^\top\vect_{i,j}\in\Z_q$ for each $i\in [\ell]$ and $j\in[N_i]$. Similarly define $\vect_{i,j}'=\matP_{i}\matG^{-1}(\vecv_{j}')+\matB_{i}\vecr_{j}'\in \Z_q^n$ and $y_{i,j}'=\vecs_i^\top\vect_{i,j}'\in\Z_q$ for each $i\in [\ell]$ and $j\in[Q']$.

In summary, we obtain \begin{align*}
\vecz_{3,i}^{(1)\top} & =\vecs_i^\top\matQ_i^{(1)}+\vece_{3,i}^{(1)\top}                =\left[\vecs_i^\top\vect_{i,1}\mid\cdots\mid\vecs_i^\top\vect_{i,N_i}\right]+\vece_{3,i}^{(1)\top}=\left[y_{i,1}\mid\cdots\mid y_{i,N_i}\right]+\vece_{3,i}^{(1)\top}\in\Z_q^{N_i}, \\
\vecz_{3,i}^{(2)\top} & =\vecs_i^\top\matQ_i^{(2)}+\vece_{3,i}^{(2)\top}=\left[\vecs_i^\top\vect_{i,1}'\mid\cdots\mid\vecs_i^\top\vect_{i,Q'}'\right]+\vece_{3,i}^{(2)\top}=\left[y_{i,1}'\mid\cdots\mid y_{i,Q'}'\right]+\vece_{3,i}^{(2)\top}\in\Z_q^{Q'}.
\end{align*}

\item The challenger sends the tuple $(1^\lambda,\vecr_{\pub},\{(\matA_i,\vecz_{1,i}^\top)\}_{i\in [\ell]},[\matB\mid\matP],\vecz_2^\top,\{\vecz_{3,i}^\top\}_{i\in[\ell]})$ to the distinguisher $\mathcal{D}$.

\item The distinguisher $\mathcal{D}$ outputs a bit $\hat{b}\in\{0,1\}$, which is taken as the output of the experiment.
\end{enumerate}

\iitem{Game $\H_{1,b}^{\pre}$} The experiment is identical to $\H_{0,b}^\pre$, except for how the values $\vecz_{1,i},\vecz_2,\vecz_{3,i}$ for each $i\in [\ell]$ are generated. The changes are as follows:
\begin{itemize}
 \item $\vecz_{1,i}$: For each $i\in [\ell]$, sample $\tilde{\vece}_{1,i}\leftarrow D_{\Z,\chi_s}^m$, and compute $\boxed{\vecz_{1,i}^\top\leftarrow \vecs_i^\top\matA_i+\tilde{\vece}_{1,i}^\top+\vece_{1,i}^\top}$.
						  
\item $\vecz_2$: For each $i\in [\ell]$, sample $\tilde{\vece}_{2,i}\leftarrow D_{\Z,\chi_s}^{m'}$ and $\tilde{\vece}_{2,i}'\leftarrow D_{\Z,\chi_s}^m$, then compute $\tilde{\vecz}_{2,i}^\top\leftarrow\vecs_i^\top\matB_i+\tilde{\vece}_{2,i}^\top$ and $\tilde{\vecz}_{2,i}'^\top\leftarrow\vecs_i^\top\matP_i+\tilde{\vece}_{2,i}'^\top$. Finally set $$\boxed{\vecz_2^\top\leftarrow \left[\sum_{i\in [\ell]}\tilde{\vecz}_{2,i}^\top\,\Bigg|\,\sum_{i\in [\ell]}\tilde{\vecz}_{2,i}'^\top+\vecu_b^\top\matG\right]+\vece_2^\top=\left[\sum_{i\in [\ell]}(\vecs_i^\top\matB_i+\tilde{\vece}_{2,i}^\top)\,\Bigg|\,\sum_{i\in [\ell]}(\vecs_i^\top\matP_i+\tilde{\vece}_{2,i}'^\top)+\vecu_b^\top\matG\right]+\vece_2^\top}.$$
						  
\item $\vecz_{3,i}^{(1)}$: For each $i\in [\ell]$ and $j\in [N_i]$, compute 
\begin{center}
							   \fbox{
						  \parbox{0pt}{\begin{align*}
							  y_{i,j}&\leftarrow (\vecs_i^\top\matP_i+\tilde{\vece}_{2,i}'^\top)\matG^{-1}(\vecv_{\rho_i(j)})+(\vecs_i^\top\matB_i+\tilde{\vece}_{2,i}^\top)\vecr_{\rho_i(j)}\\
							  &=\tilde{\vecz}_{2,i}'^\top\matG^{-1}(\vecv_{\rho_i(j)})+\tilde{\vecz}_{2,i}^\top\vecr_{\rho_i(j)},
						  \end{align*} }} 
						  \end{center}
						  and set $$\vecz_{3,i}^{(1)\top}\leftarrow \left[y_{i,1}\mid\cdots\mid y_{i,N_i}\right]+\vece_{3,i}^{(1)\top}\in\Z_q^{N_i}. $$
						  
						   \item $\vecz_{3,i}^{(2)}$: For each $i\in [\ell]$ and $j\in [Q']$, compute 
						   \begin{center}
							   \fbox{
						  \parbox{0pt}{\begin{align*}
						   y_{i,j}'&\leftarrow (\vecs_i^\top\matP_i+\tilde{\vece}_{2,i}'^\top)\matG^{-1}(\vecv_{j}')+(\vecs_i^\top\matB_i+\tilde{\vece}_{2,i}^\top)\vecr_{j}'\\
							  &=\tilde{\vecz}_{2,i}'^\top\matG^{-1}(\vecv_{j}')+\tilde{\vecz}_{2,i}^\top\vecr_{j}',
						  \end{align*}}}
						   \end{center}
and set $$\vecz_{3,i}^{(2)\top}\leftarrow \left[y_{i,1}'\mid\cdots\mid y_{i,Q'}'\right]+\vece_{3,i}^{(2)\top}\in\Z_q^{Q'}. $$
					  \end{itemize}
					  
\iitem{Game $\H_{2,b}^{\pre}$} The experiment is identical to $\H_{1,b}^\pre$, except for how the values $\vecz_{1,i},\vecz_2,\vecz_{3,i}$ for each $i\in [\ell]$ are generated.
\begin{itemize}
\item $\vecz_{1,i}$: For each $i\in [\ell]$, sample $\boxed{\vecz_{1,i}\rand\Z_q^m}$.
						  
\item $\vecz_2$: For each $i\in [\ell]$, sample $\boxed{\tilde{\vecz}_{2,i}\rand\Z_q^{m'}}$ and $\boxed{\tilde{\vecz}_{2,i}'\rand\Z_q^m}$. Then set $$\vecz_2^\top\leftarrow \left[\sum_{i\in [\ell]}\tilde{\vecz}_{2,i}^\top\,\Bigg|\,\sum_{i\in [\ell]}\tilde{\vecz}_{2,i}'^\top+\vecu_b^\top\matG\right]+\vece_2^\top.$$
						  
\item $\vecz_{3,i}^{(1)}$: For each $i\in [\ell]$ and $j\in [N_i]$, compute 
$$y_{i,j}\leftarrow \tilde{\vecz}_{2,i}'^\top\matG^{-1}(\vecv_{\rho_i(j)})+\tilde{\vecz}_{2,i}^\top\vecr_{\rho_i(j)},$$ then set $$\vecz_{3,i}^{(1)\top}\leftarrow \left[y_{i,1}\mid\cdots\mid y_{i,N_i}\right]+\vece_{3,i}^{(1)\top}\in\Z_q^{N_i}. $$
						  
 \item $\vecz_{3,i}^{(2)}$: For each $i\in [\ell]$ and $j\in [Q']$, compute 
$$y_{i,j}'\leftarrow \tilde{\vecz}_{2,i}'\matG^{-1}(\vecv_{j}')+\tilde{\vecz}_{2,i}^\top\vecr_{j}',$$ then set $$\vecz_{3,i}^{(2)\top}\leftarrow \left[y_{i,1}'\mid\cdots\mid y_{i,Q'}'\right]+\vece_{3,i}^{(2)\top}\in\Z_q^{Q'}. $$
\end{itemize}
						  
\iitem{Game $\H_{3,b}^\pre$} The experiment is identical to $\H_{2,b}^\pre$, except it samples $\boxed{\vecz_{3,1}^{(1)}\rand\Z_q^{N_1}}$.

\begin{lemma}\label{indpre0}
	Suppose that $\chi\geq\lambda^{\omega(1)}\cdot(\sqrt{\lambda}(m+\ell)\chi_s+\lambda m'\chi'\chi_s)$. Then we have $\H_{0,b}^\pre \overset{s}{\approx} \H_{1,b}^\pre$.
\end{lemma}

\begin{proof}
	The claim follows directly from a standard noise smudging argument, as in the proof of Lemma \ref{evpre0}.
\end{proof}

\begin{lemma}\label{indpre1}
	Suppose that the assumption $\LWE_{n,m_1,q,\chi_s}$ holds for $m_1=\poly(m,m')$. We have $\H_{1,b}^\pre\overset{c}{\approx}\H_{2,b}^\pre$ for $b\in \{0,1\}$.
\end{lemma}

\begin{proof}
	The claim follows directly from a standard application of the \LWE indistinguishability argument, as in the proof of Lemma \ref{evpre1}.
\end{proof}

\begin{lemma}\label{indpre2}
	Let $m'>6n\log q$ and $\chi'=\Omega(\sqrt{n\log q})$. Suppose that $\chi_s$ is an error parameter such that $\chi\geq \lambda^{\omega(1)}\cdot\sqrt{\lambda}\chi_s$, and that the $\LWE_{n,m_1,q,\chi_s}$ assumption holds for $m_1=\poly(Q_0)$, where $Q_0$ is the upper bound of the number of secret-key queries submitted by the adversary. We have that $\H_{2,b}^\pre\overset{c}{\approx}\H_{3,b}^\pre$.
    \end{lemma}

\begin{proof}
	This claim follows essentially the same argument as in the proof of Lemma \ref{evpre2}. Specifically, the only syntactic difference between the security model of the $\manipfe$ scheme and that of the $\maev$ scheme lies in whether the Type II secret-key queries are determined solely by the adversary or jointly with the challenger. However, in the proof of $\H_{2,b}^\pre\overset{c}{\approx}\H_{3,b}^\pre$ (and likewise $\H_{2}^\pre\overset{c}{\approx}\H_3^\pre$ in Lemma \ref{evpre2}), the distribution of the components involving Type II secret-key queries remains unchanged throughout the transition. Therefore, the indistinguishability follows directly from the argument in Lemma~\ref{evpre2}.
\end{proof}

\begin{lemma}\label{indpre3}
	Suppose that $\chi>\lambda^{\omega(1)}\cdot B_1$. Then we have $\H_{3,0}^\pre\overset{s}{\approx}\H_{3,1}^\pre$.
\end{lemma}

\begin{proof}
	We prove this by defining an intermediate hybrid experiment $\H_{3,b,\midd}^\pre$ between $\H_{3,0}^\pre$ and $\H_{3,1}^\pre$.
	
	\iitem{Game $\H_{3,b,\midd}^\pre$} The game is identical to $\H_{3,b}^\pre$, except for how the challenger samples $\tilde{\vecz}_{2,1}'$.
	Specifically:
	\begin{itemize}
		\item $\vecz_{2}$: The challenger first samples $\tilde{\tilde{\vecz}}_{2,1}'\rand\Z_q^m$, and sets $\boxed{\tilde{\vecz}_{2,1}'^\top\leftarrow \tilde{\tilde{\vecz}}_{2,1}'-\vecu_b^\top\matG}$. Then it computes $$\boxed{\vecz_2^\top\leftarrow \left[\sum_{i\in [\ell]}\tilde{\vecz}_{2,i}^\top\,\Bigg|\,\sum_{i\in [\ell]}\tilde{\vecz}_{2,i}'^\top+\vecu_b^\top\matG\right]+\vece_2^\top=\left[\sum_{i\in [\ell]}\tilde{\vecz}_{2,i}^\top\,\Bigg|\,\ \tilde{\tilde{\vecz}}_{2,1}'+\sum_{2\leq i\leq \ell}\tilde{\vecz}_{2,i}'^\top\right]+\vece_2^\top}.$$

\item $\vecz_{3,1}^{(2)}$: In particular, we have the affected component $$y_{1,j}'=\tilde{\vecz}_{2,1}'^\top\matG^{-1}(\vecv_{j}')+\tilde{\vecz}_{2,1}^\top\vecr_{j}'=\tilde{\tilde{\vecz}}_{2,1}'^\top\matG^{-1}(\vecv_{j}')+\tilde{\vecz}_{2,1}^\top\vecr_{j}'-\vecu_b^\top\vecv_j'.$$ 
\end{itemize}

\begin{claim}\label{h3b}
	For $b\in \{0,1\}$, the two experiments $\H_{3,b}^\pre$ and $\H_{3,b,\midd}^\pre$ are identical, i.e., $\H_{3,b}^\pre\equiv \H_{3,b,\midd}^\pre$.
	\end{claim}
	
	\begin{proof}
	The only difference between Game $\H_{3,b}^\pre$ and Game $\H_{3,b,\midd}^\pre$ lies in how the challenger samples $\tilde{\vecz}_{2,1}'$. In Game $\H_{3,b,\midd}^\pre$, $\tilde{\vecz}_{2,1}'$ is generated by first sampling an independent uniform vector $\tilde{\tilde{\vecz}}_{2,1}'\leftarrow \Z_q^m$ and then shifted by $-\vecu_b^\top\matG$, the resulting $\tilde{\vecz}_{2,1}'$ is uniform and independent of other components, which matches exactly the distribution of $\tilde{\vecz}_{2,1}'$ in $\H_{3,b}^\pre$.
\end{proof}

\begin{claim}\label{h3bmid}
	Suppose that $\chi\geq \lambda^{\omega(1)}\cdot B_1$. Then we have $\H_{3,0,\midd}^\pre\overset{s}{\approx}\H_{3,1,\midd}^\pre$.
\end{claim}

\begin{proof}
	Since in experiment $\H_{3,b,\midd}^\pre$, the components $\vecz_{1,i},\vecz_2,\vecz_{3,i}\ (i\neq 1)$ are sampled independently of the choice of the challenge bit $b$, it suffices to consider the distribution of $\vecz_{3,1}$ in the two experiments. In $\H_{3,0,\midd}^{\pre}$, $\vecz_{3,1}^{(2)\top}=[y_{1,1}'\mid\cdots\mid y_{1,Q'}']+\vece_{3,1}^{(2)\top}$. Specifically, for each $j\in[Q']$, the $j$-th entry of $\vecz_{3,1}^{(2)}$ is given by \begin{align}\label{eqsub6}
		y_{1,j}'+e_{3,1,j}^{(2)} & =\tilde{\vecz}_{2,1}'^\top\matG^{-1}(\vecv_j')+\tilde{\vecz}_{2,1}^\top\vecr_{j}'+e_{3,1,j}^{(2)}\nonumber \\
        & =(\tilde{\tilde{\vecz}}_{2,1}'^\top-\vecu_0^\top\matG)\matG^{-1}(\vecv_j')+\tilde{\vecz}_{2,1}^\top\vecr_{j}'+e_{3,1,j}^{(2)}\nonumber                           \\
	& =\tilde{\tilde{\vecz}}_{2,1}'^\top\matG^{-1}(\vecv_j')+\tilde{\vecz}_{2,1}^\top\vecr_{j}'+e_{3,1,j}^{(2)}-\vecu_0^\top\vecv_j'\nonumber                           \\
	 & = \tilde{\tilde{\vecz}}_{2,1}'^\top\matG^{-1}(\vecv_j')+\tilde{\vecz}_{2,1}^\top\vecr_{j}'+e_{3,1,j}^{(2)}-(\vecu_1^\top+(\vecu_0-\vecu_1)^\top)\vecv_j'\nonumber \\
	 & \overset{s}{\approx} \tilde{\tilde{\vecz}}_{2,1}'^\top\matG^{-1}(\vecv_j')+\tilde{\vecz}_{2,1}^\top\vecr_{j}'+e_{3,1,j}^{(2)}-\vecu_1^\top\vecv_j'\nonumber       \\
	& =(\tilde{\tilde{\vecz}}_{2,1}'^\top-\vecu_1^\top\matG)\matG^{-1}(\vecv_j')+\tilde{\vecz}_{2,1}^\top\vecr_{j}'+e_{3,1,j}^{(2)}.\end{align}
	The approximate identity holds due to the following reason: for each Type II secret-key query $(\gid_j',A_j',\vecv_j')$, we have $\|(\vecu_0-\vecu_1)^\top\vecv_j'\|\leq B_1$ by the bounded inner product condition.  Since $\chi\geq \lambda^{\omega(1)}\cdot B_1$, Lemma \ref{smudge} implies that the noise term statistically hides the difference, i.e., $$(\vecu_0-\vecu_1)^\top\vecv_j'+e_{3,1,j}^{(2)}\overset{s}{\approx} e_{3,1,j}^{(2)}.$$ Moreover, the expression \eqref{eqsub6} matches exactly the distribution of $y_{1,j}'+e_{3,1,j}^{(2)}$ in experiment $\H_{3,1,\midd}^\pre$. Since the distribution of $\vecz_{3,1}^{(1)}$ remains unchanged across the two experiments, the distributions of $\vecz_{3,1}$ are statistically indistinguishable between $\H_{3,0,\midd}^{\pre}$ and $\H_{3,1,\midd}^{\pre}$. This completes the proof.
\end{proof}

\noindent\emph{Proof of Lemma \ref{indpre3} (Continued).} By Claims \ref{h3b} and \ref{h3bmid}, and by applying a standard hybrid argument, we conclude that $\H_{3,0}^\pre\overset{c}{\approx}\H_{3,1}^\pre$.
\end{proof}

\noindent\emph{Proof of Claim \ref{indpre} (Continued)}. Recall that our goal in Claim \ref{indpre} is to show $\H_{0,0}^\pre\overset{c}{\approx}\H_{0,1}^\pre$. By Lemmas \ref{indpre0}$\sim$\ref{indpre3} and standard hybrid argument, we conclude that no efficient distinguisher can distinguish between experiments $\H_{0,0}^\pre$ and $\H_{0,1}^\pre$ with non-negligible advantage. Specifically, \begin{align*}
		&(1^\lambda,\vecr_{\pub},\{\matA_i,\vecs_i^\top\matA_i+\vece_{1,i}^\top\}_{i\in[\ell]},[\matB\mid\matP],\vecs^\top[\matB\mid\matP]+\vece_2^\top+[\veczero\mid\vecu_0^\top\matG],\{\vecs_i^\top\matQ_i+\vece_3^\top\}_{i\in[\ell]}) \\\overset{c}{\approx}\ &(1^\lambda,\vecr_{\pub},\{\matA_i,\vecs_i^\top\matA_i+\vece_{1,i}^\top\}_{i\in[\ell]},[\matB\mid\matP],\vecs^\top[\matB\mid\matP]+\vece_2^\top+[\veczero\mid\vecu_1^\top\matG],\{\vecs_i^\top\matQ_i+\vece_3^\top\}_{i\in[\ell]}).
			\end{align*} 
\end{proof}

\subsection{Parameters}\label{con2par}
Let $\lambda$ be the security parameter.
\begin{itemize}
    \item We set the lattice dimension $n=\lambda^{1/\epsilon}$ and the modulus $q=2^{\tilde{O}(n^{\epsilon})}=2^{\tilde{O}(\lambda)}$ for some constant $\epsilon>0$, where $\tilde{O}(\cdot)$ suppresses constant and logarithmic factors. Then $m=n\lceil\log q\rceil=\tilde{O}(\lambda^{1+1/\epsilon})$.
    \item We set $L=2^\lambda$ to support ciphertext policies of arbitrary polynomial size, i.e., $\ell=\poly(\lambda)$.         
    \item For static security (Theorem~\ref{security2}), we set $m'=\tilde{O}(\lambda^{1+1/\epsilon}),Q_0=\poly(\lambda)$ and $\chi'=\Theta(\lambda^{1+1/\epsilon})$. We relies on the $\LWE_{n,m_1,q,\chi_s}$ assumption and the $\INDIPFE_{n,m,m',q,\chi,\chi}$ assumption, where the noise parameter $\chi_s=\chi_s(\lambda)$ is set polynomial-bounded and $\chi$ satisfies the bound $\chi\geq \lambda^{\omega(1)}\cdot\max\{B_1,\sqrt{\lambda}\chi_s\}$. By setting $B_1=\poly(\lambda)$, we can choose $\chi=2^{\tilde{O}(n^{\epsilon})}$. 
    
    \item To ensure correctness (Theorem~\ref{correct11}), we require $\chi > \sqrt{n \log q} \cdot \omega(\sqrt{\log n})$, which is also satisfied by setting $\chi = 2^{\tilde{O}(n^{\epsilon})}$. In addition, the correctness bound requires $B_0=\sqrt{\lambda}\chi m+\lambda \chi\chi' m'+\lambda \chi^2 mL$, which implies $B_0=2^{\tilde{O}(n^{\epsilon})}$.
     \end{itemize}
     
     We summarize with the following instantiation:
     \begin{corollary}[$(B_0,B_1)$-$\manipfe$ for subset policies in the random oracle model]
    Let $\lambda$ be a security parameter.  Assuming polynomial hardness of \LWE and the $\INDIPFE$, both holding under a sub-exponential modulus-to-noise ratio, there exists a statically secure $(B_0,B_1)$-$\manipfe$ scheme for subset policies in the random oracle model with $B_0/B_1=\superpoly(\lambda)$.
\end{corollary}

\section{Noiseless $\MAABIPFE$ Scheme from $\INDIPFE$ Assumption (in the Random Oracle Model)} \label{sec:8}
In this section, we present a construction of a \emph{noiseless} $\MAABIPFE$ scheme for subset policies in the random oracle model. This construction is obtained by applying the modulus-switching
technique to Construction \ref{con1}.

\begin{construction}[$\maipfe$ for subset policies in the random oracle model]\label{con2}
    Let $\lambda$ be the security parameter. Let $n,p$ and $q$ be lattice parameters, and define $m = n\lceil\log q\rceil$. Let $m'$, $\chi$, and $\chi'$ be additional lattice parameters.  Let $\AU=\{0,1\}^{\lambda}$ be the universe of authority identifiers, and $\GID=\{0,1\}^{\lambda}$ be the universe of global user identifiers. Let $\H: \GID\times\Z_p^n\rightarrow\Z_q^{m'}$ be a hash function, modeled as a random oracle which outputs samples drawn from the discrete Gaussian distribution $D_{\Z, \chi'}^{m'}$, as in Construction \ref{con1}. 
    
    A multi-authority attribute-based (noiseless) inner-product functional encryption scheme ($\maipfe$) over $\Z_p^n$ for subset policies consists of a tuple of efficient algorithms $$\Pi_{\maipfe}=(\globalset, \authset,\keygen,\enc,\dec).$$ The algorithms proceed as follows: 

\begin{itemize}
    \item $\globalset(1^\lambda)\rightarrow \gp$: The global setup algorithm takes as input the security parameter $\lambda$, and outputs the global parameters $\gp=(\lambda,n,m,m',p,q,\chi,\chi',\H)$.
    
    \item $\authset(\gp,\aid)\rightarrow (\pk_{\aid},\msk_{\aid})$: The authority setup algorithm takes as input the global parameters $\gp$ and an authority identifier $\aid\in\AU$. It samples $(\matA_{\aid},\td_{\aid})\leftarrow\trapgen(1^n,1^m,q),\matB_{\aid}\rand\Z_q^{n\times m'},\matP_{\aid}\rand\Z_q^{n\times m}$. It outputs a public key $\pk_{\aid}=(\matA_{\aid},\matB_{\aid},\matP_{\aid})$ and a master secret key $\msk_{\aid}=\td_{\aid}$.
    
    \item $\keygen(\gp,\pk_\aid,\msk_{\aid},\gid,\vecv)\rightarrow \sk_{\aid,\gid,\vecv}$: The key generation algorithm takes as input the global parameters $\gp$, the public key $\pk_\aid=(\matA_\aid,\matB_{\aid},\matP_{\aid})$, the authority's master secret key $\msk_\aid=\td_\aid$, the user identifier $\gid\in \GID$, the key vector $\vecv\in\Z_p^n$ (naturally lifted to $\Z_q^n$ when required). It first computes $\vecr\leftarrow \H(\gid,\vecv)$ and then uses the trapdoor $\td_\aid$ for $\matA_\aid$ to sample $\veck\leftarrow\samplepre(\matA_{\aid},\td_{\aid},\matP_\aid\matG^{-1}(\vecv)+\matB_{\aid}\vecr,\chi)$. It outputs a secret key $\sk_{\aid,\gid,\vecv}=\veck$.
    
    \item $\enc(\gp,\{\pk_{\aid}\}_{\aid\in A},\vecu)\rightarrow\ct$: The encryption algorithm takes as input the global parameters $\gp$, a set of public keys $\{\pk_{\aid}\}_{\aid\in A}=\{(\matA_{\aid},\matB_{\aid},\matP_{\aid})\}_{\aid\in A}$ associated with authorities $A\subseteq \AU$, and a plaintext vector $\vecu\in \Z_p^n$. For each $\aid\in A$, it samples $\vecs_{\aid}\rand\Z_q^n, \vece_{1,\aid}\leftarrow D_{\Z,\chi}^{m}$. It also samples $\vece_2\leftarrow D_{\Z,\chi}^{m'}$, and $\vece_3\leftarrow D_{\Z,\chi}^{m}$. It outputs the ciphertext $\ct\in\Z_q^m\times\Z_q^{m'}\times\Z_q^{m}$, where $$\ct=\left(\left\{\vecs_{\aid}^\top\matA_{\aid}+\vece_{1,\aid}^\top\right\}_{\aid\in A},\sum_{\aid\in A}\vecs_{\aid}^\top\matB_{\aid}+\vece_2^{\top},\sum_{\aid\in A}\vecs_{\aid}^{\top}\matP_{\aid}+\vece_3^{\top}+\lceil\vecu^\top\rfloor_{p\rightarrow q}\matG\right).$$
    
    \item $\dec(\gp,\{\sk_{\aid,\gid,\vecv}\}_{\aid\in A},\gid,\vecv,\ct)\rightarrow \Gamma$: The decryption algorithm takes as input the global parameters $\gp$, a collection of secret keys $\sk_{\aid,\gid,\vecv}=\veck_{\aid,\gid,\vecv}$ associated with authorities $\aid\in A$, a user identifier $\gid\in\GID$, a key vector $\vecv\in\Z_p^n$, and a ciphertext $\ct=(\{\vecc_{1,\aid}^\top\}_{\aid\in A},\vecc_2^\top,\vecc_3^\top)$. It computes $\vecr\leftarrow \H(\gid, \vecv)$ and outputs 
$$\Gamma=\lceil\vecc_3^{\top}\matG^{-1}(\vecv)+\vecc_2^\top\vecr-\sum_{\aid\in A}\vecc_{1,\aid}^\top\veck_{\aid,\gid,\vecv}\rfloor_{q\rightarrow p}.$$
\end{itemize}
\end{construction}

\subsection{Correctness}
\begin{theorem}[Correctness]~\label{correct3}
     Let $\chi_0=\sqrt{n\log q}\cdot \omega(\sqrt{\log n})$ be a polynomial such that Lemma \ref{preimage} holds. Suppose that the lattice parameters $n,m,m',p,q,\chi,\chi'$ satisfy the following conditions:
    \begin{itemize}
        \item $\chi\geq \chi_0(n,q)$.
        \item $q>np^2+2pB_0$, where $B_0=\sqrt{\lambda}\chi m+\lambda \chi\chi' m'+\lambda \chi^2 m\ell$.
    \end{itemize} Then the scheme $\Pi_{\maipfe}$ in Construction \ref{con2} is \emph{correct} as a noiseless $\MAABIPFE$ scheme.
\end{theorem}

\begin{proof}
    The proof idea follows the same structure as that of Theorem  \ref{correct1}. 
 Take any plaintext vector $\vecu\in\Z_p^n$, any key vector $\vecv\in\Z_p^n$, an arbitrary set of authorities $\{\aid\}_{\aid\in A}$, and an arbitrary user identifier $\gid\in\GID$. First, sample the global parameters $\gp\leftarrow\globalset(1^{\lambda})$, and generate the authority keys $(\pk_{\aid},\msk_{\aid})\leftarrow\authset(\gp,\aid)$ for each $\aid\in A$. Then, generate the secret keys $\sk_{\aid,\gid,\vecv}=\veck_{\aid,\gid,\vecv}\leftarrow\keygen(\gp,\msk_{\aid},\gid,\vecv)$ for each $\aid\in A$, and obtain the resulting ciphertext $\ct\leftarrow\enc(\gp,\{\pk_{\aid}\}_{\aid\in A},\vecu)$. 

 The decryption algorithm outputs:  $$\Gamma=\left\lceil\lceil\vecu^\top\rfloor_{p\rightarrow q}\cdot \vecv+\vece_3^\top\matG^{-1}(\vecv)+\vece_2^\top\vecr-\sum_{\aid\in A}\vece_{1,\aid}^\top\veck_{\aid,\gid,\vecv}\right\rfloor_{q\rightarrow p}.$$ 
 Suppose that $\lceil\vecu^\top\rfloor_{p\rightarrow q}=\frac{q}{p}\vecu^\top+\delta_\vecu^\top$ with $\|\delta_\vecu\|\leq 1/2$. Substituting into the expression for $\Gamma$, we obtain \begin{align*}
     \Gamma&=\left\lceil\frac{q}{p}\vecu^\top\vecv+\delta_{\vecu}^\top\vecv+\vece_3^\top\matG^{-1}(\vecv)+\vece_2^\top\vecr-\sum_{\aid\in A}\vece_{1,\aid}^\top\veck_{\aid,\gid,\vecv}\right\rfloor_{q\rightarrow p}\\
     &=\frac{p}{q}\left(\frac{q}{p}\vecu^\top\vecv+\delta_{\vecu}^\top\vecv+\vece_3^\top\matG^{-1}(\vecv)+\vece_2^\top\vecr-\sum_{\aid\in A}\vece_{1,\aid}^\top\veck_{\aid,\gid,\vecv}\right)+\delta'\\
     &=\vecu^\top\vecv+\dfrac{p}{q}\left(\delta_{\vecu}^\top\vecv+\vece_3^\top\matG^{-1}(\vecv)+\vece_2^\top\vecr-\sum_{\aid\in A}\vece_{1,\aid}^\top\veck_{\aid,\gid,\vecv}\right)+\delta'\in \Z_p,
 \end{align*}
 for some $|\delta'|\leq 1/2$.
 Since both $\Gamma$ and $\vecu^\top\vecv$ lie in $\Z_p$, we have
$\Gamma=\vecu^\top\vecv$ if and only if $$|e'|=\left|\dfrac{p}{q}\left(\delta_{\vecu}^\top\vecv+\vece_3^\top\matG^{-1}(\vecv)+\vece_2^\top\vecr-\sum_{\aid\in A}\vece_{1,\aid}^\top\veck_{\aid,\gid,\vecv}\right)+\delta'\right|< 1.$$ According to the noise bounds analysis in Theorem \ref{correct1}, we can get the upper bound $$|e'|\leq \frac{1}{2q}np^2+\frac{p}{q}B_0+\frac{1}{2},$$ where $B_0=\sqrt{\lambda}\chi m+\lambda\chi\chi'm'+\lambda\chi^2mL$. Thus, when $q>np^2+2pB_0$, it holds that the total rounding error $|e'|<1$, which ensures that $\Gamma=\vecu^\top\vecv$. Therefore,  $\Pi_{\maipfe}$ is correct as a noiseless $\MAABIPFE$ scheme.
\end{proof}

\subsection{Static Security}
The proof follows essentially the same structure and reduction strategy as that of Theorem~\ref{security2}, and is therefore omitted.
\begin{theorem}[Static Security]~\label{security3}
     Let $\chi_0(n,q)=\sqrt{n\log q}\cdot\omega(\sqrt{\log n})$ be a polynomial such that Lemma \ref{preimage} holds. Let $Q_0$ be an upper bound on the number of secret-key queries submitted by the adversary $\mathcal{A}$. Suppose that the following conditions hold:
\begin{itemize}
\item $m'>6n\log q$.
\item Let $\chi'$ be an error distribution such that $\chi'=\Omega(\sqrt{n\log q})$.
\item Let $\chi$ be an error distribution parameter such that $\chi\geq \max\{\chi_0,\lambda^{\omega(1)}\cdot\sqrt{\lambda}\chi_s\}$, where $\chi_s$ is an error parameter such that $\LWE_{n,m_1,q,\chi_s}$ assumption holds for some $m_1=\poly(m,m',Q_0)$.
\item The assumption $\INDIPFE_{n,m,m',q,\chi,\chi}$ holds.
			\end{itemize}
			Construction \ref{con2} is \emph{statically secure} as a noiseless $\maipfe$ scheme.
\end{theorem}

\subsection{Parameters}
Let $\lambda$ be the security parameter.
\begin{itemize}
    \item We set the lattice dimension $n=\lambda^{1/\epsilon}$ and the modulus $q=2^{\tilde{O}(n^{\epsilon})}=2^{\tilde{O}(\lambda)}$ for some constant $\epsilon>0$, where $\tilde{O}(\cdot)$ suppresses constant and logarithmic factors. Then $m=n\lceil\log q\rceil=\tilde{O}(\lambda^{1+1/\epsilon})$.
    
    \item We set $L=2^\lambda$ to support ciphertext policies of arbitrary polynomial size, i.e., $\ell=\poly(\lambda)$.   
    
    \item For static security (Theorem~\ref{security3}), we set $m'=\tilde{O}(\lambda^{1+1/\epsilon}),Q_0=\poly(\lambda)$ and $\chi'=\Omega(\lambda^{1+1/\epsilon})$. The construction relies on the $\LWE_{n,m_1,q,\chi_s}$ assumption and the $\INDIPFE_{n,m,m',q,\chi,\chi}$ assumption, where the noise parameter $\chi_s=\chi_s(\lambda)$ is polynomial-bounded and $\chi$ must satisfy the bound $\chi\geq \lambda^{\omega(1)}\cdot\sqrt{\lambda}\chi_s$. We can set $\chi=2^{\tilde{O}(n^{\epsilon})}$.
    \item To ensure correctness (Theorem~\ref{correct3}), we require $\chi > \sqrt{n \log q} \cdot \omega(\sqrt{\log n})$, which can also be satisfied by setting $\chi = 2^{\tilde{O}(n^{\epsilon})}$. In addition, the correctness bound requires $q>np^2+2pB_0$ where $B_0=\sqrt{\lambda}\chi m+\lambda \chi\chi' m'+\lambda \chi^2 m\ell=2^{\tilde{O}(n^{\epsilon})}$, and $p = 2^{\tilde{O}(n^{\epsilon})}$ is chosen accordingly to meet this bound.
     \end{itemize}
     
     We summarize with the following instantiation:
     \begin{corollary}[Noiseless $\MAABIPFE$ for Subset Policies in the Random Oracle Model]\label{con3par}
    Let $\lambda$ be a security parameter.  Assuming polynomial hardness of \LWE and the $\INDIPFE$, both holding under a sub-exponential modulus-to-noise ratio, there exists a statically-secure noiseless $\MAABIPFE$ scheme for subset policies in the random oracle model.
\end{corollary}

\bibliographystyle{alpha}
\bibliography{reference}
\begin{appendices}
    \section{Hardness of Evasive $\IPFE$}\label{sec:app}
    \subsection{Evasive $\LWE$ from \cite{WWW22}}
    We recall the evasive $\LWE$ assumption proposed in~\cite{WWW22}.
    \begin{theorem}[Evasive $\LWE$~\cite{WWW22}]\label{wwwev}
    Let $\lambda\in\N$ be a security parameter, and let $n,q,m,K,\chi,\chi'$ be lattice parameters specified by $\lambda$. Denote $\gp=(1^\lambda,q,1^n,1^m,1^K,1^\chi,1^{\chi'})$. Let $\samp$ be a sampling algorithm, which takes as input the global parameter $\gp$, and outputs a matrix $\matC\in\Z_q^{n\times K}$, a set of matrices $\matQ_1\in \Z_q^{n\times k_1},\ldots,\matQ_{\ell}\in\Z_q^{n\times k_{\ell}}$, and auxiliary information $\aux\in\{0,1\}^*$.

     For two adversaries $\mathcal{A}_0$ and $\mathcal{A}_1$, we define their advantage functions as follows:
    \begin{align*}
        \Adv_{\samp,\mathcal{A}_0}^{\pre}(\lambda)&:=\left|\pr\left[\mathcal{A}_0(1^\lambda,\vecr_{\pub},\{\matA_i,\vecz_{1,i}^\top\}_{i\in[\ell]},\matC,\vecz_2^\top,\{\vecz_{3,i}^\top\}_{i\in [\ell]})=1\right]\right.\\
        &-\left.\pr\left[\mathcal{A}_0(1^\lambda,\vecr_{\pub},\{\matA_i,\vecdelta_{1,i}^\top\}_{i\in[\ell]},\matC,\vecdelta_2^\top,\{\vecdelta_{3,i}^\top\}_{i\in 
        [\ell]})=1\right]\right|;\\
        \Adv_{\samp,\mathcal{A}_1}^{\post}(\lambda)&:=\left|\pr\left[\mathcal{A}_1(1^\lambda,\vecr_{\pub},\{\matA_i,\vecz_{1,i}^\top\}_{i\in[\ell]},\matC,\vecz_2^\top,\{\matK_i\}_{i\in [\ell]})=1\right]\right.\\
        &-\left.\pr\left[\mathcal{A}_1(1^\lambda,\vecr_{\pub},\{\matA_i,\vecdelta_{1,i}^\top\}_{i\in[\ell]},\matC,\vecdelta_2^\top,\{\matK_i\}_{i\in [\ell]})=1\right]\right|;
    \end{align*}
    where the parameters are sampled as follows:
    \begin{itemize}
        \item $(\matB,\matQ_1,\ldots,\matQ_{\ell},\aux)\leftarrow\samp(\gp),$
        \item $(\matA_1,\td_1),\ldots,(\matA_{\ell},\td_\ell)\leftarrow\trapgen(1^n,1^m,q),$
        \item $\vecs_1,\ldots,\vecs_{\ell}\rand\Z_q^n, \vecs^\top\leftarrow[\vecs_1^\top\mid\ldots\mid\vecs_{\ell}^\top]\in\Z_q^{n\ell},$
        \item $\vece_{1,i}\leftarrow D_{\Z,\chi}^m,\vece_{3,i}\leftarrow D_{\Z,\chi}^{k_i} \text{ for each } i\in [\ell], \vece_2\leftarrow D_{\Z,\chi}^{K},$
        \item $\vecdelta_{1,i}\rand\Z_q^m,\vecdelta_{3,i}\rand\Z_q^{k_i} \text{ for each } i\in [\ell], \vecdelta_2\rand\Z_q^{K},$
        \item $\vecz_{1,i}^\top\leftarrow \vecs_i^\top\matA_i+\vece_{1,i}^\top,\vecz_{3,i}^\top\leftarrow\vecs_i^\top\matQ_i+\vece_{3,i}^\top\text{ for each } i\in [\ell], \quad\vecz_{2}^\top\leftarrow\vecs^\top\matC+\vece_2^\top$,
        \item $\matK_i\leftarrow\samplepre(\matA_{i},\td_{i},\matQ_i,\chi') \text{ for each } i\in [\ell].$
    \end{itemize}
    We say that the $\evlwe_{n,m,K,q,\chi,\chi'}$ assumption holds, if for all efficient samplers $\mathcal{S}$, the following implication holds: If there exists an efficient adversary $\A_1$ with a non-negligible advantage function $\Adv_{\samp,\mathcal{A}_1}^\post(\lambda)$, then there exists another efficient adversary $\A_0$ with a non-negligible advantage function $\Adv_{\samp,\mathcal{A}_0}^\pre(\lambda)$.
    \end{theorem}

    \subsection{Reduction from $\evlwe$ (\cite{WWW22}) to $\EVIPFE$}
    \begin{theorem}
        Suppose that the assumption $\evlwe_{n,m,K,q,\chi,\chi'}$ holds, where $K=m+m'$, then the assumption $\EVIPFE_{n,m,m',q,\chi,\chi'}$ holds.
    \end{theorem}

    \begin{proof}
        Suppose, for contradiction, that $\EVIPFE_{n,m,m',q,\chi,\chi'}$ does not hold with respect to some sampler $\mathcal{S}=(\mathcal{S}_{\vecv},\mathcal{S}_{\vecu})$. First, we begin by constructing a sampler $\samp$ for the $\evlwe_{n,m,K,q,\chi,\chi'}$ assumption. On input the global parameter $\gp=(1^\lambda,q,1^n,1^m,1^K,1^\chi,1^{\chi'})$, the sampler $\samp$ proceeds as follows.
        \begin{enumerate}
            \item It runs $\mathcal{S}_{\vecv}(\gp')$ where $\gp'=(1^\lambda,q,1^n,1^m,1^{m'},1^\chi,1^{\chi'})$, which outputs $$\begin{array}{l}
          1^{\ell}, \{1^{k_i}\}_{i\in[\ell]};\\
        (\vecr_{1,1},\vecv_{1,1}),\ldots,(\vecr_{1,k_1},\vecv_{1,k_1});\\
            \cdots\\
            (\vecr_{\ell,1},\vecv_{\ell,1}),\ldots,(\vecr_{\ell,k_\ell},\vecv_{\ell,k_\ell}).      \end{array}$$
            \item It samples $\matC\rand\Z_q^{n\ell\times K}$, and parses it as $$\matC=[\matB\mid\matP]=\left[\begin{array}{c|c}
                \matB_1 & \matP_1 \\
                \vdots&\vdots \\
                \matB_{\ell} &\matP_{\ell}
            \end{array}\right],$$ where $\matB\in \Z_q^{n\ell\times m'},\matP\in\Z_q^{n\ell\times m}$ and $\matB_i\in\Z_q^{n\times m'},\matP_i\in\Z_q^{n\times m}$ for each $i\in [\ell]$.

            \item It constructs each $\matQ_i$ as 
            $$\matQ_i\leftarrow [\matB_i\mid\matP_i]\left[\begin{array}{c|c|c}
            \vecr_{i,1} & \cdots &\vecr_{i,k_i}  \\
            \matG^{-1}(\vecv_{i,1}) & \cdots &\matG^{-1}(\vecv_{i,k_i})
        \end{array}\right]\in \Z_q^{n\times k_i}$$
        \item It then outputs $(\matB,\matQ_1,\ldots,\matQ_{\ell},\aux)$, where $\aux$ consists of the randomness used by $\mathcal{S}_{\vecv}$.
        \end{enumerate}
        We claim that the assumption $\evlwe_{n,m,m',q,\chi,\chi'}$ does not hold with respect to the sampler $\samp$. Suppose that there exists an adversary $\mathcal{A}_1$ such that breaks the $\EVIPFE$ assumption. We then construct an adversary $\mathcal{A}_1'$ that breaks the $\evlwe$ assumption, using $\mathcal{A}_1$ as a subroutine. The adversary $\mathcal{A}_1'$ proceeds as follows:
        \begin{enumerate}
            \item Adversary $\mathcal{A}_1'$ begins by receiving a $\EVIPFE$ challenge $(1^\lambda,\vecr_{\pub},\{\matA_i,\vecy_{1,i}^\top\}_{i\in [\ell]},\matC,\vecy_2^\top,\{\matK_i\}_{i\in [\ell]})$, where $\matA_i\in \Z_q^{n\times m}, \vecz_{1,i}\in \Z_q^m,\matK_i\in\Z_q^{m\times k_i}$ for each $i\in [\ell]$, and $\matC\in\Z_q^{n\ell\times K},\vecz_2\in \Z_q^{n\times K}$.
        \item It samples $\vecr_{\pri}\rand\{0,1\}^\kappa$, where $\kappa$ denotes the upper bounds of the random bits used by $\mathcal{S}_{\vecu}$, runs $\vecu\leftarrow \mathcal{S}_{\vecu}(\gp,\vecr_{\pub};\vecr_{\pri})$, and sends the tuple $$(1^\lambda,\vecr_{\pub},\{\matA_{1},\vecy_{1,i}^\top\}_{i\in [\ell]},\matC,\vecy_2^\top+[\veczero_{m'}\mid \vecu^\top\matG],\{\matK_i\}_{i\in [\ell]})$$ to $\mathcal{A}_1$.
        \item Finally, adversary $\mathcal{A}_1'$ outputs whatever $\mathcal{A}_1$ outputs.
        \end{enumerate}

   By the definition of $\mathcal{S}$ and $\samp$ above, adversary $\mathcal{A}_1'$ simulates perfectly the challenger of the $\evlwe$ distinguishing game. The tuple sent to $\mathcal{A}_1$ has the same distribution as the two possible types of $\EVIPFE$ challenge instances, depending on which one is received.

   It remains to show that the parameters provided by $\mathcal{A}_1'$ satisfies the precondition of the $\evlwe$ assumption, i.e., $$(1^\lambda,\vecr_{\pub},\{\matA_i,\vecz_{1,i}^\top\}_{i\in [\ell]},\matC,\vecz_2^\top,\{\vecz_{3,i}^\top\}_{i\in [\ell]})\overset{c}{\approx}(1^\lambda,\vecr_{\pub},\{\matA_i,\vecdelta_{1,i}^\top\}_{i\in [\ell]},\matC,\vecdelta_2^\top,\{\vecdelta_{3,i}^\top\}_{i\in [\ell]}),$$ then the corresponding $\EVIPFE$ precondition naturally holds, i.e., $$(1^\lambda,\vecr_{\pub},\{\matA_i,\vecz_{1,i}^\top\}_{i\in [\ell]},\matC,\vecz_2^\top+\vecu^\top\matG,\{\vecz_{3,i}^\top\}_{i\in [\ell]})\overset{c}{\approx}(1^\lambda,\vecr_{\pub},\{\matA_i,\vecdelta_{1,i}^\top\}_{i\in [\ell]},\matC,\vecdelta_2^\top,\{\vecdelta_{3,i}^\top\}_{i\in [\ell]}),$$ where the parameters $\vecz_{1,i},\vecz_2,\vecz_{3,i},\vecdelta_{1,i},\vecdelta_2,\vecdelta_{3,i}$ are sampled as in Theorem \ref{wwwev}. This equivalence of distributions follows directly from the parameter construction defined in the reduction procedure.
    \end{proof} 
\end{appendices}

\end{document}